\documentclass[acmsmall]{acmart}

\usepackage{booktabs} % For formal tables

\usepackage[ruled, vlined, linesnumbered]{algorithm2e} % For algorithms

\SetAlFnt{\small}
\SetAlCapFnt{\small}
\SetAlCapNameFnt{\small}
\SetAlCapHSkip{0pt}
\IncMargin{-\parindent}

\usepackage{epstopdf}

\usepackage{multirow,multicol,enumitem,epsfig,url,array,makecell,balance,tabularx}
\usepackage{amsmath}
\usepackage{graphicx}
\usepackage[misc,geometry]{ifsym}

\def\done{}
\def\header{\vspace{2.5mm} \noindent}

\newcommand{\E}{{\mathbb E}}

\def\inN{N^{in}}
\def\outN{N^{out}}

\newcommand{\pf}{p_f}
\def\a{\alpha}

\def\eps{\epsilon}
\def\epi{\hat{\pi}}
\newcommand{\vect}[1]{\boldsymbol{#1}}
\def\rese{\pi^\circ}
\def\resi{r}

\def\cblue{\color{black}}
\def\rmax{r_{max}}

\def\ssppr{FORA}
\def\rsum{r_{sum}}
\def\e{\epsilon}

\acmJournal{TODS}
\acmVolume{1}
\acmNumber{1}
\acmArticle{1}
\acmYear{2019}
\acmMonth{8}
\begin{document}

\begin{sloppy}

\title{Efficient Algorithms for Approximate Single-Source Personalized PageRank Queries}
\thanks{Sibo Wang is supported by CUHK Direct Grant No. 4055114, CUHK University Startup Grant No. 4930911 and No. 5501570, and a donation from Tencent. Xiaokui Xiao is supported by MOE, Singapore under grant MOE2015-T2-2-069, and by NUS, Singapore under an SUG. Zhewei Wei is supported by in part by National Natural Science Foundation of China (No. 61832017 and No. 61732014)}
\author{Sibo Wang}
\email{swang@se.cuhk.edu.hk,}
\affiliation{The Chinese University of Hong Kong}
\author{Renchi Yang}
\email{yang0461@ntu.edu.sg, }
\affiliation{Nanyang Technological University}
\author{Runhui Wang}
\email{runhui.wang@rutgers.edu, }
\affiliation{Rutgers University}
\author{Xiaokui Xiao}
\email{xkxiao@nus.edu.sg, }
\affiliation{National University of Singapore}
\author{Zhewei Wei}
\authornote{Corresponding author}
\email{zhewei@ruc.edu.cn, }
\affiliation{Renmin University of China}
\author{Wenqing Lin}
\email{edwlin@tencent.com, }
\affiliation{Tencent}
\author{Yin Yang}
\email{yyang@hkbu.edu.qa, }
\affiliation{Hamad Bin Khalifa University}
\author{Nan Tang}
\email{ntang@hkbu.edu.qa,}
\affiliation{Qatar Computing Research Institute, HBKU }

\begin{abstract}
	
	Given a graph $G$, a source node $s$ and a target node $t$, the \emph{personalized PageRank} (\emph{PPR}) of $t$ with respect to $s$ is the probability that a random walk starting from $s$ terminates at $t$. An important variant of the PPR query is \emph{single-source PPR} (\emph{SSPPR}), which enumerates all nodes in $G$, and returns the top-$k$ nodes with the highest PPR values with respect to a given source $s$. PPR in general and SSPPR in particular have important applications in web search and social networks, e.g., in Twitter's Who-To-Follow recommendation service. However, PPR computation is known to be expensive on large graphs, and resistant to indexing. Consequently, previous solutions either use heuristics, which do not guarantee result quality, or rely on the strong computing power of modern data centers, which is costly.
	
	Motivated by this, we propose effective index-free and index-based algorithms for approximate PPR processing, with rigorous guarantees on result quality. We first present {\ssppr}, an approximate SSPPR solution that combines two existing methods Forward Push (which is fast but does not guarantee quality) and Monte Carlo Random Walk (accurate but slow) in a simple and yet non-trivial way, leading to both high accuracy and efficiency. Further, {\ssppr} includes a simple and effective indexing scheme, as well as a module for top-$k$ selection with high pruning power. Extensive experiments demonstrate that the proposed solutions are orders of magnitude more efficient than their respective competitors.
	Notably, on a billion-edge Twitter dataset, FORA answers a top-500 approximate SSPPR query within 1 second, using a single commodity server.
\end{abstract}


\renewcommand{\shortauthors}{Wang et al.}
\maketitle

%\input{0-Abstract/Abstract.tex}

%!TEX root=../ssppr_kdd17_newformat.tex
\section{Introduction} \label{sec:intro}

\emph{Personalized PageRank} (\emph{PPR}) is a fundamental operation first proposed by Google \cite{page1999pagerank}, a major search engine. Specifically, given a graph $G$ and a pair of nodes $s, t$ in $G$, the PPR value $\pi(s, t)$ is defined as the probability that a random walk starting from $s$ (called the \emph{source node}) terminates at $t$ (the \emph{target node}), which reflects the importance of $t$  with respect to $s$. One particularly useful variant of PPR is the \emph{single-source PPR} (\emph{SSPPR}), which takes as input a source node $s$ and a parameter $k$, and returns the top-$k$ nodes in $G$ with the highest PPR values with respect to $s$. According to a recent paper \cite{gupta2013wtf}, Twitter, a leading microblogging service, applies SSPPR in their Who-To-Follow application, which recommends to a user $s$ (who is a node in the social graph) a number of other users (with high PPR values with respect to $s$) that user $s$ might want to follow. Clearly, such an application computes SSPPR for every user in the social graph on a regular basis. Hence, accelerating PPR computation may lead to improved user experience (e.g., faster response time), as well as reduced operating costs (e.g., lower power consumption in the data center).

Similar to PageRank \cite{page1999pagerank}, PPR computation on a web-scale graph is immensely expensive, which involves extracting eigenvalues of a $n \times n$ matrix, where $n$ is the number of nodes that can reach millions or even billions in a social graph. Meanwhile, unlike PageRank, PPR values cannot be easily materialized: since each pair of source/target nodes lead to a different PPR value, storing all possible PPR values requires $O(n^2)$ space, which is infeasible for large graphs. For these reasons, much previous work focuses on \emph{approximate PPR} computation (defined in Section \ref{sec:problemdef}), which provides a controllable tradeoff between the execution time and result accuracy. Meanwhile, compared to heuristic solutions, approximate PPR provides rigorous guarantees on result quality.

However, even under the approximate PPR definition, SSPPR computation remains a challenging problem, since it requires sifting through all nodes in the graph. To our knowledge, the majority of existing methods (e.g., \cite{lofgren2014fast,lofgren2015personalized,WangTXYL16}) focus on approximate pair-wise (i.e., with given source and target nodes) PPR computations. A naive solution is to compute pair-wise PPR $\pi(s, v)$ for each possible target node $v$, and subsequently applies top-$k$ selection. Clearly, the running time of this approach grows linearly to the number of nodes in the graph, which is costly for large graphs.

Motivated by this, we propose {\em \ssppr} (short for {\bf FO}ward Push and {\bf RA}ndom Walks), an efficient algorithm for approximate SSPPR computation. The basic idea of {\ssppr} is to combine two existing solutions in a simple and yet non-trivial way, which are (i) Forward Push \cite{AndersenBCHMT07}, which can either computes the exact SSPPR results at a high cost, or terminate early but with no guarantee at all on the result quality, and (ii) Monte Carlo \cite{fogaras2005towards}, which samples and executes random walks and provides rigorous guarantees on the accuracy of SSPPR results, but is rather inefficient. In fact, this idea is so effective that even without any indexing, basic {\ssppr} already outperforms its main competitors BiPPR \cite{lofgren2015personalized} and HubPPR \cite{WangTXYL16}. Then, we describe a simple and effective indexing scheme for {\ssppr}, as well as a novel algorithm for top-$k$ selection. Extensive experiments using several real graphs demonstrate that { \ssppr} is more than two orders of magnitude faster than BiPPR, and more than an order of magnitude faster than HubPPR. In particular, on a billion-edge Twitter graph, { \ssppr} answers top-500 SSPPR query within 1 second, using a single commodity server.

%!TEX root=../ssppr_kdd17_newformat.tex

\section{Background} \label{sec:prelim}
\subsection{Problem Definition} \label{sec:problemdef}

Let $G= (V, E)$ be a directed graph. In case the input graph is undirected, we simply convert it to a directed one by treating each edge as two directed edges of opposing directions. Given a source node $s \in V$ and a decay factor $\alpha$, a random walk (or more precisely, random walk with restart \cite{FujiwaraNYSO12}) from $s$ is a traversal of $G$ that starts from $s$ and, at each step, either (i) terminates at the current node with $\alpha$ probability, or (ii) proceeds to a randomly selected out-neighbor of the current node. For any node $v \in V$, the personalized PageRank (PPR) $\pi(s,v)$ of $v$ with respect to $s$ is then the probability that a random walk from $s$ terminates at $v$ \cite{page1999pagerank}.

A \emph{single-source PPR} (\emph{SSPPR}) query takes as input a graph $G$, a source node $s$, and a parameter $k$, and returns the top-$k$ nodes with the highest PPR values with respect to $s$, together with their respective PPR values. This paper focuses on approximate SSPPR processing, and we first define a simpler version of the approximate SSPPR without top-$k$ selection (called approximate whole-graph SSPPR), as follows.

\begin{definition}[Approximate Whole-Graph SSPPR]\label{def:ss}
Given a source node $s$, a threshold $\delta$, an error bound $\eps$, and a failure probability $\pf$, an approximate whole-graph SSPPR query returns an estimated PPR $\epi(s,v)$ for each node $v\in V$, such that for any $\pi(s,v)>\delta$,
\begin{equation}\label{eqn:ss-def}
|\pi(s,v)- \epi(s,v)| \le \eps\cdot \pi(s,v)
\end{equation}
holds with at least $1- \pf$ probability. \done
\end{definition}

The above definition is consistent with existing work, e.g., \cite{lofgren2014fast,lofgren2015personalized,WangTXYL16}. Next we define the approximate top-$k$ SSPPR, as follows.

\begin{definition}[Approximate Top-$k$ SSPPR]\label{def:topk}
Given a source node $s$, a threshold $\delta$, an error bound $\eps$, a failure probability $\pf$, and a positive integer $k$, an approximate top-$k$ SSPPR query returns a sequences of $k$ nodes, $v_1, v_2,\cdots, v_k$, such that with probability $1-\pf$, for any $i\in [1,k]$ with $\pi(s,v^*_i)> \delta$,
\begin{align}
\epi(s,v_i) \ge (1-\eps) \pi(s, v_i)\label{eqn:topk-self} \\
\pi(s,v_i)  \ge (1-\eps) \cdot \pi(s, v^*_i) \label{eqn:topk-k}
\end{align}
hold with at least $1-\pf$ probability, where $v^*_i$ is the node whose actual PPR with respect to $s$ is the $i$-th largest. \done
\end{definition}
Note that Equation~\ref{eqn:topk-self} ensures the accuracy of the estimated PPR values, while Equation \ref{eqn:topk-k} guarantees that the $i$-th result returned has a PPR value close to the $i$-th largest PPR score. This definition is consistent with previous work \cite{WangTXYL16}. Following previous work \cite{lofgren2015personalized,lofgren2014fast,WangTXYL16}, we assume that $\delta = O(1/n)$, where $n$ is the number of nodes in $G$. The intuition is that, we provide approximation guarantees for nodes with above-average PPR values.

In addition, most applications of personalized PageRank concern web graphs and social networks, in which case the underlying input graphs are generally {\em scale-free}. That is, for any $k \ge 1$, the fraction $f(k)$ of nodes in $G$ that have $k$ edges satisfies
\begin{equation}\label{eqn:sf}
f(k) = c\cdot k^{- \gamma},
\end{equation}
where $\gamma$ is a parameter with $2 \le \gamma \le 3$, and $c$ is a constant smaller than $1$. It can be verified that, in a scale-free graph with $2 \le \gamma \le 3$, the average node degree $m/n = O(\log n)$.
We will analyze the asymptotic performance of our algorithm on both general graphs and scale-free graphs.
%Besides the analysis of our algorithms' asymptotic performance on general graphs, we will further exploit this property to derive results on scale-free graphs.
%in the analysis of our algorithms' asymptotic performance.
Table \ref{tbl:notations} lists the frequently-used notations throughout the paper.

\header
{\bf Remark.} Our algorithms can also handle SSPPR queries where the source $s$ is not fixed but sampled from a node distribution. Interested readers are referred to Section~\ref{app:foraextension} for details.

\begin{table}%
%\tbl{Frequently used notations.\label{tbl:notations}}{%
\caption{Frequently used notations.}
\label{tbl:notations}
\begin{tabular}{|p{0.5in}|p{4in}|}
    \hline
    {\bf Notation} &  {\bf Description}\\

    \hline
    $G$=$(V,E)$   & The input graph $G$ with node set $V$ and edge set $E$\\
    \hline
    $n,m$   & The number of nodes and edges in $G$, respectively\\
    \hline
    \hline
    $\outN(v)$ & The set of  out-neighbors of node $v$ \\
    \hline
     $\inN(v)$ & The set of  in-neighbors of node $v$ \\
     \hline
    $\pi(s, t)$   & The exact PPR value of $t$ with respect to $s$\\
    \hline
    $\alpha$   & The probability that a random walk terminates at a step\\
    \hline
    $\delta, \epsilon, \pf$   & Parameters of an approximate PPR query, as in Definitions \ref{def:ss} and \ref{def:topk}\\
    \hline
    $\rmax$ &   The residue threshold for local update \\
    \hline
    $r(s,v)$ & The residue of $v$ during a local update process from $s$ \\
    \hline
    $\rese(s,v)$ & The reserve of $v$ during a local update process from $s$ \\
    \hline
    $\rsum$ & The sum of all nodes' residues during a local update process from $s$\\
    \hline
    $\pi(s,v_k^*)$ & The $k$-th largest PPR value with respect to $s$ \\ \hline
\end{tabular}%}
\end{table}%

\subsection{Main Competitors} \label{sec:existingsolutions}

\header
{\bf Monte-Carlo.} A classic solution for approximate PPR processing is the {\em Monte-Carlo (MC)} approach \cite{fogaras2005towards}. Given a source node $s$, MC generates $\omega$ random walks from $s$, and it records, for each node $v$, the fraction of random walks $f(v)$ that terminate at $v$. It then uses $f(v)$ as an estimation of the PPR $\epi(s, v)$ of $v$ with respect to $s$. According to \cite{fogaras2005towards}, MC satisfies Definition \ref{def:ss} with a sufficiently large number of random walks: $\omega = \Omega\left(\frac{\log{(1/\pf)}}{\eps^2 \delta}\right)$. According to Refs. \cite{lofgren2014fast,lofgren2015personalized,WangTXYL16} as well as our experiments in Section \ref{sec:exp}, MC is rather inefficient.
Specifically, the time complexity of { MC} is $O\left(\frac{\log{(1/\pf)}}{\eps^2 \delta}\right)$. As will be explained later in Section \ref{sec:ana}, when $\delta = O(1/n)$ and the graph is scale-free, in which case $m/n = O(\log n)$, this time complexity is a factor of $1/\eps$ larger than that of {\em \ssppr} even without indexing or top-$k$ pruning.

\header
{\bf BiPPR and HubPPR.} {\em BiPPR} \cite{lofgren2015personalized} and its successor {\em HubPPR} \cite{WangTXYL16} are currently the states of the art for answering \emph{pairwise} PPR queries, in which both the source node $s$ and the target node $t$ are given, and the goal is to approximate the PPR value $\pi(s, t)$ of $t$ with respect to $s$. The main idea of BiPPR is a bi-direction search on the input graph $G$. The forward direction simply samples and executes random walks, akin to MC described above. Unlike MC, however, BiPPR requires a much smaller number of random walks, thanks to additional information provided by the backward search.

The backward search in BiPPR (dubbed as {\em reverse push}) is originally proposed in \cite{AndersenBCHMT07}, and is rather complicated. In a nutshell, the reverse push starts from the target node $t$, and recursively propagate \emph{residue} and \emph{reserve} values along the reverse directions of edges in $G$. Initially, the residue is 1 for node $t$, and 0 for all other nodes. The original reverse push \cite{AndersenBCHMT07} requires complete propagation until the residues of all nodes become very small, which is rather inefficient as pointed out in \cite{lofgren2015personalized}. BiPPR performs the same backward propagations, but terminates early when the residues of all nodes are below a pre-defined threshold. Then, the method performs forward search, i.e., random walks, utilizing the residue and reserve information computed during backward search. The main tricky part in BiPPR is how to set this residue threshold to minimize computation costs, while satisfying Inequality \ref{eqn:ss-def}. Intuitively, if the residue threshold is set too high, then the forward search requires numerous random walks to reach the approximation guarantee; conversely, if the residue threshold is too low, then the cost of backward search dominates. Ref. \cite{lofgren2015personalized} provides a careful analysis, and reports that a residue threshold of $O\left(\eps\cdot\sqrt{\frac{m\cdot \delta}{n\cdot\log{(1/\pf)}}}\right)$ strikes a good balance between forward and backward searches, and achieves a low overall cost for pair-wise PPR computation.

\begin{algorithm}[t]
\caption{Forward Push} \label{alg:lp}
\BlankLine
\KwIn{Graph $G$, source node $s$, probability $\a$, residue threshold $\rmax$ }
\KwOut{$\rese(s,v)$, $r(s,v)$ for all $v\in V$}
$r(s, s) \gets 1; r(s,v) \gets 0$ for all $v \ne s$\;
$\rese(s,v) \gets 0$ for all $v$\;
\While {$\exists v\in V$ such that $r(s, v)/|\outN(v)| > r_{max}$}
{
    \For {each $u \in \mathcal{N}^{out}(v)$}
    {
       $r(s, u) \gets r(s, u) + (1-\alpha) \cdot \frac{ r(s, v)}{|\outN(v)|}$
    }
    $\rese(s,v) \gets \rese(s,v) + \alpha  \cdot r(s, v)$\;
    $r(s, v) \gets 0$\;
}
\end{algorithm}

To extend BiPPR to SSPPR, one simple method is to enumerate all nodes in $G$, and compute the PPR value for each of them with respect to the source node $s$. The problem, however, is that the residue threshold designed in \cite{lofgren2015personalized} is not optimized for SSPPR, leading to poor performance. To explain, observe that applying BiPPR for SSPPR involves one backward search at each node in $G$, but only one single forward search from $s$. Therefore, we improve the performance of BiPPR by tuning down overhead of each backward search at the cost of a less efficient forward search.
%the cost for each backward search should be tuned down since it is executed numerous times.
This optimization turns out to be non-trivial, and we present it in Section \ref{app:ssbippr}. Nevertheless, the properly optimized version of BiPPR still involves high costs since it either {\em (i)}  degrades to the Monte-Carlo approach if the residue threshold is large or {\em (ii)} incurs a large number of backward searches if the residue threshold is small. %Additionally, BiPPR does not address top-$k$ selection, and, thus, must apply a post-processing filtering step for top-$k$ processing similarly as in MC, which exacerbating its inefficiency.

HubPPR \cite{WangTXYL16} is an index structure based on BiPPR that features an improved algorithm for top-$k$ queries. Since HubPPR inherits the deficiencies of the BiPPR, it is not suitable for SSPPR, either. We will demonstrate this in our experiments in Section \ref{sec:exp}.

%adds indexing and efficient top-$k$ selection to accelerate approximate pairwise PPR computation. Similar to BiPPR, HubPPR involves the same forward and backward searches. However, HubPPR is not suitable for SSPPR since it inherits the deficiencies of the BiPPR for processing SSPPR. This is further demonstrated by our experiments in Section \ref{sec:exp}.

\header
{\bf Forward Push.} {\em Forward Push} \cite{AndersenCL06} is an earlier solution that is not as efficient as BiPPR and HubPPR. We describe it in detail here since the proposed solution {\ssppr} uses its components. Specifically, Forward Push can compute the \emph{exact} PPR values at a high cost. It can also be configured to terminate early, but without any guarantee on result quality. Algorithm~\ref{alg:lp} shows the pseudo-code of Forward Push for whole-graph SSPPR processing. It takes as input $G$, a source node $s$, a probability value $\alpha$, and a threshold $\rmax$; its output consists of two values for each node $v$ in $G$: a {\em reserve} $\rese(s, v)$ and a {\em residue} $\resi(s, v)$. The reserve $\rese(s, v)$ is an approximation of $\pi(s, v)$, while the residue $\resi(s, v)$ is a by-product of the algorithm. In the beginning of the algorithm, it sets $\resi(s, s) = 1$ and $\rese(s, s) = 0$, and sets $\resi(s, v) = \rese(s, v) = 0$ for any $v \ne s$ (Lines 1-2 in Algorithm~\ref{alg:lp}). Subsequently, the residue of $s$ is converted into other nodes' reserves and residues in an iterative process (Lines 3-7).

Specifically, in each iteration, the algorithm first identifies every node $v$ with $\frac{\resi(s, v)}{|\outN(v)|} > \rmax$, where $\outN$ denotes the set of out-neighbors of $v$ (Line 3). After that, it {\em propagates} part of $v$'s residue to each $u$ of $v$'s out-neighbors, increasing $u$'s residue by $(1-\alpha)\cdot \frac{\resi(s, v)}{|\outN(v)|}$. Then, it increases $v$'s reserve by $\alpha \cdot \resi(s, v)$, and resets $v$'s residue to $\resi(s, v) = 0$. This iterative process terminates when every node $v$ has $\frac{\resi(s, v)}{|\outN(v)|} \le \rmax$ (Line 3).

Andersen et al.\ \cite{AndersenCL06} show that Algorithm~\ref{alg:lp} runs in $O(1/\rmax)$ time, and that the reserve $\rese(s, v)$ can be regarded as an estimation of $\pi(s, v)$. This estimation, however, does not offer any worst-case assurance in terms of absolute or relative error. As a consequence, Algorithm~\ref{alg:lp} itself is insufficient for addressing the problem formulated in Definitions \ref{def:ss} and \ref{def:topk}.

{\cblue
\header
{\bf TopPPR.} Most recently, {\em TopPPR} \cite{WeiHX0SW18} is proposed to combine the Forward Push, Monte-Carlo, and the backward search to process the top-$k$ queries. The main idea is to use the Filter-Refinement paradigm so as to accelerate the top-$k$ query processing. They first use the Forward Push and Monte-Carlo approach to derive the upper and lower bound of the PPR $\pi(s,v)$ for each target node $v$. It further maintains a set $C$ of candidates that are the potential top-$k$ answers by examining the upper and lower bound of each node. For example, if a node $v$ has an upper bound $\pi(s,v)$ that is larger than that of the $k$-th largest lower bound. Then we know this node will not be in the top-$k$ answer. When the candidate set is sufficiently small, backward search is started from these node and refine the upper bound and lower bound of the candidate nodes adaptively. The algorithm explores the power-law property and achieves a time complexity of $O(\frac{k^{1 \over 4}\cdot n^{3 \over 4}\cdot \log{n}}{\sqrt{gap_\rho)}})$, where  $gap_\rho$ is a value that quantifies the difference between the top-$k$ and non-top-$k$ PPR values, and $\rho$ is a precision parameter to guarantee that at least $\rho$  fraction of the returned nodes are among the true top-$k$ answers. Notice that $gap_\rho$ should be no larger than $\pi(s,v_k^*)$ where $\pi(s,v_k^*)$ is the $k$-th largest PPR with respect to $s$. Therefore, their time complexity can be written as $O(\frac{k^{1 \over 4}\cdot n^{3 \over 4}\cdot \log{n}}{\sqrt{\pi(s,v_k^*))}})$. On general graphs, the time complexity will degrade to $O(\frac{m+n\cdot \log{n}}{\sqrt{\pi(s,v_k^*)}})$.  As we will see in Section \ref{sec:topk-running-time-guarantee} and Section \ref{sec:exp}, our top-$k$ algorithm achieves both better theoretical result and practical performance.

\header
{\bf Comparison with the conference version \cite{Wang17}.} We make the following new contributions over the conference version.
\begin{itemize}
\item For whole-graph SSPPR queries, we revised the time complexity analysis to derive a refined bound (Section \ref{sec:ana}). Then, in Section \ref{sec:opt}, we further present optimization techniques for whole-graph SSPPR queries. With the new optimization technique, our index-free method improves over the solution in \cite{Wang17} by 2x. Our new index-based method improves over the index-based solution in \cite{Wang17} by at least 2x and up to 3x with 2x space consumption, which demonstrates the effectiveness of the new algorithm and a good trade-off of the new algorithms between the space consumption and query efficiency.

\item For top-$k$ SSPPR queries, the solution proposed in \cite{Wang17} has a worst time complexity similar to the whole-graph SSPPR queries. In this paper,  we derive a new top-$k$ algorithm in Section \ref{sec:topk-running-time-guarantee}, whose time complexity depends on the $k$-th largest PPR value, denoted as $\pi(s,v_k^*)$.  When $\pi(s,v_k^*)$ is a constant, we then improve over the whole-graph SSPPR query by $O(1/n)$; when $\pi(s,v_k^*)$ is $O(1/n)$, then the time complexity of the top-$k$ algorithm is identical to that of the whole-graph SSPPR algorithm. Since the $k$-th largest PPR is typically in between $O(1)$ and $O(1/n)$, the new proposed algorithm improves over the solution proposed in \cite{Wang17}, and is shown to outperform the solution in \cite{Wang17} by around $5x$ in our experimental evaluation.

\item In Section \ref{sec:extensions}, we further extend our results to global PageRank. The results show that it outperforms the classic Monte-Carlo approach and the Power-Iteration method. In addition, our new top-$k$ algorithm can be further used to return the top-$k$ nodes with the hightest global PageRank with a time complexity that linearly depends on the inverse of the $k$-th largest global PageRank.

\item In the experimental evaluation, we have added four game social networks from Tencent Games to examine the effectiveness of our algorithms in real applications and 2 large synthetic datasets to examine the scalability of our proposed algorithms. Extensive experiments demonstrate that our solution is also effective on real social networks and is scalable to huge graphs with up to 8.6 billion edges.
\end{itemize}

Table \ref{tbl:theoretical-guarantees} list the time complexity and space consumption of all approximate algorithms to provide $\epsilon$ relative error guarantee for PPR values no smaller than $\delta$ with at least $1-1/n$ probability. As we can see, the proposed FORA/FORA+ achieves the best time complexity for both the whole-graph SSPPR and top-$k$ queries.

\begin{table*} [!t]
\centering
\begin{small}
  \vspace{-4mm}
%\tblcapup
\caption{Comparison of approximate whole-graph SSPPR and top-$k$ algorithms  with $\boldsymbol{1 -  1/n}$ success probability (Ref. Table \ref{tbl:notations} for the definition of $\epsilon, \delta,$ and $\pi(s,v_k^*)$).}\label{tbl:theoretical-guarantees}
\vspace{-3mm}
%\tblcapdown
 \begin{tabular} {|c|c|c|c|} \hline
    \multirow{2}{*}{\bf Algorithm}  &  \multirow{2}{*}{\bf Space Overhead}& \multicolumn{2}{c|}{\bf Query Time}\\ \cline{3-4}
    & &whole-graph & Top-$k $ \\\hline
   \multirow{2}{*}{MC} & \multirow{2}{*}{0}
                                         & \multirow{2}{*}{ $O\left({\frac{\log n}{\delta\cdot \e^2} } \right)$} & $O\left({\frac{\log n }{\pi(s,v_k^*) \e^2}} \right)$ \\
   ~\cite{fogaras2005towards} & & &  (Using our top-$k$ algorithm) \\ \hline
 {BiPPR}  & {0}
                                         & \multicolumn{2}{c|}{$O\left( \frac{1}{\eps}\sqrt{\frac{m n \cdot \log{n}}{\delta}}\right)$} \\
                                    ~\cite{lofgren2015personalized}     & &  \multicolumn{2}{c|}{}\\  \hline

   {HubPPR}  &  $O(n+m)$ & \multicolumn{2}{c|}{$O\left( \frac{1}{\eps}\sqrt{\frac{m n \cdot \log{n}}{\delta}}\right)$}\\
    ~\cite{WangTXYL16}& & \multicolumn{2}{c|}{}\\\hline
   \multirow{4}{*}{ TopPPR}  &  \multirow{4}{*}{0}
                                          & \multirow{4}{*}{N.A.} & $O(\frac{m+n\cdot \log{n}}{\sqrt{\pi(s,v_k^*)}})$  \\
     & & &     (general graphs) \\ \cline{4-4}
   \cite{WeiHX0SW18} & & &     $O(\frac{k^{1 \over 4}\cdot n^{3 \over 4}\cdot \log{n}}{\sqrt{\pi(s,v_k^*)}})$  \\
    & & &     (power-law graphs) \\ \hline
   FORA &0 & $O(\min \{ \frac{\sqrt{m \cdot \log{n}}}{ \eps \cdot \sqrt{\delta}},$ & $O( \min \{ \frac{\sqrt{m  \log{n}}}{ \eps \cdot \sqrt{\pi(s,v_k^*)}},$  \\ \cline{1-2}
   FORA+&{$O(\min \{ n+ \frac{1}{\eps\cdot \sqrt{\delta}}\sqrt{m \log{(1/\pf)}},m\})$}  &  $\frac{\log{n}}{  \eps^2 \cdot \delta}\})$ & $\frac{\log{n}}{\epsilon^2 \cdot \pi(s,v_k^*)}\})$ \\ \hline
 \end{tabular}
\end{small}
%\normalsize
%\tbldown
\vspace{-3mm}
\end{table*}

}

%!TEX root=../ssppr_kdd17_newformat.tex

\section{\ssppr}\label{sec:ssppr}

This section presents the proposed { \ssppr} algorithm. We first describe a simpler version of { \ssppr} for whole-graph SSPPR (Definition \ref{def:ss}) without indexing in Sections \ref{sec:ss-rationale} and  Sections  \ref{sec:ana}. Then, we present the indexing scheme of { \ssppr} in Section \ref{sec:ehi}. We present optimization techniques for whole-graph SSPPR in Section \ref{sec:opt} and top-$k$ query in Section \ref{sec:topk}.

\subsection{Main Idea} \label{sec:ss-rationale}

As reviewed in Section \ref{sec:existingsolutions}, (i) MC is inefficient due to a large number of random walks required to satisfy the approximation guarantee, (ii) BiPPR and HubPPR either degrade to MC, or require fewer forward random walks but still incur high cost due to numerous backward search operations, and (iii) Forward Push with early termination provides no formal guarantee on result quality. The proposed solution { \ssppr} can be understood as a combination of these methods. In particular, {\ssppr} first performs Forward Push with early termination, and subsequently runs random walks. Similar to BiPPR and HubPPR, {\ssppr} utilizes information obtained through Forward Push to significantly cut down the number of required random walks while satisfying the same result quality guarantees. But unlike BiPPR and HubPPR, in {\ssppr} there is a single invocation of Forward Push starting from the source node $s$, while BiPPR and HubPPR invokes numerous backward search operations.

\header
{\bf Novelty.} The main difference between the proposed approach and BiPPR/HubPPR is that the former combines forward push and MC, whereas the latter uses MC with backward propagation. The proposed idea (i.e., using forward instead of backward push) is indeed natural, and highly effective as shown in experiments (Section \ref{sec:exp}). Intuitively, forward push limits the search to nodes in the vicinity of the source node, which is more efficient for SSPPR compared to backward propagation, since the latter involves numerous nodes far from the source, as explained in Sections \ref{sec:existingsolutions} (which explains BiPPR).

Combining forward push with MC, however, is far from straightforward. The devil is in the details. Specifically, in  backward propagation, the residue $r(v, t)$ for a given node $v$ is bounded by $r_{max}$, which is a controllable parameter ($t$ is the destination node). In Forward Push (Algorithm \ref{alg:lp}), the corresponding concept is $r(s, v)$ ($s$ being the source node), which depends on both $r_{max}$ and the out-degree of $v$, which can be as large as $O(n)$ in the worst case. The proposed solution addresses this challenge with a novel mechanism that utilizes $r_{sum}$ (Algorithm 2), whose correctness is rigorously established in this section. The novelty of our method lies in the fact that we realize a natural, effective and yet challenging combination of forward push and MC with a non-trivial algorithmic design, explained below.

\header
{\bf Details.}
Specifically, the reason that Forward Push with early termination fails to obtain any result quality guarantee is that it uses $\rese(s, v)$ to approximate $\pi(s, v)$, and yet, the two values are not guaranteed to be close. To mitigate this deficiency, we aim to utilize the residue $\resi(s, v)$ to improve the accuracy of $\rese(s, v)$. Towards this end, we utilize the following result from \cite{AndersenCL06}:
\begin{equation}\label{eqn:layerppr}
\pi(s,t) = \rese(s,t) + \sum_{v\in V} r(s,v) \cdot \pi(v,t),
\end{equation}
for any $s$, $t$, $v$ in $G$.
Our idea is to derive a rough approximation of $\pi(v, t)$ for each node $v$ (denoted as $\pi'(v, t)$), and then combine it with the reserve of each node to compute an estimation of $\pi(s, t)$:
\begin{equation*} %\label{eqn:layerppr}
\pi(s,t) = \rese(s,t) + \sum_{v\in V} r(s,v) \cdot \pi'(v,t).
\end{equation*}
In particular, we derive $\pi'(v, t)$ by performing a number of random walks from $v$, and set $\pi'(v, t)$ to the fraction of walks that ends at $t$.

It remains to answer two key questions in {\ssppr}: (i) how many random walks do we need for each node $v$? and (ii) how should we set the residue threshold $\rmax$ in Forward Push? It turns out that although the {\ssppr} algorithm itself is simple, deriving the proper values for its parameters is rather challenging, since they must optimize efficiency while satisfying the result quality guarantee. In the following, we first present the complete {\ssppr} and answer question (i); then we answer question (ii) in Section \ref{sec:ana}.

\begin{algorithm}[t]
\caption{{\em \ssppr} for Whole-Graph SSPPR}\label{alg:ss-algo}
\BlankLine
\KwIn{Graph $G$, source node $s$, probability $\a$, threshold $\rmax$, relative error threshold $\epsilon$ }
\KwOut{Estimated PPR $\epi(s,v)$for all $v\in V$}

Invoke Algorithm \ref{alg:lp} with input parameters $G$, $s$, $\a$, and $\rmax$\;
let $r(s,v_i), \rese(s,v_i)$ be the returned residue and reserve of node $v_i$\;
Let $\rsum = \sum_{v_i\in V}r(s,v_i)$ and $\omega=\rsum\cdot\frac{(2\eps/3+2)\cdot \log{(2/\pf)}}{\eps^2\cdot \delta}$\;
Let $\epi(s,v_i)=\rese(s,v_i)$ for all $v_i \in V$\;
\For{$v_i\in V$ with $r(s,v_i)>0$}
{
    Let $\omega_i = \lceil r(s,v_i)\cdot\omega/\rsum\rceil$\;
    Let $a_i = \frac{\resi(s,v_i)}{\rsum}\cdot\frac{\omega}{\omega_i}$\;
    \For{$i=1$ to $\omega_i$}
    {
        Generate a random walk $W$ from $v_i$\;
        Let $t$ be the end point of $W$\;
        $\epi(s,t) += \frac{a_i\cdot \rsum}{\omega}$\;
    }
}
return $\epi(s,v_1),\cdots, \epi(s,v_n)$\;
\end{algorithm}

Algorithm \ref{alg:ss-algo} illustrates the pseudo-code of { \ssppr}. Given $G$, a source node $s$, a probability value $\alpha$, and a residue threshold $\rmax$, { \ssppr} first invokes Algorithm \ref{alg:ss-algo} on $G$ to obtain a reserve $\rese(s, v_i)$ and a residue $\resi(s, v_i)$ for each node $v_i$ (Line 1 in Algorithm \ref{alg:ss-algo}). After that, it computes the total residue of all nodes $\rsum$, based on which it derives a value $\omega$ that will be used to decide the number of random walks required from each node $v_i$ (Line 2). Then, it initializes the PPR estimation of each $v_i$ to be $\epi(s, v_i) = \rese(s, v_i)$, and it proceeds to inspect the nodes whose residues are larger than zero (Line 3-4).

For each $v_i$ of those nodes, it performs $\omega_i$ random walks from $v_i$, where
$$\omega_i = \left\lceil \frac{\resi(s,v_i)}{\rsum}\cdot\omega\right\rceil.$$
If a random walk ends at a node $t$, then { \ssppr} increases $\epi(s, v_i)$ by $\frac{a_i \cdot \rsum}{\omega}$, where
$$a_i = \frac{\resi(s,v_i)}{\rsum}\cdot\frac{\omega}{\omega_i}.$$
After all $v_i$ are processed, the algorithm returns $\epi(s, v_i)$ as the approximated PPR value for $v_i$ (Line 11).

To explain why { \ssppr} can provide accurate results, let us consider the $\omega_i$ random walks that it generates from a node $v_i$. Let $X_j(t)$ be a Bernoulli variable that takes value $1$ if the $j$-th random walk terminates at $t$, and value $0$ otherwise. By definition,
$$\E[X_j] = \pi(v_i, t).$$
Then, based on the definition of $\omega$, $\omega_i$, and $a_i$, we have
\begin{equation}\label{eqn:connection}
\E\left[\frac{\rsum}{\omega}\cdot\sum_{j = 1}^{\omega_i} \left(a_i \cdot X_j\right)\right] = \resi(s, v_i) \cdot \pi(v_i, t).
\end{equation}
Observe that $\frac{\rsum}{\omega}\cdot\sum_{j = 1}^{\omega_i} \left( a_i \cdot X_j\right)$ is exactly the amount of increment that $\epi(s, t)$ receives when { \ssppr} processes $v_i$ (see Lines 7-10 in Algorithm \ref{alg:ss-algo}). We denote this increment as $\psi_i$. It follows that
\begin{equation} \label{eqn:ss-exp}
\E\left[\sum_{i = 1}^{n} \psi_i\right] = \sum_{i=1}^n\resi(s, v_i) \cdot \pi(v_i, t).
\end{equation}
Combining Equations \ref{eqn:layerppr} and \ref{eqn:ss-exp}, we can see that { \ssppr} returns, for each node $v$, an estimated PPR $\epi(s, v)$ whose expectation equals $\pi(s, v)$.
Next, we will show that $\epi(s, v)$ is very close to $\pi(s, v)$ with a high probability. For this purpose, we utilize the following concentration bound:
\begin{theorem}[\cite{ChungL06}]\label{thm:conc}
Let $X_1, \cdots, X_\omega$ be independent random variables with
$$\Pr[X_i = 1]=p_i \;\; \textrm{and} \;\; \Pr[X_i=0] = 1- p_i.$$
Let $X=\frac{1}{\omega}\cdot\sum^\omega_{i=1}a_i X_i$ with $a_i> 0$, and $\nu = \frac{1}{\omega}\sum^\omega_{i=1}a_i^2\cdot p_i$. Then,
\begin{equation*}
\Pr[|X-\E[X]|\ge \phi] \le 2\cdot \exp\left(-\frac{\phi^2\cdot \omega}{2\nu + 2a\phi/3}\right),
\end{equation*}
where $a=\max\{a_1,\cdots, a_\omega\}$. \done
\end{theorem}

To apply Theorem~\ref{thm:conc}, let us consider the $\omega' = \sum_{i=1}^n \omega_i$ random walks generated by { \ssppr}. Let $b_j = a_i$ if the $j$-th random walk starts from $v_i$. Then, we have $\max_{j} b_j = 1$, and $b_j^2 \le b_j$ for any $j$. In addition, let $Y_j(t)$ be the a random variable that equals $1$ if the $j$-th walk terminates at $t$, and $0$ otherwise. Then, by Theorem~\ref{thm:conc} and Equations \ref{eqn:layerppr} and \ref{eqn:ss-exp}, we have the following lemma.

\begin{lemma}\label{lem:approximation}
For any node $t$, given an arbitrary relative error threshold $\epsilon$, an arbitrary absolute threshold $\lambda$ we have that:
\begin{equation}\label{eqn:relative-error}
\Pr[|\pi(s,t) - \epi(s,t)|\ge \eps\cdot \pi(s,t)]\le 2\cdot \exp\left(-\frac{\eps^2\cdot \omega\cdot \pi(s,t)}{\rsum\cdot(2 + 2\eps/3)}\right).
\end{equation}
\begin{equation}\label{eqn:absolute-error}
\Pr[|\pi(s,t) - \epi(s,t)|\ge \lambda]\le 2\cdot \exp\left(-\frac{\lambda^2\cdot \omega}{\rsum\cdot(2\pi(s,t) + 2\lambda/3)}\right).
\end{equation}
\end{lemma}
\begin{proof}
  Firstly, define $Y'= \frac{1}{\omega'}\sum_{j=1}^{\omega'} b_j Y_j(t)$, and $\nu= \frac{1}{\omega'}\sum_{j=1}^{\omega'}b_j^2\E[Y_j(t)]$. Let $a = \max\{b_1,\cdots, b_{\omega'}\}$. By definition, $b_j^2 \le 1$, and hence, $\nu \le \E[Y']$ and $a\le 1$.
  By Theorem \ref{thm:conc}, for any $\phi$, we have that $\Pr[|Y'-\E[Y']|\ge \phi] \le 2\cdot \exp\left(-\frac{\phi^2\cdot \omega'}{2\nu + 2a\phi/3}\right)$. Apply $\nu\le \E[Y']$, we have that:

$$\Pr[|Y'-\E[Y']|\ge \phi] \le 2\cdot \exp\left(-\frac{\phi^2\cdot \omega'}{2\E[Y'] + 2a\phi/3}\right).$$
Observe that $\frac{\omega' \cdot \rsum}{\omega}(\E[Y']-Y')=\pi(s,t)-\hat{\pi}(s,t)$, the above inequality can be rewritten as:
$$\Pr[|\pi(s,t) - \epi(s,t)|\ge \frac{\omega'\cdot \rsum}{\omega}\phi]\le 2\cdot \exp\left(-\frac{\phi^2\cdot \omega'}{2\E[Y'] + 2a\phi/3}\right).$$
Besides, by Equation \ref{eqn:connection}, we have that:
 $\E[Y'] \le \frac{\omega }{\omega'\cdot\rsum}\cdot\pi(s,t)$, it is satisfied that:
$$\Pr[|\pi(s,t) - \epi(s,t)|\ge \frac{\omega'\cdot \rsum}{\omega}\phi]\le 2\cdot \exp\left(-\frac{\phi^2\cdot \omega'}{2\frac{\omega\cdot \pi(s,t)}{\omega'\cdot \rsum} + 2a\phi/3}\right).$$
Let $\phi=\frac{\omega\cdot\eps\cdot \pi(s,t)}{\omega'\cdot \rsum}$, we have:
$$\Pr[|\pi(s,t) - \epi(s,t)|\ge \eps\cdot \pi(s,t)]\le 2\cdot \exp\left(-\frac{\eps^2\cdot \omega\cdot \pi(s,t)}{\rsum\cdot(2 + 2a\cdot\eps/3)}\right).$$
Since $a\le1$, we get that:
$$\Pr[|\pi(s,t) - \epi(s,t)|\ge \eps\cdot \pi(s,t)]\le 2\cdot \exp\left(-\frac{\eps^2\cdot \omega\cdot \pi(s,t)}{\rsum\cdot(2 + 2\eps/3)}\right).$$
By setting $\epsilon = \lambda /\pi(s,t)$, we further have that:
$$\Pr[|\pi(s,t) - \epi(s,t)|\ge \lambda]\le 2\cdot \exp\left(-\frac{\lambda^2\cdot \omega}{\rsum\cdot(2\pi(s,t) + 2\lambda/3)}\right).$$
This finishes the proof.
\end{proof}

\begin{lemma}\label{lem:layerppr-conc}
For any node $t$ with $\pi(s, t) > \delta$, Algorithm \ref{alg:ss-algo} returns an approximated PPR $\epi(s, t)$ that satisfies Equation~\ref{eqn:ss-def} with at least $1-\pf$ probability. \done
\end{lemma}

\begin{proof}

Since $\omega=\rsum\cdot\frac{(2\eps/3+2)\cdot \log{(2/\pf)}}{\eps^2\cdot \delta}$, according to Lemma \ref{lem:approximation}, we have that:
\begin{align*}
  \Pr[|\pi(s,t) - \epi(s,t)|\ge \eps\cdot \pi(s,t)] & \le \exp\left(-\frac{\eps^2\cdot\pi(s,t)}{\rsum\cdot(2 + 2\eps/3)} \cdot  \rsum\cdot\frac{(2\eps/3+2)\cdot \log{(2/\pf)}}{\eps^2\cdot \delta}\right) \\
   & \leq 2\exp\left(-\frac{\pi(s,t)}{\delta}\cdot  \log{(2/\pf)} \right)
\end{align*}
Since $\pi(s,t)> \delta$, we have that:
$$\Pr[|\pi(s,t) - \epi(s,t)|\ge \eps\cdot \pi(s,t)] <2\exp(\log{(2/\pf)})=\pf .$$

Also notice that the target $t$ is arbitrarily chosen, and we can derive this bound for all nodes $t\in V$. Hence, the returned answer for the single-source PPR query satisfies Definition~\ref{def:ss}, which finishes the proof.
\end{proof}

\subsection{Choosing $\rmax$} \label{sec:ana}

Recall from Sections \ref{sec:existingsolutions} and \ref{sec:ss-rationale} that parameter $\rmax$ determines how quickly we can terminate Forward Push. A high value for $\rmax$ leads to low cost for Forward Push (since it can terminate early), but high cost for random walks (since a large number of them are required), and vice versa. Thus, finding the appropriate value of $\rmax$ requires modelling the overall running time of $\ssppr$. Recall that, the Forward Push runs in $O\left(\frac{1}{\rmax}\right)$ time. In addition, the expected time complexity of the random walk phase is $O\left(\rsum\cdot\frac{(2\eps/3+2)\cdot \log{(2/\pf)}}{\eps^2 }\right)$, since each random walk takes $O(1)$ expected time to generate. Observe that
\begin{align*}
 \rsum &= \sum_{v_i\in V}r(s,v_i)\le \sum_{v_i\in V} \rmax \cdot |\outN(v_i)|= m\cdot \rmax.
\end{align*}

Therefore, the expected running time of Algorithm \ref{alg:ss-algo} is
$$\textstyle O\left(\frac{1}{ \rmax}+ m \cdot \rmax \cdot\frac{(2\eps/3+2)\cdot \log{(2/\pf)}}{  \eps^2 \cdot \delta}\right).$$

Using the method of Lagrange multipliers, we can see that the above time complexity is minimized when
\begin{equation}\label{eqn:rmax}
\rmax = \frac{\eps}{\sqrt{m}} \cdot \sqrt{\frac{\delta}{(2\eps/3+2)\cdot \log{(2/\pf)}}}.
\end{equation}
Accordingly, the expected time complexity of Algorithm \ref{alg:ss-algo} becomes
$$\textstyle O\left( \frac{1}{ \eps \cdot \sqrt{\delta}}\sqrt{m \cdot (2\eps/3+2)\cdot \log{(2/\pf)}}\right).$$

However, we also note that $\rsum$ can be bounded by $1$. Therefore, we have to consider two cases.

\begin{itemize}
  \item {\bf Case 1: $ \mathbf{m\cdot \rmax \leq 1}$.} Then, it is easy to verify that $\frac{1}{ \eps \cdot \sqrt{\delta}}\sqrt{m \cdot (2\eps/3+2)\cdot \log{(2/\pf)}} \leq \frac{(2\eps/3+2)\cdot \log{(2/\pf)}}{  \eps^2 \cdot \delta}$. Then the time complexity can be bounded by $O(\frac{1}{ \eps \cdot \sqrt{\delta}}\sqrt{m \cdot (2\eps/3+2)\cdot \log{(2/\pf)}})$.
  \item {\bf Case 2: $\mathbf{ m \cdot \rmax > 1}$.} In this case, $\frac{1}{ \eps \cdot \sqrt{\delta}}\sqrt{m \cdot (2\eps/3+2)\cdot \log{(2/\pf)}} > \frac{(2\eps/3+2)\cdot \log{(2/\pf)}}{  \eps^2 \cdot \delta}$. Therefore, if we set $\rmax$ according to Equation \ref{eqn:rmax}, we will have sub-optimal performance. To remedy this issue, we set $\rmax$ to $\frac{  \eps^2 \cdot \delta}{(2\eps/3+2)\cdot \log{(2/\pf)}}$. Then the time complexity can be bounded by $O(\frac{(2\eps/3+2)\cdot \log{(2/\pf)}}{  \eps^2 \cdot \delta})$.
\end{itemize}

Therefore, combining the above two cases, the time complexity of FORA can be bounded by
$$
\textstyle O\left(\min \{ \frac{1}{ \eps \cdot \sqrt{\delta}}\sqrt{m \cdot \log{(2/\pf)}}, \frac{\log{(2/\pf)}}{  \eps^2 \cdot \delta}\}\right).
$$

When $\delta = O(1/n)$, $\pf = O(1/n)$, the above time complexity becomes
$$O\left(\frac{1}{ \eps}\min\{ \sqrt{m\cdot n \cdot \log{n}},\frac{ n\cdot\log{n}}{  \eps^2 } \}\right)$$
for general graphs.
When the graph is scale-free,  in which case $m/n=O(\log{n})$, the time complexity becomes $O\left(\frac{1}{ \eps}n\cdot \log{n}\right)$, improving over the {MC} approach by $1/\epsilon$.

%!TEX root=../ssppr_kdd17_newformat.tex

\subsection{Indexing Scheme}\label{sec:ehi}

Based on { \ssppr}, we propose a simple and effective index structure to further improve the efficiency of whole-graph SSPPR queries. The basic idea is to pre-compute a number of random walks from each node $v$, and then store the destination of each walk. During query processing, if { \ssppr} requires performing $x$ random walks from $v$, we would inspect the set $S$ of random walk destinations pre-computed for $v$, and then retrieve the first $x$ nodes in $S$. As such, we avoid generating any random walks on-the-fly, which considerably reduces query overheads.

A natural question to ask is: how many random walks should we pre-compute for each node $v$? To answer this question, we first recall that, when the local update phase of { \ssppr} terminates, the residue of each node $v$ is at most $|\outN(v)| \cdot \rmax$. Combining this with Lemma~\ref{lem:layerppr-conc}, we can see that the number of random walks from $v$ required by { \ssppr} is
\begin{equation} \label{eqn:index-rwnum}
\omega_{max}(v) =\left\lceil|\outN(v)|\cdot\rmax \cdot \frac{(2\eps/3+2)\cdot \log{(2/\pf)}}{\eps^2\cdot \delta} \right\rceil.
\end{equation}

If we set $\rmax$ according to Equation~\ref{eqn:rmax} in which case $m\cdot \rmax \leq 1$, we have
\begin{align*}
\omega_{max}(v) = \left\lceil|\outN(v)|\cdot\frac{1}{\eps \cdot \sqrt{m\cdot \delta}} \cdot \sqrt{(2\eps/3+2)\cdot \log{(2/\pf)}} \right\rceil
\end{align*}

Otherwise, $\rmax$ is set to $\frac{\eps^2\cdot \delta}{(2\eps/3+2)\cdot \log{(2/\pf)}} $, we have
\begin{align*}
\omega_{max}(v) =|\outN(v)|.
\end{align*}

In summary, we pre-compute $\omega_{max}(v)$ random walks from each node $v$, and record the last nodes of those walks in our index structure. The total space overhead incurred is then bounded by
\begin{align*}
\sum_{v} \omega_{max}(v) & \leq \min\{ \sum_{v} \textstyle \left\lceil|\outN(v)|\cdot\frac{\sqrt{(2\eps/3+2)\cdot \log{(2/\pf)}}}{\eps \cdot \sqrt{m\cdot \delta}}  \right\rceil, \sum_v{|\outN(v)|} \}\\
& \le \min \{n + \frac{\sqrt{m}}{\eps \cdot \sqrt{\delta}} \cdot \sqrt{(2\eps/3+2)\cdot \log{(2/\pf)}}, m\}.
\end{align*}
Therefore, we have the following lemma.
\begin{lemma}\label{thm:idx-size}
The space consumption of our index structure is
\begin{equation}\label{eqn:idx-size}
\textstyle O\left(\min\{n+ \frac{1}{\eps}\sqrt{\frac{m \log{(1/\pf)}}{\delta}}, m\}\right).
\end{equation}
When $\delta = O(1/n)$, $\pf = O(1/n)$, and $m/n = O(\log n)$, the above space complexity becomes
$O\left(\min\{\frac{1}{ \eps}n\cdot \log{n}, m\}\right).$ \done
\end{lemma}

%\header
{\bf Remark.} One may wonder whether we can also pre-compute the Forward Push result for each node, so that we can answer each query by a simple combination of pre-processed Forward Push and random walks, which could lead to higher query efficiency. However, we note that storing the Forward Push results for all nodes incurs significant space overheads. In particular, it requires $O\left(\min\{n, 1/\rmax\}\right)$ space for each node, where $\rmax$ is set according to Equation~\ref{eqn:rmax}. As such, the total space consumption for preprocessing Forward Push results is

\vspace{-2mm}
$$\textstyle O\left(\min\left\{n^2, \frac{n}{\epsilon}\cdot \sqrt{\frac{m\cdot \log{(1/\pf)}}{\delta}}\right\}\right),$$
which is prohibitive for large graphs. Therefore, we do not store Forward Push results in our index structure.

\section{Optimizations for whole-graph SSPPR queries}\label{sec:opt}
In this section, we present optimization techniques to reduce the index size or improve the query efficiency of our {\ssppr} algorithm on whole-graph SSPPR queries. In particular, in Section \ref{sec:opt-zero-hop}, we will present our technique to reduce the index size by avoiding zero-hop nodes in the index structure. In Section \ref{sec:opt-balanced}, we will present our technique to reduce the query time by balancing the forward push and the random walk costs.

\subsection{Pruning zero-hop random walks}\label{sec:opt-zero-hop}
Recall that in {\ssppr}, after the local update phase, we sample random walks from each source with non-zero residues. Our main observation is that: $\alpha$ portion of the random walks is expected to stop at the current node, and with $O(1)$ time we can immediately record the portion of such random walks and hence avoid simulating $\alpha$ portion of the total random walks. However, a question is that, can we still provide approximation guarantee while exploring this pruning strategy? We next demonstrate how to use the reduction of zero-hop random walk idea to ensure the approximation guarantee. We first define two random variables $\pi_0(s,t)$ and $\pi_1(s,t)$. We define $\pi_0(s,t)$ as the probability that a random walk from $s$ immediately stopped at node $t$, i.e., the length of the random walk is $0$; we further define $\pi_{1}(s,t)$ as the probability that a random walk from $s$ that stopped at node $t$ after traversing at least one node, i.e., the length of the random walk is at least $1$.
Then it is clear that the personalized PageRank $\pi(s,t)$ satisfies the following equation.
$$
\pi(s,t) = \pi_0(s,t) + \pi_{1}(s,t)
$$

We hence rewrite Equation \ref{eqn:layerppr} as follows:
$$
\pi(s,t) = \rese(s,t) + \sum_{v\in V} r(s,v) \cdot (\pi_0(v,t) + \pi_1(v,t)).
$$
Also notice that $\pi_0(s,t)$ either equals $\alpha$ or 0 depending on whether $s = t$ or not.
Therefore, the above equation can be further rewritten as:
$$
\pi(s,t) = \rese(s,t) + r(s,t)\cdot \alpha+ \sum_{v\in V} r(s,v) \cdot  \pi_1(v,t).
$$

We define a random variable $X_v$ as follows: we randomly select one of a out-neighbor $u$ of $v$, and then start a random walk from $u$. If the random walk stops at $t$, then $X_v=1$, otherwise, $X_v=0$. Then, it is not difficult to verify that $(1-\alpha)\cdot\E[X]= \pi_1(s,t)$.

$$
\pi(s,t) = \rese(s,t) + r(s,t)\cdot \alpha+ \sum_{v\in V} (1-\alpha)\cdot r(s,v) \cdot  \E[X_v].
$$

Define $r'(s,v) = (1-\alpha)\cdot r(s,v)$, we have that:
$$
\pi(s,t) = \rese(s,t) + r(s,t)\cdot \alpha+ \sum_{v\in V}  r'(s,v) \cdot  \E[X_v].
$$
With the above equation, we can then further apply the same technique proposed in Section \ref{sec:ss-rationale}. We define $\rsum' = \sum_{v\in V} r'(s,v)$. Then, by sampling $\omega'=\rsum'\cdot\frac{(2\eps/3+2)\cdot \log{(2/\pf)}}{\eps^2\cdot \delta}$ random walks, we can provide approximation guarantee for the whole-graph SSPPR queries. With this approach, when the same $\rmax$ is used as Algorithm \ref{alg:layer-topk}, the maximum number of random walks sampled from a node $v$ can be bounded by
$$
(1-\alpha)\cdot\rmax \cdot \outN(v)\cdot \frac{(2\eps/3+2)\cdot \log{(2/\pf)}}{\eps^2\cdot \delta},
$$
while previously, it requires
$$
\rmax \cdot \outN(v)\cdot \frac{(2\eps/3+2)\cdot \log{(2/\pf)}}{\eps^2\cdot \delta}.
$$

Therefore, with this optimization, we can also reduce the index size by $\alpha$ portion. Next, we further demonstrate our second optimization technique to improve the query efficiency.

\subsection{Balancing Forward Push and Random Walk Cost}\label{sec:opt-balanced}
Recall that in {\ssppr}, we set $\rmax$ according to Equation \ref{eqn:rmax} to minimize the time complexity of {\ssppr}. However, in practice, the derived $\rmax$ may not be the best choice since Equation \ref{eqn:rmax} considers the worst case while in practice the running time might be quite different. Of course, we may tune $\rmax$ for different dataset and choose $\rmax$ that derives the best piratical performance. However, the tuned $\rmax$ will typically be data-dependent. Here we are aiming to propose a solution that balances the forward push and random walk costs without any dependency on the datasets.

{\cblue
In \cite{Lofgrenthesis15}, they propose a balanced approach for {\em BiPPR} by doing the backward propagation and maintaining the largest residue using a max-heap. Since the total number of random walks depends linearly on the maximum residue, it estimates the running time of the random walk and stops the backward propagation as soon as the running time is around the same as the forward random walk. Our proposed balancing strategy shares the similar spirit as theirs. However, we do not maintain the priority queue since the number of random walks of FORA depends linearly on the total sum of the residues instead of the maximum residue. When we finish the forward push, we can accurately estimate the running time of the random walk part since {\em (i)} we know the total number of random walks; {\em (ii)} the average running time of one random walk depends only on $\alpha$, which is dataset-independent. Therefore, we can easily estimate the cost of a random walk and use it for random walk cost estimation no matter what dataset we are running on. Therefore, we propose the adaptive approach to balance the forward push and random walk cost as shown in Algorithm \ref{alg:ss-algo-opt}.

Initially, we start the forward push by setting $\rmax =1$ and  calculate the current accumulated forward push cost (Algorithm \ref{alg:ss-algo-opt} Line 11). Then, it checks if the forward push cost is still lower than the estimated random walk cost (Algorithm \ref{alg:ss-algo-opt} Line 4). If this is the case, the algorithm continue the forward push process, update the accumulated forward push cost, and update the estimated random walk cost (Algorithm \ref{alg:ss-algo-opt} Lines 5-11). The forward push terminates as soon as the estimated random walk cost is larger than the forward push cost (Algorithm \ref{alg:ss-algo-opt} Line 4). By this strategy, it guarantees that, when the forward push terminates, the cost will not differ from the random walk cost by a large margin. The random walk phase is similar to the one in Algorithm \ref{alg:ss-algo} except that here we prune the zero-hop random walks as mentioned in Section \ref{sec:opt-zero-hop}. In particular, we convert $\alpha$ portion of the residue $r(s,v_i)$ to its reserve (Algorithm \ref{alg:ss-algo-opt} Line 12), and the residue of node $v_i$ is reduced to $(1-\alpha)\cdot v_i$ (Algorithm \ref{alg:ss-algo-opt} Line 14). When we sample a random walk from each node $v_i$, we first randomly select one of its out-neighbor and then do random walk from these nodes thus avoiding the zero-hop random walks.
}

As we will see in our experimental evaluation, the balanced strategy can help reduce the average running time of whole-graph SSPPR queries by almost half, which demonstrates the effectiveness of the balancing strategy.
For the indexing version of our FORA and the top-$k$ algorithm, we tune $\rmax$ to evaluate the trade-off between the index size and the query performance in Section \ref{sec:exp}.

\begin{algorithm}[t]
\caption{{\em \ssppr} for Whole-Graph SSPPR with optimization}\label{alg:ss-algo-opt}
\BlankLine
\KwIn{Graph $G$, source node $s$, probability $\a$, relative error threshold $\epsilon$ }
\KwOut{Estimated PPR $\epi(s,v)$for all $v\in V$}

Let $rc$ be the cost of a random walk, and $FC$ be the running time of the forward push\;
Let $r(s,v_i), \rese(s,v_i)$ be the residue and reserve of node $v_i$ in forward push, and initially only $r(s,s)=1$ while all other values are zero\;
Let $FC=0, \rmax = 1, \rsum =1, \omega = \rsum\cdot\frac{(2\eps/3+2)\cdot \log{(2/\pf)}}{\eps^2\cdot \delta}$\;
\While {$\exists v\in V$ such that $r(s, v)/|\outN(v)| > r_{max}$ and $FC< \omega*rc$}
{
    \For {each $u \in \mathcal{N}^{out}(v)$}
    {
       $r(s, u) \gets r(s, u) + (1-\alpha) \cdot \frac{ r(s, v)}{|\outN(v)|}$
    }
    $\rese(s,v) \gets \rese(s,v) + \alpha  \cdot r(s, v)$\;
    $\rsum = \rsum - \alpha\cdot r(s,v)$\;
    $\omega=\rsum\cdot(1-\alpha)\cdot\frac{(2\eps/3+2)\cdot \log{(2/\pf)}}{\eps^2\cdot \delta}$\;
    $r(s, v) \gets 0$\;
    $FC \leftarrow$ current elapsed time\;
}

Let $\epi(s,v_i)=\rese(s,v_i)+ \alpha\cdot r(s,v_i)$ for all $v_i \in V$\;
\For{$v_i\in V$ with $r(s,v_i)>0$}
{
    Let $\resi(s,v_i)= (1-\alpha)\cdot\resi(s,v_i), \omega_i = \lceil r(s,v_i)(1-\alpha)\cdot\omega/\rsum\rceil$\;
    Let $a_i = \frac{\resi(s,v_i)}{\rsum}\cdot\frac{\omega}{\omega_i}$\;
    \For{$i=1$ to $\omega_i$}
    {
        Randomly select a out-neighbor $u$ of $v_i$ and generate a random walk $W$ from $u$\;
        Let $t$ be the end point of $W$\;
        $\epi(s,t) += \frac{a_i\cdot \rsum}{\omega}$\;
    }
}
return $\epi(s,v_1),\cdots, \epi(s,v_n)$\;
\end{algorithm}

%!TEX root=../ssppr_kdd17_newformat.tex

\section{Top-$\boldsymbol{k}$ SSPPR} \label{sec:topk}

In this section, we discuss how { \ssppr} handles approximate top-$k$ SSPPR queries.

\header
{\bf Rationale.} A straightforward approach to answer a top-$k$ SSPPR query with { \ssppr} is to first apply it to perform a whole-graph SSPPR query, and then returns the $k$ nodes with the largest approximate PPR values. However, if we are to satisfy the accuracy requirement described in Definition~\ref{def:topk}, we would need to set the parameters of { \ssppr} according to the exact $k$-th largest PPR value $\pi(s, v^*_k)$, which is unknown in advance. To address this, a naive solution is to conservatively set $\pi(s, v^*_k) = 1/n$, which, however, would lead to unnecessary overheads.

To avoid the aforementioned overheads, we propose a trial-and-error approach as follows. We first assume that $\pi(s, v^*_k)$ is a large value (e.g., $1/k$), and we set the parameters of { \ssppr} accordingly to perform a whole-graph SSPPR query. After that, we inspect the results obtained to check whether the estimated PPR values are indeed large. If they are not as large as we have assumed, then we re-run { \ssppr} with more conservative parameters, and check the new results returned. This process is conducted iteratively, until we are confident that the results from { \ssppr} conform to the requirements in Definition~\ref{def:topk}. In this section, we first present a top-$k$ algorithm by iteratively refining the upper and lower bounds of the top-$k$ PPR results in Section \ref{sec:topk-bound-refine}. Nevertheless, it is expensive to calculate the upper and lower bounds for each node in each iteration, and it is unclear whether the algorithm will terminate with $\delta$ close to $\pi(s,v^*_k)$ or not. Therefore, the expected running time of Algorithm \ref{alg:topk-bound-refine} might be identical to that of invoking Algorithm \ref{alg:ss-algo} with $\delta=1/n$. Therefore, in Section \ref{sec:topk-running-time-guarantee}, we further propose a new top-$k$ query algorithm that provides guarantee on the expected running time, and that the algorithm has high probability to terminate with $\pi(s,v_k^*)/4 \leq \delta \leq \pi(s,v_k^*)$.

\subsection{Top-$\boldsymbol{k}$ with bound refinement}\label{sec:topk-bound-refine}

\begin{algorithm}[t]
\begin{small}
\caption{Top-$k$ {\ssppr} with bound refinement} \label{alg:topk-bound-refine}
\BlankLine
\KwIn{Graph $G$, source node $s$, probability $\a$}
\KwOut{$k$ nodes with the highest approximate PPR scores}
\For{$\delta = \frac{1}{k}, \frac{1}{2k}, \frac{1}{4k},\cdots, \frac{1}{n}$}{
    Invoke Algorithm 2 with $G$, $s$, $\a$, and $\rmax$ set by Equation \ref{eqn:rmax} and fail probability $\pf'=\frac{\pf}{n \cdot \log{n}}$\;
    Let $C=\{v'_1,\cdots,v'_k\}$ be the set that contains the $k$ nodes with the top-$k$ largest lower bounds (from Theorem~\ref{thm:conc})\;
    Let $LB(u)$ and $UB(u)$ be the lower and upper bounds of $\pi(s, u)$ (from Theorem~\ref{thm:conc})\;
    \If{$UB(v'_i)<(1+ \eps)\cdot LB(v'_i)$ for $i\in[1,k]$ \textbf{and} $LB(v'_k) \ge \delta$}{
        Let $U $ be the set of nodes $u\in V \setminus C$ such that $UB(u)>(1+\epsilon)\cdot LB(v'_k)$\;
        \If{$\not\exists u \in U \textbf{ such that }   UB(u)< (1+\epsilon) \cdot LB(u)/(1-\epsilon)$ }{
            return $v'_1, v'_2,\cdots, v'_k$ and their estimated PPR\;
        }
    }
}
\end{small}
\end{algorithm}

\header
{\bf Algorithm.} Algorithm~\ref{alg:topk-bound-refine} shows the pseudo-code of the top-$k$ extension of { \ssppr} with bound refinement. The algorithm consists of at most $\log n$ iterations. In the $i$-th iteration, we invoke Algorithm 2 with $\delta$ set to $\frac{1}{2^{i-1}\cdot k}$, and the failure probability set to $\pf' = \frac{\pf}{n\log{n}}$ (Lines 1-2 in Algorithm~\ref{alg:topk-bound-refine}). (The reason for this setting will be explained shortly). After we obtain the results from {\em \ssppr}, we compute an upper bound and a lower bound of each node's PPR value, and use them to decide whether the current top-$k$ results are sufficiently accurate (Lines 3-8). If the top-$k$ results are accurate, then we return them as the top-$k$ answers (Line 8); otherwise, we proceed to the next iteration. In the following, we elaborate how the upper and lower bounds of each node's PPR value is derived.

Define $LB_0(v) = 0$ and $UB_0(v)=1$ for any $v\in V$. We have the following theorem that establishes the lower bound $LB_j(v)$ and upper bound $UB_j(v)$ of $\pi(s,v)$ in the $j$-th iteration of Algorithm~\ref{alg:topk-bound-refine}:
\begin{theorem}\label{thm:bounds}
In the $j$-th iteration of Algorithm~\ref{alg:topk-bound-refine}, let $\omega_j$ be the $\omega$ calculated by {\it \ssppr} (Algorithm 2 Line 2) in this iteration, and $\rese_j(s,v)$ and $\epi_j(s, v)$ be the reserve and estimated PPR of $v$. Define
$$\eps_j= \sqrt{\frac{3\rsum\cdot\log{(2/\pf')}}{\omega_j \cdot \max\{ \rese_j(s,v), LB_{j-1}(s,v) \}}}, \textrm{and } \lambda_j = \frac{2/3\log{(2/pf')}}{2\omega_j}$$ $$+\frac{\sqrt{\frac{4}{9}\rsum^2\cdot\log^2{(2/\pf')}+8\rsum\cdot \omega_j\cdot \log{(2/\pf')}\cdot UB_{j-1}(v)}}{2\omega_j}$$
Then, with at least $1-\pf'$ probability, the following two inequalities hold simultaneously:
\begin{align*}
 \epi_j(s,v)/(1+\epsilon_j) \le &\pi(s,v) \le  \epi_j(s,v)/(1- \epsilon_j) \\
 \epi_j(s,v) - \lambda_i \le &\pi(s,v) \le  \epi_j(s,v) + \lambda_i.
\end{align*}
%Let $\eps'= \sqrt{\frac{\log{(1/\pf')}\cdot(3)}{\omega \max\{}}}$
\end{theorem}
\begin{proof}
  Ref. Appendix \ref{app:proofs}.
\end{proof}

Theorem~\ref{thm:bounds} enables us to derive tight lower and upper bounds of each node's PPR value in each iteration. In particular, we set
$$UB_j(v) = min\{1,  \epi_j(s,v)/(1-\epsilon), \epi_j(s,v) + \lambda_i\},$$
$$LB_j(v) = \max\{  \epi_j(s,v)/(1+\epsilon),\epi_j(s,v) - \lambda_i, 0 \}.$$
With these upper and lower bounds, the following theorem shows that if Lines 5 and 7 in Algorithm~\ref{alg:topk-bound-refine} holds, then Algorithm~\ref{alg:topk-bound-refine} returns the answer for the approximate top-$k$ SSPPR query.
\begin{theorem}[Approximate Top-$k$]\label{thm:layer-topk}
Let $v'_1, \cdots, v'_k$ be the $k$ nodes with the largest lower bounds in the $j$-th iteration of Algorithm~\ref{alg:topk-bound-refine}. Let $U $ be the set of nodes $u\in V \setminus C$ such that $UB_j(u)>(1+\epsilon)\cdot LB_j(v'_k)$, If $UB(v'_i)<(1+ \eps)\cdot LB(v'_i)$ for $i\in[1,k]$,  $LB_j(v'_k) \ge \delta$, and there exists no $u\in U$ such that $UB_j(u)<(1+\epsilon)\cdot LB_j(u)/(1-\epsilon)$,
then returning $v'_1, \cdots, v'_k$ and their estimated PPR values would satisfy the requirements in Definition~\ref{def:topk} with at least $1-j\cdot n\cdot\pf'$ probability.
\end{theorem}
\begin{proof}
  Ref. Appendix \ref{app:proofs}.
\end{proof}
Now recall that the number of iterations in Algorithm~\ref{alg:topk-bound-refine} is $\log n$, and in each iteration, we assume that the upper and lower bounds are correct. Hence, by applying union bound, the failure probability will be at most $n\log{n}\cdot \pf'$. Note that $\pf'= \frac{\pf}{n\log{n}}$. The failure probability is hence no more than $\pf$, and we guarantee that the returned answer has approximation with at least $1-\pf$ probability.

\subsection{Top-$\boldsymbol{k}$ with improved time complexity} \label{sec:topk-running-time-guarantee}
Despite the fact that Algorithm \ref{alg:topk-bound-refine} provides superb performance on top-$k$ query processing as shown in \cite{Wang17}, there is no guarantee that the running time will depends on $\delta =\pi(s,v_k^*)$ instead of $\delta = 1/n$.
Also, after each iteration, we need to re-calculate the lower- and upper-bounds for each node, which may take more than half of the query running time. This motivates us to propose our new top-$k$ algorithm that avoids the overheads of the bound-refinement, and provides running time guarantees with respect to $\pi(s,v_k^*)$.

\header
{\bf Algorithm.} Algorithm~\ref{alg:layer-topk} shows the pseudo-code of the top-$k$ extension of { \ssppr}. The algorithm consists of at most $\log_2{( n/k)}$ iterations. In the $i$-th iteration, we invoke Algorithm 2 with $\delta$ set to $\frac{1}{k\cdot 2^{i-1}}$, relative error threshold $\eps'=\eps/2$ and the failure probability set to $\pf' = \frac{\pf}{n\log_2{(n/k)}}$ (Lines 1-2 in Algorithm~\ref{alg:layer-topk}). (The reason for this setting will be explained shortly). After we obtain the results from {\em \ssppr}, we compute the estimated PPR scores for each node and get the $k$-th largest estimated PPR value. We compare the $k$-th estimated PPR score with $(1+\epsilon)\cdot \delta$ and use this as evidence to see whether the current top-$k$ results are sufficiently accurate (Line 4). If the top-$k$ results are accurate, then we return them as the top-$k$ answers (Lines 5-6); otherwise, we halve the value of $\delta$ and proceed to the next iteration. In the following, we analyze the approximation guarantee and time complexity of Algorithm \ref{alg:layer-topk}.
% theorem establishes the approximation guarantee and time complexity of our top-$k$ algorithm.
%
%\begin{theorem}
% When Algorithm \ref{alg:layer-topk} terminates, $\delta \in [\pi(s,v_k^*)/8, \pi(s,v_k^*)/4]$, and it returns an approximate top-$k$ answer with $1-\pf$ probability. The expected time complexity of Algorithm \ref{alg:layer-topk} is $$\textstyle O\left( \frac{1}{ \eps \cdot \sqrt{\pi(s,v_k^*)}}\sqrt{m \cdot (2\eps/3+2)\cdot \log{(2/\pf)}}\right).$$
%\end{theorem}
Firstly, we have the following lemma about the value of $\delta$ when Algorithm \ref{alg:layer-topk} terminates.
\begin{lemma}\label{lem:topk-delta}
Let $v_k$ be the node that has the $k$-th largest estimated PPR value and $\hat{\pi}(s,v_k)$ be the estimated PPR value for node $v_k$ with respect to $s$. Let $v_k^*$ be the node with the true $k$-th largest PPR. Then, when Algorithm \ref{alg:layer-topk} terminates,  it holds for $\delta$  that:
\begin{itemize}
  \item $\delta>\pi(s,v_k^*)$ with at most $n\cdot p_f'/8$ probability;
  \item $\delta\leq \pi(s,v_k^*)$ with at least $ 1-\log_2\frac{1}{\pi(s,v_k^*)}\cdot n\cdot \pf'/8$;
  \item $\delta \leq \pi(s,v_k^*)/2^{x+1}$ ($x=1,2,3\cdots$) with at most $\frac{\pf'}{2^{x}}$ probability.
\end{itemize}
\end{lemma}

\begin{proof}
  We first consider the case when Algorithm \ref{alg:layer-topk} terminates with $\delta > \pi(s,v_k^*)$.
  %\header

  \textbf{$\boldsymbol {\delta>\pi(s,v_k^*)}$.} If $\delta >\pi(s,v_k^*)$ when Algorithm \ref{alg:layer-topk} terminates, we know that there exists at least $n-k+1$ nodes such that $\pi(s,v)<\delta $. Denote $X$ as the set of nodes such that $v\in X$ if $\pi(s,v)<\delta$. For any of these nodes, we consider the probability that their PPR values with respect to $s$ is greater than $(1+\epsilon)\cdot \delta$. Note that according to Line 2 of Algorithm \ref{alg:layer-topk}, the number of random walks is set to $\frac{\rsum \cdot (\eps/2\cdot2/3+2)\cdot \log{(2/\pf')}}{(\eps/2)^2\cdot \delta}$.
  Let $t$ be a node in $X$, and $\lambda=\epsilon\cdot \delta$, according to Lemma \ref{lem:approximation}, we have that:
\begin{align}
    \Pr[\hat{\pi}(s,t)-\pi(s,t)>\lambda] \leq  & \exp\left( - \frac{\lambda^2}{\rsum \cdot (2\pi(s,t)+2\lambda/3)}\cdot \frac{4\rsum \cdot (2+\epsilon/3)\cdot \log{(2/\pf')}}{\epsilon^2\cdot \delta}  \right)\nonumber \\
    \leq & \exp\left ( - \frac{4\delta \cdot (2+\epsilon/3)}{2\cdot \pi(s,t)+2\epsilon\cdot \delta/3}  \cdot \log{(2/\pf')}\right) \nonumber \\
    \leq & \exp\left ( -\frac{4\delta \cdot (2+\epsilon/3)}{2\cdot \delta+2\epsilon\cdot \delta/3}\cdot \log{(2/\pf')} \right )   \text{ \quad (}\delta >\pi(s,t))\nonumber\\
    \leq & \exp\left ( -\frac{4 \cdot (2+\epsilon/3)}{2+2\epsilon/3} \cdot \log{(2/\pf')}\right )  \nonumber \\
    < &  \exp\left ( -3 \cdot \log{(2/\pf')}\right )< \pf'/8 \nonumber
\end{align}

So, by union bound,  it is satisfied that $\hat{\pi}(s,t)<\pi(s,t) + \epsilon \cdot \delta$ holds for any node $t\in X$ with at least $1-n\cdot \pf'/8$ probability. Also note that $\pi(s,t)<\delta$. This indicates that with at least probability $1-n\cdot \pf'/8$,  for all nodes $t\in X$,it is satisfied that
$\hat{\pi}(s,t) < (1+\epsilon)\cdot\delta$. Since there are at least $n-k+1$ nodes in $X$. It indicates that the returned $\hat{\pi}(s,v_k)$ must be no larger than $max_{v\in X} \hat{\pi}(s,v)$, which is less than $(1+\epsilon)\cdot \delta$. However, this contradicts to the fact that Algorithm \ref{alg:layer-topk} terminates when $\hat{\pi}(s,v_k)\geq (1+\epsilon)\cdot \delta$. As a result, with at most $n\cdot \pf'/8$ probability, the algorithm terminates when $\delta>\pi(s,v_k^*)$.

\begin{algorithm}[t]
\caption{Top-$k$ \ssppr} \label{alg:layer-topk}
\BlankLine
\KwIn{Graph $G$, source node $s$, probability $\a$}
\KwOut{$k$ nodes with the highest approximate PPR scores}
\For{$\delta = \frac{1}{k},\frac{1}{2k}\cdots, \frac{1}{n}$}{

    Invoke Algorithm 2 with $G$, $s$, and $\a$ with $\epsilon'=\epsilon/2$, failure probability $\pf'=\frac{\pf}{n\log_2{(n/k)}}$, and $\rmax$ set  by Equation \ref{eqn:rmax}\;

    Let $\hat\pi(s,v_k)$ be the $k$-th largest estimated PPR score returned by Algorithm 2\;
    \If{$ \hat\pi(s,v_k)\ge (1+\epsilon)\cdot \delta$}{
            Let $v'_1, v'_2,\cdots, v'_k$ be the $k$ nodes with the top-$k$ largest PPR values\;
            return $v'_1, v'_2,\cdots, v'_k$ and their estimated PPR values\;
    }
    %$\delta = \delta /2$\;
}
\end{algorithm}
Next, we consider the probability when Algorithm \ref{alg:layer-topk} terminates with $\delta \leq \pi(s, v_k^*)$.

  \textbf{$\boldsymbol {\delta\leq \pi(s,v_k^*)}$.} Since in Algorithm \ref{alg:layer-topk}, there are at most $\log_2{\frac{1}{k\cdot \pi(s,v_k*)}}$ iterations such that $\delta>\pi(s,v_k^*)$. The probability that the algorithm terminates when $\delta>\pi(s,v_k^*)$ is bounded by $\log_2{\frac{1}{k\cdot \pi(s,v_k^*)}}n\cdot \pf'/8$. Hence, with at least  probability $1-\log_2{\frac{1}{k\cdot \pi(s,v_k^*)}}n\cdot \pf'/8$, the algorithm terminates with $\delta \leq \pi(s, v_k^*)$.

Finally, we consider whether $\delta$ will be too small when Algorithm \ref{alg:layer-topk} terminates.

\textbf{$\boldsymbol{\delta \leq \pi(s,v_k^*)/2^{x+1}}$.} Consider the $k$ nodes with the $k$ largest PPR values. Denote these nodes as $v_1^*, v_2^*, v_3^*,\cdots,v_k^*$. Consider the estimated PPR of $v_i^*$ with respect to $s$ ($1\leq i \leq k$).
According to Lemma \ref{lem:approximation}, it is satisfied that:
\begin{align*}
  \Pr[\hat{\pi}(s,v_i^*) \leq (1+\epsilon)\cdot \delta ] \leq & \Pr[\hat{\pi}(s,v_i^*) \leq (1+\epsilon)\cdot \pi(s,v_k^*)/2^{x+1} ]  \text{ \quad (}\delta \leq \pi(s,v_k^*)/2^{x+1})\\
  \leq & \Pr[\hat{\pi}(s,v_i^*) \leq (1+\epsilon)\cdot \pi(s,v_i^*)/2^{x+1} ]  \text{ \quad (} \pi(s,v_k^*) \leq \pi(s,v_i^*) )\\
  \leq & \Pr[\hat{\pi}(s,v_i^*) \leq (1-(1-1/2^x)\cdot\epsilon)\cdot \pi(s,v_i^*) ] \\
  &(1+\epsilon)/2^{x+1} \leq 1- (1-1/2^x)\cdot \epsilon \text{ for }  0<\epsilon <1 \\
  \leq & \exp\left(-\frac{((1-1/2^x)\cdot\epsilon)^2\cdot \pi(s,v_i^*)}{\rsum\cdot(2 + 2(1-1/2^x)\cdot\epsilon/3)} \cdot  \frac{4\rsum \cdot (2+\epsilon/3)\cdot \log{(2/\pf')}}{\epsilon^2\cdot \delta} \right) \\
  \leq & \exp \left(- {2(1-1/2^x)^2\cdot 2^{x+1}}\cdot \frac{2(2+\epsilon/3)}{2+2(1-1/2^x)\cdot\epsilon} \cdot \log{(2/\pf')} \right) \\
  & 2(1-1/2^x)^2>1, \quad \frac{2(2+\epsilon/3)}{2+(1-1/2^x)}>1  \\
  \leq &\exp\left( - 2^{x+1}\cdot \log{(2/\pf')} \right)< \frac{(\pf')^2}{2^{x+1}}
\end{align*}

By union bound, the probability that $\hat{\pi}(s,v_i^*)\geq (1+\epsilon)\cdot \delta$ holds for any $1\leq i\leq k$ simultaneously is at least $1- k\cdot \frac{(\pf')^2}{2^{x+1}}\geq 1- \frac{\pf'}{2^{x+1}}$. Since there are $k$ estimations no smaller than $(1+\epsilon)\cdot \delta$, Algorithm \ref{alg:layer-topk} will terminate. Therefore, Algorithm \ref{alg:layer-topk} terminates when $\delta \leq \pi(s,v_i^*)/2^{i+1}$ with at most $\frac{\pf'}{2^{x+1}}$ probability, which finishes the proof.
\end{proof}

Lemma \ref{lem:topk-delta} indicates several desired properties of our top-$k$ algorithm. Firstly, the algorithm stops with $\delta > \pi(s,v_k^*)$ with low probability. Besides, it terminates with $\delta \leq \pi(s,v_k^*)$ with high probability, which is an important condition for providing approximate top-$k$ answer. Thirdly, when the algorithm terminates, $\delta$ will not deviate from $\pi(s,v_k^*)$ by a large margin. The larger the margin is between $\delta$  and $\pi(s,v_k^*)$, the lower the probability it is. By leveraging \ref{lem:topk-delta}, we further have the following lemma on the time complexity of our top-$k$ algorithm.

\begin{lemma}\label{lem:topk-complexity}
The expected running time of Algorithm \ref{alg:layer-topk} can be bounded by $$ O\left( \min \{ \frac{\sqrt{m \cdot \log{(2/\pf)}}}{ \eps \cdot \sqrt{\pi(s,v_k^*)}}, \frac{\log{(2/\pf)}}{\epsilon^2 \cdot \pi(s,v_k^*)}\}\right).$$
\end{lemma}
\begin{proof} Let $c$ be the constant factor of the time complexity of Algorithm \ref{alg:ss-algo} and $\delta^*$ be the value of $\delta$ when Algorithm \ref{alg:layer-topk} terminates.  Suppose that the algorithm terminates after $i+1$ iterations, and note that the time complexity of each iteration is bounded by $O\left(\frac{c}{ \eps \cdot \sqrt{\delta}}\sqrt{m \cdot (2\eps/3+2)\cdot \log{(2/\pf)}}\right)$.  Then, the cost $\mathcal{C}$ of {\ssppr} can be bounded by
\begin{align*}
& \mathcal{C}= \sum_{\delta= 1/k}^{1/(2^i \cdot k)}\frac{c}{ \eps \cdot \sqrt{\delta}}\sqrt{m \cdot (2\eps/3+2)\cdot \log{(2/\pf)}}  \\
   & = \frac{c}{ \eps }\sqrt{m \cdot (2\eps/3+2)\cdot \log{(2/\pf)}} \cdot \sum_{\delta = 1/k}^{1/(2^i \cdot k)}\frac{1}{\sqrt{\delta}} \\
   & \leq \frac{c}{ \eps }\sqrt{m \cdot (2\eps/3+2)\cdot \log{(2/\pf)}} \cdot \frac{1+1/\sqrt{\delta^*}}{1-1/\sqrt{2}}\\
   & \leq \frac{c}{ \eps }\sqrt{m \cdot (2\eps/3+2)\cdot \log{(2/\pf)}} \cdot \frac{4}{\sqrt{\delta^*}}
\end{align*}
Denote $$\phi = \frac{4c}{ \eps }\sqrt{m \cdot (2\eps/3+2)\cdot \log{(2/\pf)}} .$$ Next, we further consider the expected cost of Algorithm \ref{alg:layer-topk}.
\begin{align*}
  \E[\mathcal{C}]= & \Pr[\delta>\pi(s,v_k^*)]\cdot \mathcal{C}_{\delta>\pi(s,v_k^*)}+ \Pr[\pi(s,v_k^*)/4< \delta \leq \pi(s,v_k^*)]\cdot \mathcal{C}_{\pi(s,v_k^*)/4<\delta\leq \pi(s,v_k^*)} \\
   &+\sum_{x=1}^{\log_2{(n/\pi(s,v_k^*))}}\Pr[\pi(s,v_k^*)/2^{i+2}< \delta \leq \pi(s,v_k^*)/2^{i+1}]\cdot \mathcal{C}_{\pi(s,v_k^*)/2^{i+2}< \delta \leq \pi(s,v_k^*)/2^{i+1}} \\
  & <1\cdot \frac{\phi}{\sqrt{\pi(s,v_k^*)}} + 1\cdot \frac{ \phi}{\sqrt{\pi(s,v_k^*)/4}} + \sum_{i=1}^{\log_2{(n/\pi(s,v_k^*))}}\frac{\pf'}{2^{i+1}}\frac{\phi}{\sqrt{\pi(s,v_k^*)/2^{i+1}}}\\
  &< \frac{4\phi}{\sqrt{\pi(s,v_k^*)}} = \frac{16}{ \eps\cdot \sqrt{\pi(s,v_k^*)} }\sqrt{m \cdot (2\eps/3+2)\cdot \log{(2/\pf)}}
\end{align*}

Besides, note that the time complexity of each iteration can also be bounded by $O\left(\frac{(2\eps/3+2)\cdot \log{(2/\pf)}}{\delta}\right)$. Still let $c$ denote the constant in the time complexity. Then, the cost $\mathcal{C}$ of {\ssppr} can be bounded by

\begin{align*}
& \mathcal{C}= c\cdot\frac{(2\eps/3+2)\cdot \log{(2/\pf)}}{\epsilon^2} \cdot\sum_{\delta= 1/k}^{1/(2^i \cdot k)} \frac{1}{\delta} \leq c \cdot\frac{(2\eps/3+2)\cdot \log{(2/\pf)}}{\epsilon^2}\cdot \frac{2}{\delta}  \\
\end{align*}
Let $\phi = 2c \cdot\frac{(2\eps/3+2)\cdot \log{(2/\pf)}}{\epsilon^2}$, the expected cost of Algorithm \ref{alg:layer-topk} then can be further bounded by:
\begin{align*}
  \E[\mathcal{C}]= & \Pr[\delta>\pi(s,v_k^*)]\cdot \mathcal{C}_{\delta>\pi(s,v_k^*)}+ \Pr[\pi(s,v_k^*)/4< \delta \leq \pi(s,v_k^*)]\cdot \mathcal{C}_{\pi(s,v_k^*)/4<\delta\leq \pi(s,v_k^*)} \\
   &+\sum_{x=1}^{\log_2{(n/\pi(s,v_k^*))}}\Pr[\pi(s,v_k^*)/2^{i+2}< \delta \leq \pi(s,v_k^*)/2^{i+1}]\cdot \mathcal{C}_{\pi(s,v_k^*)/2^{i+2}< \delta \leq \pi(s,v_k^*)/2^{i+1}} \\
  & <1\cdot \frac{\phi}{\pi(s,v_k^*)} + 1\cdot \frac{ \phi}{\pi(s,v_k^*)/4} + \sum_{i=1}^{\log_2{(n/\pi(s,v_k^*))}}\frac{\pf'}{2^{i+1}}\cdot\frac{\phi}{\pi(s,v_k^*)/2^{i+1}}\\
  &< \frac{6\phi}{\pi(s,v_k^*)} = \frac{12c\cdot(2\eps/3+2)\cdot \log{(2/\pf)}}{\epsilon^2\cdot \pi(s,v_k^*)}
\end{align*}
Therefore, the expected time complexity of Algorithm \ref{alg:layer-topk} can be bounded by $$ O\left( \min \{ \frac{\sqrt{m \cdot \log{(2/\pf)}}}{ \eps \cdot \sqrt{\pi(s,v_k^*)}}, \frac{\log{(2/\pf)}}{\epsilon^2\cdot\pi(s,v_k^*)}\}\right),$$
which finishes the proof.
\end{proof}

It still remains to clarify whether Algorithm \ref{alg:layer-topk} returns approximate top-$k$ answers. The following lemma shows that our algorithm returns approximate top-$k$ answer with high probability.
\begin{lemma}\label{lem:topk-approximation}
Algorithm \ref{alg:layer-topk} returns an $\epsilon$-approximate top-$k$ answer with at least $1-\pf$ probability.
\end{lemma}
\begin{proof}
  Let $v_1, v_2,\cdots,v_k$ be the returned $k$ nodes by Algorithm \ref{alg:layer-topk}, and $R=\{v_1^*, v_2^*, \cdots ,v_k^*\}$ be the $k$ nodes with the real top-$k$ largest PPR values. According to Lemma \ref{lem:topk-delta}, Algorithm \ref{alg:layer-topk} terminates with $\delta \leq \pi(s,v_k^*)$ (denoted as Condition $C1$) with at least $ 1-\log_2\frac{1}{\pi(s,v_k^*)}\cdot n\cdot \pf'/8$ probability.
   When $C1$ holds, we note that for $v_i^*$, it is satisfied that, with $1-\pf'/2$ probability, $\hat{\pi}(s,v_i^*)-\pi(s,v_i^*)>\frac{\epsilon}{2}\cdot \pi(s,v_i^*)$.
As a result,
\begin{equation}\label{eqn:condition}
\hat{\pi}(s,v_i^*)-\pi(s,v_i^*)>\frac{\epsilon}{2}\cdot \pi(s,v_i^*) \text{ for any $1\leq i \leq k$}
\end{equation}
holds  with at least $1-k\cdot \pf'/2$ probability. We denote this condition as $C2$.

We next consider when conditions $C1$ and $C2$ both hold, the probability that the single source {\ssppr} fails to provide an $\epsilon$-approximate top-$k$ answer.
When $C1$ holds, we know that $\hat{\pi(s,v_i^*)} > (1-\epsilon/2)\cdot \pi(s,v_i^*)$. With this condition, $\hat{\pi}(s,v_i)$ must be larger than $(1-\epsilon/2)\cdot \pi(s,v_i^*)$ since its estimation is $i$-th largest and we know that there are at least $i$ nodes with estimated PPR larger than $(1-\epsilon/2)\cdot \pi(s,v_i^*)$, i.e., $\hat{\pi}(s,v_1^*), \hat{\pi}(s,v_2^*), \cdots, \hat{\pi}(s,v_i^*)$. We say a query fails if there exists a returned node $v_i$ such that:
\begin{enumerate}
  \item $\hat{\pi}(s,v_i)>(1- \epsilon/2)\cdot\pi(s,v_i)$,
  \item $\pi(s, v_i)< (1-\epsilon)\cdot \pi(s, v_i^*)$.
\end{enumerate}

Next, we prove that $v_i$ fails with very low probability.
%Suppose $\pi(s,v_i) \geq \pi(s,v_k^*)$. Then, we know that $\hat{\pi}(s,v_i)>(1-\epsilon/2)\cdot \pi(s,v_i)$, and this node will not fail. Therefore, we consider when $\pi(s,v_i)< \pi(s,v_k^*)$ and $\pi(s, v_i)< (1-\epsilon)\cdot \pi(s, v_i^*)$, the estimation $\hat{\pi}(s,v_i)$ of node $v_i$ satisfies that
%$\hat{\pi}(s,v_i)>(1- \epsilon/2)\cdot\pi(s,v_i)$.
Let $\epsilon' =\frac{(1-\epsilon/2)\cdot \pi(s,v_i^*)}{\pi(s,v_i)} -1$.
$$
\Pr[\hat{\pi}(s,v_i) > (1-\epsilon/2)\cdot \pi(s,v_i^*)] = \Pr[\hat{\pi}(s,v_i) > (1+\epsilon')\cdot \pi(s,v_i)]
$$
Since $\pi(s,v_i)<(1-\epsilon)\pi(s,v_i^*)$, we have that $\epsilon' > \frac{\epsilon/2}{1-\epsilon}$. Also note that
$$
\epsilon'\cdot \frac{\pi(s,v_i)}{\pi(s,v_i^*)} = (1-\epsilon/2) - \frac{\pi(s,v_i)}{\pi(s,v_i^*)} > (1-\epsilon/2)- (1-\epsilon) = \epsilon/2 .
$$
Then, according to Lemma \ref{lem:approximation}, it is satisfied that:
\begin{align*}
  \Pr[\hat{\pi}(s,v_i)&  > (1-\epsilon/2)\cdot \pi(s,v_i^*)] = \Pr[(1+\epsilon')\cdot \pi(s,v_i)] \\
   & \leq \exp\left( -\frac{(\epsilon')^2\cdot \pi(s,v_i)}{\rsum \cdot (2+2\epsilon'/3)} \cdot \frac{\rsum \cdot(2+\epsilon/3)\cdot \log{(2/\pf')}}{(\epsilon/2)^2\cdot \delta} \right) \\
   & \leq \exp\left( -\frac{\epsilon'}{(2+2\epsilon'/3)} \cdot \frac{\epsilon' \cdot \pi(s,v_i)}{\pi(s,v_i^*)} \cdot \frac{\pi(s,v_i^*)}{\delta}\cdot   \frac{ (2+\epsilon/3)\cdot \log{(2/\pf')}}{(\epsilon/2)^2} \right) \\
   &(\frac{\epsilon'}{2+2\epsilon'/3} \text{ is monotonically increasing,}  \frac{\epsilon'\cdot \pi(s,v_i)}{\pi(s,v_i^*)}>\epsilon/2, \frac{\pi(s,v_i^*)}{\delta} >1 ) \\
   &\leq \exp \left( -\frac{\epsilon/(2-2\epsilon)}{2+\epsilon/(3-3\epsilon)}\cdot \epsilon/2 \cdot 1\cdot \frac{ (2+\epsilon/3)\cdot \log{(2/\pf')}}{(\epsilon/2)^2} \right) \\
   & \leq \exp \left( -\frac{\epsilon/3+2}{2-5\epsilon/3}\cdot \log{(2/\pf')} \right) \quad (\frac{\epsilon/3+2}{2-5\epsilon/3}>1) \\
   &\leq \exp(-\log{(2/\pf')})
   = \pf'/2.
\end{align*}

As a result, when $C1$ holds, the probability that the query does not fail on any node is at least $1- n\cdot \pf'$ by applying union bound on the events that no $v_i$ fails and conditio $C2$ holds. Since in the worst case, there exists $\log{(n/\pi(s,v_k^*))}$ iterations when $C1$ holds. Therefore, we further have that the query returns $\epsilon$-approximate answer with at least
$$1-(\log_2\frac{1}{\pi(s,v_k^*)}\cdot n\cdot \pf'/8+\log{(n/\pi(s,v_k^*))}\cdot n\cdot \pf' )\geq 1-\pf$$
probability. This finishes the proof.
\end{proof}

%!TEX root=../ssppr_kdd17_newformat.tex
\section{Extensions}\label{sec:extensions}

\subsection{Extending {\em BiPPR} to Whole-Graph SSPPR}\label{app:ssbippr}
Recall from Section \ref{sec:existingsolutions} that in BiPPR, it includes both a forward phase and a backward phase. It is proved in \cite{AndersenBCHMT07} that the amortized time complexity for the backward phase is $O\left(\frac{m}{n \cdot \rmax}\right)$, and in \cite{lofgren2015personalized}, it shows that the forward phase requires $O\left( \frac{\rmax \cdot \log{(1/\pf)}}{\epsilon^2\cdot \delta} \right)$ time, given the backward phase threshold $\rmax$. Afterwards, they choose $\rmax = O\left(\epsilon\cdot\sqrt{\frac{m\cdot\delta}{n\cdot \log{(1/\pf)}}}\right)$ to minimize the time complexity for the pairwise PPR query, which is $O\left( \frac{1}{\eps}\sqrt{\frac{m \cdot \log{(1/\pf)}}{n\cdot \delta}}\right)$. To apply {BiPPR} for whole-graph SSPPR queries, a straightforward approach is to use it to answer $n$ point-to-point PPR queries (i.e., from $s$ to every other node). This, however, leads to a total time complexity of $O\left( \frac{1}{\eps}\sqrt{\frac{m n \cdot \log{(1/\pf)}}{\delta}}\right)$, which is a factor of $\sqrt{n}$ larger than that of the whole-graph SSPPR {\ssppr}.

To improve this, we observe that the $n$ point-to-point PPR queries share the same forward phase, and hence, we can conduct the forward phase once and then re-use its results for all $n$ backward phase. In addition, to reduce the total cost of $n$ backward phases, we can set $\rmax$ to a larger value; although it would require more random walks to be generated in the forward phase, the tradeoff is still favorable as the overhead of the forward phase has been significantly reduced by the re-usage of results. Since the backward phase (for all target nodes) has a cost of $O\left(\frac{m}{\rmax}\right)$, it can be verified that, by setting  $\rmax=O\left(\eps\cdot\sqrt{\frac{m\cdot \delta}{\log{(1/\pf)}}}\right)$, the expected time complexity of this optimized version of { BiPPR} (for SSPPR queries) is $$O\left( \frac{1}{ \eps \cdot \sqrt{\delta}}\sqrt{m \cdot \log{(1/\pf)}}\right),$$ which is identical to that of single-source { \ssppr}.

However, as we show in Section~\ref{sec:exp}, the optimized {\em BiPPR} is significantly outperformed by Whole-Graph SSPPR { \ssppr}. The reason is that, even after the aforementioned optimization, { BiPPR} either degrades to { MC} when $\rmax$ is large, or still requires performing a backward phase from each node $v$ in $G$, even if $\pi(s, v)$ is extremely small and can be omitted. In contrast, single-source { \ssppr} does not suffer from these deficiencies, and avoid examining nodes with very small PPR values.
Instead, it performs a forward search phase, followed by a number of random walks from the nodes visited in the search; this process tends to avoid examining nodes with very small PPR values, since those nodes are unlikely to be visited by the forward push or the random walks.

\subsection{Extending {\ssppr} to Source Distributions}\label{app:foraextension}
In many real applications of the personalized PageRank, the source $s$ can be a distribution (e.g., on a set of bookmark pages) instead of a single node. We show that our algorithms for single-source-node {\em \ssppr} can be extended to the case of arbitrary source distributions. 

Let $\sigma$ be the node distribution that the source node $s$ is sampled from. %Based on the definition,
For any target node $t$, its personalized PageRank with respect to $\sigma$ is defined as \cite{Haveliwala02,lofgren2015personalized}
$$\pi(\sigma, t)=\sum_{v\in V} \sigma(v) \cdot \pi(v,t),$$
where $\sigma(v)$ is the probability that a sample from $\sigma$ equals $v$.
To apply our algorithms, we modify Line 1 of Algorithm~\ref{alg:lp} to set the initial residue of each node $v$ as $\sigma(v)$. Let $\rese(\sigma,v)$ (resp.\ $r(\sigma,v)$) denote the reserve (resp.\ residue) of node $v$ in the modified version of Forward Push. Then, it is easy to prove that the following invariant holds for the modified version of Forward Push:
$$\pi(\sigma, t) =\rese(\sigma,t) + \sum_{v\in V} r(\sigma, v)\cdot \pi(v,t).$$
In particular, the initial states satisfy the above invariant, and by induction, it can be proved that the invariant still holds after every push operation. Given the above invariant, our algorithms can be applied to compute $\pi(\sigma, t)$ without compromising their asymptotic guarantees. Besides, the indexing scheme presented in Section~\ref{sec:ehi} is still applicable, since the maximum number of random walks required for each node is identical to that in the single-source-node algorithms.

{\cblue

\subsection{Extending {\ssppr} to Global PageRank}\label{app:foraextension}

Global PageRank can be regarded as the personalized PageRank with a source distribution of $(1/n, 1/n,\cdots, 1/n)$. According to our discussion in Section \ref{app:foraextension}, {\ssppr} can further be used to calculate the global PageRank. The classic solution for PageRank is the Power-Method, which takes a running time of $O(m\cdot \log{\frac{1}{\delta\cdot\epsilon}})$, to provide $\epsilon$-approximation for PageRank scores above the threshold $\delta$. To apply the Monte-Carlo method, we can first randomly sample source node, record the number of random walks stops at a node $v$, and use the fraction of random walks stopped at $v$ as the estimated PageRank. To derive $\epsilon$-approximation, the running time will be $O\left(\frac{(2\eps/3+2)\cdot \log{(2/\pf)}}{  \eps^2 \cdot \delta}\right)$. When $\epsilon$ is moderate, and $m > \frac{3}{\epsilon^2\cdot \delta}$, the Monte-Carlo approach achieves a better time complexity.

With the proposed FORA framework, the running time can be bounded by:
$$
\textstyle O\left(\min \{ \frac{\sqrt{m \cdot \log{(2/\pf)}}}{ \eps \cdot \sqrt{\delta}}, \frac{ \log{(2/\pf)}}{  \eps^2 \cdot \delta}\}\right),
$$
which actually achieves the best asymptotic performance compared to the two existing solutions when $\epsilon$ is moderate and $m > \frac{3}{\epsilon^2\cdot \delta}$.

Besides, our top-$k$ algorithm can be further extended to output the top-$k$ nodes with the highest PageRank scores with a running time of:
$$ O\left( \min \{ \frac{\sqrt{m \cdot \log{(2/\pf)}}}{ \eps \cdot \sqrt{\pi(v_k^*)}}, \frac{\log{(2/\pf)}}{\epsilon^2\cdot\pi(v_k^*)}\}\right),$$

where $\pi(v_k^*)$ is the node with the $k$-th largest PageRank score.
}

%!TEX root=../ssppr_kdd17_newformat.tex
\section{Other Related Work} \label{sec:other-related}

Apart from the methods discussed in Section~\ref{sec:existingsolutions} %\cite{fogaras2005towards，lofgren2015personalized, WangTXYL1}
, there exists a plethora of techniques for whole-graph and top-$k$ SSPPR queries. Those techniques, however, are either subsumed by {\em BiPPR} and {\em HubPPR} or unable to provide worst-case accuracy guarantees. In particular, a large number of techniques adopt the {\em matrix-based} approach, which formulates PPR values with the following equation:
%\vspace{-2mm}
\begin{equation}\label{eqn:powerit}
\vect{\pi_s} = \a \cdot \vect{e_s} + (1-\a)\cdot \vect{\pi_s}\cdot D^{-1}A,
\end{equation}
where $\vect{\pi_s}$ is a vector whose $i$-th element equals $\pi(s, v_i)$, $A \in \{0, 1\}^{n\times n}$ is the adjacency matrix of $G$, and $D\in R^{ n \times n}$ is a diagonal matrix in which each $i$-th element on its main diagonal equals the out-degree of $v_i$. Matrix-based methods typically start from an initial guess of $\vect{\pi_s}$, and then iteratively apply Equation~\ref{eqn:powerit} to refine the initial guess, until converge is achieved. Recent work that adopts this approach \cite{FujiwaraNYSO12,maehara2014computing,ZhuFCY13,ShinJSK15} propose to decompose the input graph into tree structures or sub-matrices, and utilize the decomposition to speed up the PPR queries.  The state-of-the-art approach for the single-source and top-$k$ PPR queries in this line of research work is {\em BEAR} proposed by Shin et al.\ \cite{ShinJSK15}. However, as shown in \cite{WangTXYL16}, the best of these methods is still inferior to {\em HubPPR} \cite{WangTXYL16} in terms of query efficiency and accuracy.

There also exist methods that follow similar approaches to the forward search method \cite{AndersenCL06} described in Section~\ref{sec:existingsolutions}. Berkin et al.\ \cite{Berkhin05} propose to pre-compute the Forward Push results from several important nodes, and then use these results to speed up the query performance. Ohsaka et al.\ \cite{OhsakaMK15} and Zhang et al.\ \cite{ZhangLG16} further design algorithms to update the stored Forward Push results on dynamic graphs. Jeh et al.\ \cite{JehW03} propose the backward search algorithm, which (i) is the reverse variant of the Forward Push method, and (ii) can calculates the estimated PPRs from all nodes to a target node $t$. Zhang et al.\ \cite{ZhangLG16} also design the algorithms to update the stored backward push results on dynamic graphs. Nonetheless, none of these solutions in this category provide approximation guarantees for single-source or top-$k$ PPR queries on directed graphs.  % The backward search algorithm is further optimized in \cite{fogaras2005towards,AndersenCL06}.

In addition, there are techniques based on the Monte-Carlo framework. Fogaras et al.\ \cite{fogaras2005towards} propose techniques to pre-store the random walk results, and use them to speed up the query processing. Nonetheless, the large space consumption of the technique renders it applicable only on small graphs. Bahmani et al.\ \cite{BackstromL11}, Sarma et al.\ \cite{SarmaMPU13} and Lin et al.\ \cite{Lin19} investigate the acceleration of the Monte-Carlo approach in distributed environments.  Lofgren et al.\ propose {\em FastPPR} \cite{lofgren2014fast}, which significantly outperforms the Monte-Carlo method in terms of query time. However, {\em FastPPR} in turn is subsumed by {\em BiPPR} \cite{lofgren2015personalized} in terms of query efficiency. In \cite{Lofgrenthesis15}, Lofgren further proposes to combine a modified version of Forward Push, random walks, and the backward search algorithm to reduce the processing time of pairwise PPR queries. Nevertheless, the time complexity of the method remains unclear, since Lofgren does not provide any theoretical analysis on the asymptotic performance of the method in \cite{Lofgrenthesis15}. Wang et al. \cite{WangT18} also consider combine the forward random walks and the back search to accelerate the heavy hitter queries in personalized PageRank;  Wei et al. \cite{WeiHX0LDW19} find connections between SimRank and PPR, combine the forward random walks with backward search, and propose the PRSim algorithm that can answer SimRank queries with sublinear time on power-law graphs.

Finally, a plethora of research work \cite{lofgren2015personalized,GuptaPC08,AvrachenkovLNSS11,BahmaniCX11,FujiwaraNYSO12,FujiwaraNSMO13sigmod} study how to efficiently process the top-$k$ PPR queries. Gupta et al.\  \cite{GuptaPC08}  propose to use Forward Push to return the top-$k$ answers. However, their solutions do not provide any approximation guarantee. Avrachenkov et al.\ \cite{AvrachenkovLNSS11} study how to use Monte-Carlo approach to find the top-$k$ nodes. Nevertheless, the solution does not return estimated PPR values and does not provide any worst-case assurance. Fujiwara et al.\ \cite{FujiwaraNYSO12,FujiwaraNSMO13sigmod} and Shin et al.\cite{ShinJSK15} investigate how to speed up the top-$k$ PPR queries with the matrix decomposition approach. These approaches either cannot scale to large graphs or do not provide approximation guarantees.

{\cblue
Most recently, Wei et al. propose the index-free TopPPR \cite{WeiHX0SW18}, which combines the Forward Push, random walk, and the backward propagation to answer top-$k$ PPR queries with precision guarantees. However, as we will see in our experiments, our FORA+ actually achieves a better performance than TopPPR when we set $\rho=0.99$ on large datasets and achieves a better trade-off among the space consumption, query efficiency, and query accuracy.
The main reason is that FORA+ can benefit from the indexing scheme while TopPPR cannot. It is also difficult for TopPPR to benefit from indexing scheme since {\em (i)} their random walk adopts the $\sqrt{\alpha}$-random walk, and all the nodes visited will be used to estimate the PPR scores; {\em (ii)} the sources of random walks in Monte-Carlo phase are generated randomly in TopPPR and it may results in poor cache performances. In contrast, FORA+ only need to scan the index structures in order and only stores the destinations in index structure, making it light-weighted and cache-friendly.
}

%!TEX root=../ssppr_kdd17_newformat.tex

\section{Experiments} \label{sec:exp}

In this section, we experimentally evaluate our methods for whole-graph SSPPR queries and top-$k$ SSPPR queries. For whole-graph (resp. top-$k$) SSPPR queries, we include our index free methods {\em \ssppr}, which includes the optimizations as mentioned in Section \ref{sec:opt} (resp. Section \ref{sec:topk-running-time-guarantee}),  and their index-based variant, referred to as {\em \ssppr+}, against the states of the art. All experiments are conducted on a Linux machine with an Intel 2.9GHz CPU and 200GB memory.

\begin{table}[t]
\centering
\renewcommand{\arraystretch}{1.1}
%\tbl{Datasets. ($K=10^3, M=10^6, B=10^9$\label{tbl:exp-data})}{
\caption{Datasets. ($K=10^3, M=10^6, B=10^9$)}
\label{tbl:exp-data}
\begin{tabular}{|l|r|r|r|c|}
    \hline
    {\bf Name} & \multicolumn{1}{c|}{$\boldsymbol{n}$} & \multicolumn{1}{c|}{$\boldsymbol{m}$} & \multicolumn{1}{c|}{\bf Type}  & \multicolumn{1}{c|}{\bf Linking Site}\\
%    \hline
    \hline
    {\em DBLP}   &613.6K    & 2.0M          & undirected & www.dblp.com\\% 1 \\
    \hline
    {\em Web-St}      & 281.9K    & 2.3M      & directed & www.stanford.edu\\
    \hline
    {\em Pokec}      & 1.6M    &30.6M      & directed & pokec.azet.sk \\
    \hline
    {\em LJ} & 4.8M & 69.0M    & directed &  www.livejournal.com \\
    \hline
    {\em Orkut}      & 3.1M    & 117.2M      & undirected  & www.orkut.com \\
    \hline
    {\em Twitter}   & 41.7M & 1.5B & directed & twitter.com\\
    \hline\hline
%    {\em RMAT-1}   & 1.6M    &30.6M      & directed & -\\
%    \hline
%    {\em RMAT-2}   & 3.1M    & 117.2M      & directed  & -\\
%    \hline
    {\em RMAT-1}   & 41.7M & 1.5B & directed & -\\
    \hline
    {\em RMAT-2}   & 128M & 8.6B & directed & -\\

    \hline\hline
    {\em X0}   & 26.1M & 485.6M  & undirected & tencent.com\\
    \hline
    {\em X1}   &50.1M & 792.0M & undirected & tencent.com\\
    \hline
    {\em X2}   &58.2M & 1.1B & undirected & tencent.com\\
    \hline
    {\em X3}   & 74.3M & 1.5B & undirected & tencent.com\\
  \hline

   %{\em X3}   & 219.6M & 5.6B & undirected & x.com\\
   %\hline
\end{tabular}%}
\end{table}

%\vspace{-2mm}
\subsection{Experimental Settings}

%\header
{\bf Datasets and query sets.} We use 6 real graphs: {\em DBLP}, {\em Web-St}, {\em Pokec},  {\em LJ}, {\em Orkut}, and {\em Twitter}, which are public benchmark datasets used in recent work \cite{lofgren2015personalized,WangTXYL16}. We further generate two synthetic datasets using the RMAT random graph generator \cite{ChakrabartiZF04}, denoted as {\em RMAT-1} and {\em RMAT-2} to examine the scalability of our proposed top-$k$ algorithms. Note that {\em RMAT-2} includes 8.6 billion edges. Moreover, we test our methods on four different game social networks from Tencent Games. Table \ref{tbl:exp-data} summarizes the statistics of the datasets. For each dataset, we choose $50$ source nodes uniformly at random, and we generate an SSPPR query from each chosen node. In addition, we also generate $5$ top-$k$ queries from each source node, with $k$ varying in $\{100, 200, 300, 400, 500\}$. Note that the maximum $k$ is set to $500$ in accordance to Twitter's Who-To-Follow service \cite{gupta2013wtf}, whose first step requires deriving top-$500$ PPR results.

\header
{\bf Methods.} For whole-graph SSPPR queries, we compare our proposed {\em \ssppr} and {\em \ssppr+} against three methods: (i) the {\em Monte-Carlo} approach, dubbed as {\em MC}; (ii) the optimized {\em BiPPR} for SSPPR queries described in Section \ref{app:ssbippr}; (iii) {\em HubPPR}, which is the indexed version of {\em BiPPR}.
We further compare our {\em FORA} and {\em FORA+} against the version without the optimization techniques (Ref. Section \ref{sec:opt}), dubbed as {\em FORA-Basic} and {\em FORA-Basic+} for the index-free and index-based methods, respectively.

For top-$k$ SSPPR queries, we compare our algorithm with the existing approximate solutions: the single-source {\em BiPPR}, the top-$k$ algorithm for {\em HubPPR} in \cite{WangTXYL16}, and the {\em TopPPR} \cite{WeiHX0SW18}. For {\em BiPPR} and {\em HubPPR}, we use the same $\epsilon$, $\delta$, and $\pf$ as FORA. For {\em TopPPR}, we follow the settings in \cite{WeiHX0SW18} and set $\rho=0.99$. We further extend our top-$k$ algorithm (Ref. Algorithm \ref{alg:layer-topk}) to the Monte-Carlo approach, and denote this algorithm as {\em MC-Topk}. We also include the {\em Forward Push} \cite{AndersenCL06} and {\em TPA} \cite{yoon2018tpa} as a baseline for top-$k$ SSPPR queries, and we tune their accuracy control parameters on each dataset separately, so that their precisions for top-$k$ PPR queries are the same as {\em \ssppr} on each dataset. Besides, we also compare our {\em FORA} and {\em FORA+} against the version without the top-$k$ optimization techniques (Ref. Section \ref{sec:topk-running-time-guarantee}), dubbed as {\em FORA-Basic} and {\em FORA-Basic+} for the index-free and index-based algorithms, respectively.

\header
{\bf Parameter setting.} Following previous work \cite{lofgren2014fast,lofgren2015personalized,WangTXYL16}, we set $\delta = 1/n, \pf = 1/n$, and $\eps=0.5$. %For {\em HubPPR}, the index size is set to the same as {\em \ssppr+}.
For our {\em FORA} and {\em FORA+}, note that the performance and / or the index size depends on the choice of $\rmax$. On the whole-graph queries, for {\em FORA-Basic} and {\em FORA-Basic+}, $\rmax$ is set according to Section \ref{sec:ana}; we then use the balanced strategy to auto-decide the choice of $\rmax$ for {\em FORA}; for {\em FORA+}, we include the optimization technique in Section \ref{sec:opt-zero-hop} and tune $\rmax$ varying from $\rmax^*$ to $7\rmax^*$ where $\rmax^*$ is the choice of $\rmax$ set according to Equation \ref{eqn:rmax}. We find that $\rmax=2\rmax^*$ strikes the best trade-off, and therefore use this setting for {\em FORA+} on the whole-graph queries.
%The results show that the index built based on $\rmax=2\rmax^*$ strikes a good trade-off, where $\rmax^*$ is the result set according to Section \ref{sec:ana}. We further tune $\rmax$ and set $\rmax$ from $\rmax^*$ to $7\rmax^*$ to if it indeed gain a good trade-off.
For top-$k$ PPR query, we also tune $\rmax$ to find the best index size for our top-$k$ queries and vary $\rmax$ from $\rmax^*$ to $7\rmax^*$. As we show in the experiment, when $\rmax$ is set to $\rmax^*$, it achieves the best trade-off, and therefore we use this setting in our top-$k$ evaluation.
For fair comparison, the index size of {\em HubPPR} is set to be the same as that of {\em \ssppr+} for top-$k$ SSPPR processing and also {\em \ssppr-Basic+} (for whole-graph SSPPR processing).

\subsection{Whole-Graph SSPPR Queries} \label{sec:exp-ss}

In our first set of experiments, we evaluate the efficiency of each method for whole-graph SSPPR queries on the 6 public datasets.
%We include four variants of our proposed methods. We denote our implementation without optimization techniques (Ref. Section \ref{sec:opt}) as {\em FORA-Basic} and {\em FORA-Basic+} for the index-free method and index-based method, respectively. For the versions with optimizations, we denote them as {\em FORA} and {\em FORA+} for the index-free and index-based methods, respectively.
Table \ref{tbl:ppr-query} reports the average query time of each method. Observe that both {\em \ssppr} and {\em BiPPR} achieve better query performance than {\em MC}, which is consistent with our analysis that the time complexity of {\em \ssppr} and {\em BiPPR} is better than that of {\em MC}. Moreover, {\em \ssppr} is at least $4$ times faster than {\em BiPPR} on most of the datasets. %For instance, on Twitter, {\em \ssppr} is around 3.7x faster than {\em BiPPR}.
The reason, as we explain in Section \ref{sec:ana}, is that {\em BiPPR} either degrades to the {\em MC} approach when the backward threshold is large, or requires conducting a backward search from each node $v$ in $G$, even if $\pi(s, v)$ is extremely small. In contrast, {\em \ssppr} avoids degrading to {\em MC} and tends to omit nodes with small PPR values, which helps improve efficiency. In addition, {\em \ssppr+} achieves significant speedup over {\em \ssppr}, and is around 10 times faster than the latter on most of the datasets. %This is due to our advanced design of the index structure according to our theoretical analysis in Section \ref{sec:ana}.
The {\em HubPPR} also improves over {\em BiPPR}, but the improvement is far less than what {\em \ssppr+} achieves over {\em \ssppr}. %The reason is that, for single-source PPR queries, the number of required random walks is much larger than that in point-to-point PPR queries. Hence, when the stored random walks are consumed on a hub, and it needs to do random walks online.
Moreover, even without any index, {\em \ssppr} is still more efficient than {\em HubPPR}. %which indicates that the proposed {\em \ssppr} is more favorable choice on single-source PPR queries than other alternatives.

As we can observe from Table \ref{tbl:ppr-query}, with our optimization techniques introduced in Section \ref{sec:opt}, the index-free method {\em FORA} improves over {\em FORA-Basic} by up to 1.8x. Apart from the 6 public datasets, we further test the effectiveness of our methods on the four social networks from company X. We omit the results for the baseline methods (MC, BiPPR, and HubPPR) since they incur prohibitive processing costs.  As shown in Table \ref{tbl:fora-ss-opt}, {\em FORA} still improves over {\em FORA-Basic} by more than 2x almost on all datasets, which demonstrates the effectiveness of our proposed optimization technique for online algorithms. For the index-based method, {\em FORA} improves over {\em FORA+} by at least twice on almost all datasets, and up to 2.8x. As shown in Table \ref{tbl:exp-space-consumption},  the space consumption required by {\em FORA+} is twice as that of {\em FORA-Basic+}, which demonstrates that our optimization technique achieves a good trade-off between query time and space consumption.

\begin{table}[!t]
\centering
\caption{Whole-graph SSPPR performance (s) (i). ($K=10^3$)}
\label{tbl:ppr-query}
  \begin{tabular}{ | l | r | r|r|r |r|r|r|}
    \hline
        & {\em MC} & {\em BiPPR} & {\em HubPPR} & {\em FORA-Basic} & {\em FORA} & {\em FORA-Basic+}&   {\em \ssppr+} \\
    \hline
 {\em DBLP}  &14.2 &	3.8&  	2.8&	0.8& 0.6	& 0.09 & 0.05 \\
    \hline
{\em Web-St}  & 5.4	&3.7&  	1.6	& 0.03& 0.02	&  0.01& 0.01 \\
    \hline
 {\em Pokec} &69.1	&24.9& 19.6& 11.26&	6.3 & 0.9 & 0.4	 \\
    \hline
 {\em LJ}&163.5	& 61.4& 50.8	& 15.5 & 9.9	&1.2 & 0.6 \\
    \hline
{\em Orkut} &230.6	& 158.2 &126.3	&  40.1& 26.4	& 4.8	& 1.7 \\
    \hline
{\em Twitter}   &4.3K	&3.1K& 	 2.4K &  513.8& 283.1	&63.3 &	 29.8\\
    \hline
  \end{tabular}%}
\end{table}

\begin{table}[!t]
\centering
%\tbl{Whole-graph SSPPR performance (s) (ii). ($K=10^3$) \label{tbl:fora-ss-opt}}{
\caption{Whole-graph SSPPR performance (s) (ii). ($K=10^3$)}
\label{tbl:fora-ss-opt}
  \begin{tabular}{ | l | r | r|r|r|}
    \hline
        & {\em FORA-Basic} & {\em FORA} & {\em FORA-Basic+}&   {\em \ssppr+}\\
    \hline
{\em X0}   & 527.47  & 259.9	&66.7 &	31.2 \\
    \hline
{\em X1}   & 957.4 &  504.9 	&130.7 &	60.4 \\
    \hline
{\em X2}   & 1072.5 & 530.4	& 152.4 &	63.7 \\
    \hline
{\em X3}   & 2340.6 &  1163.3 	&216.8 & 97.2	 \\
    \hline
  \end{tabular}%}
\end{table}

\begin{figure*}[!t]
 \centering
   %\vspace{-2mm}
    \begin{small}
    \begin{tabular}{cc}
        %\multicolumn{2}{c}{\hspace{-2mm} \includegraphics[height=3mm]{./figure/new/algo-legend.eps}}  \\[0mm]
      	\multicolumn{2}{c}{\hspace{-8mm} \includegraphics[height=10mm]{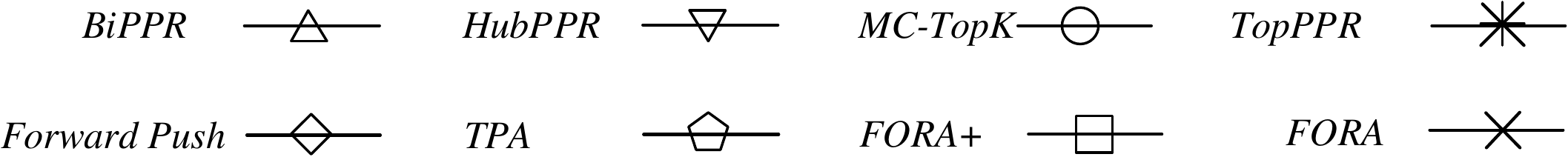}}  \\
      	
        %\hspace{-2mm} \includegraphics[height=35mm]{./figure/new_fora/dblp-topk-time.eps} &
        %\hspace{-2mm} \includegraphics[height=35mm]{./figure/new_fora/pokec-topk-time.eps} \\
        \hspace{-2mm} \includegraphics[height=35mm]{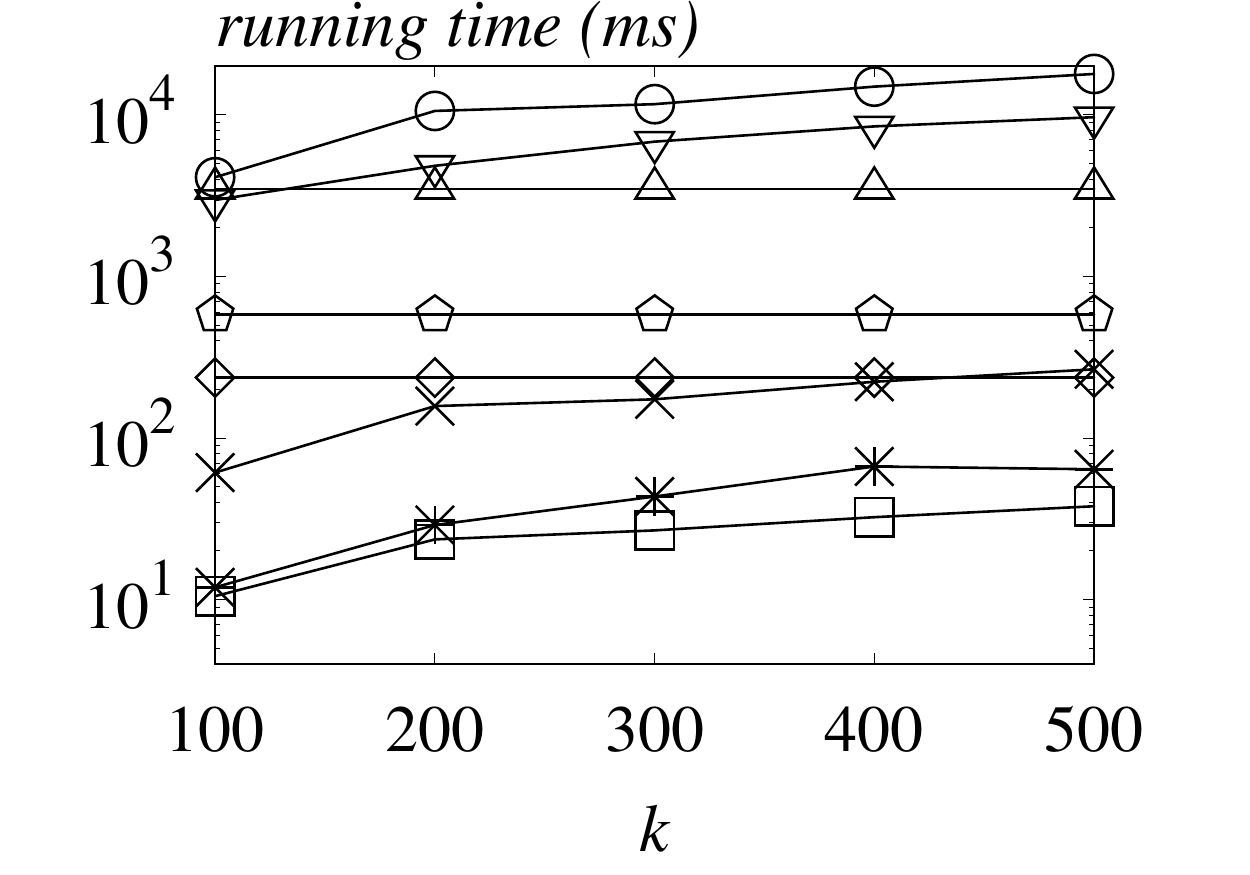} &
        \hspace{-2mm} \includegraphics[height=35mm]{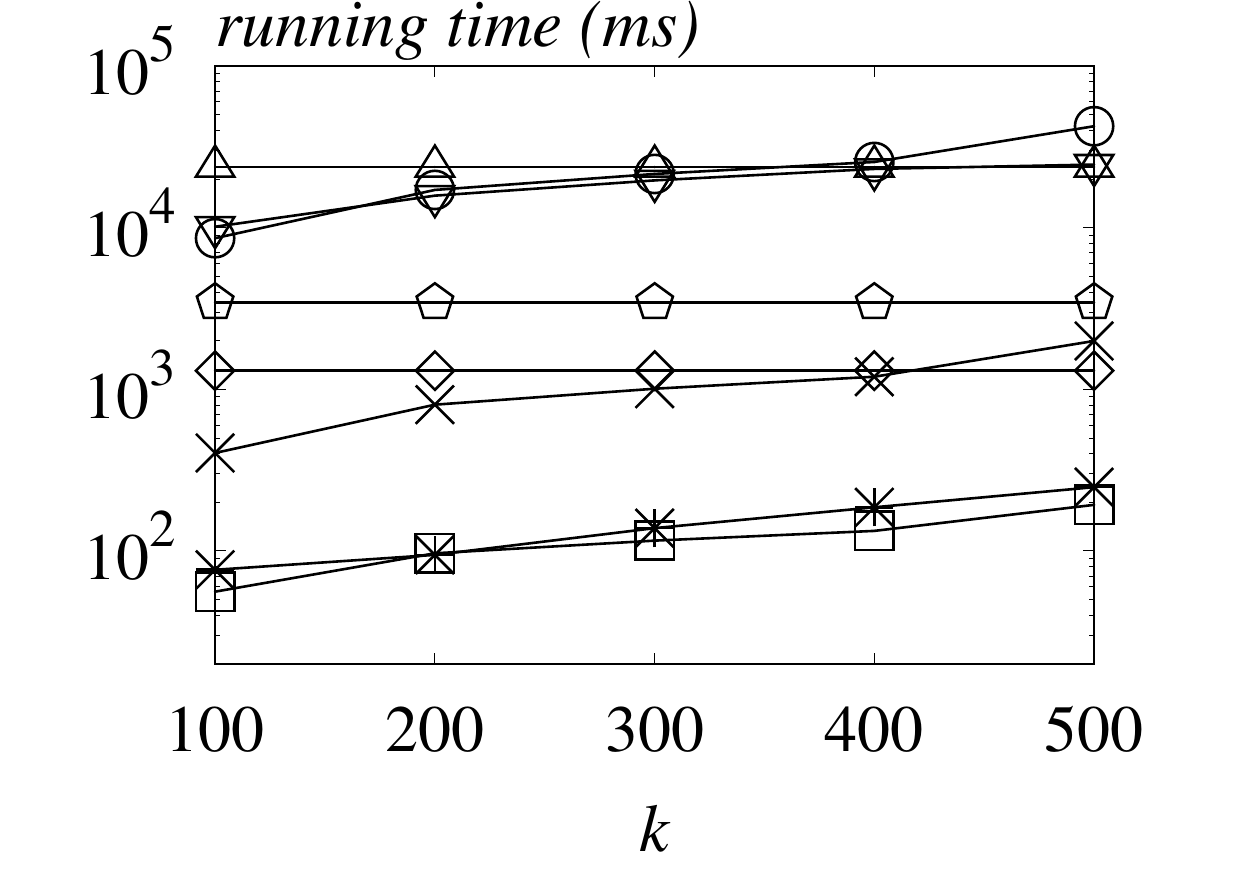} \\
        \hspace{-2mm} (a) DBLP  &
        \hspace{-2mm} (b) Pokec \\%[-1mm]

        %\hspace{-2mm} \includegraphics[height=35mm]{./figure/new_fora/orkut-topk-time.eps} &
        %\hspace{-2mm} \includegraphics[height=35mm]{./figure/new_fora/tw-topk-time.eps} \\
        \hspace{-2mm} \includegraphics[height=35mm]{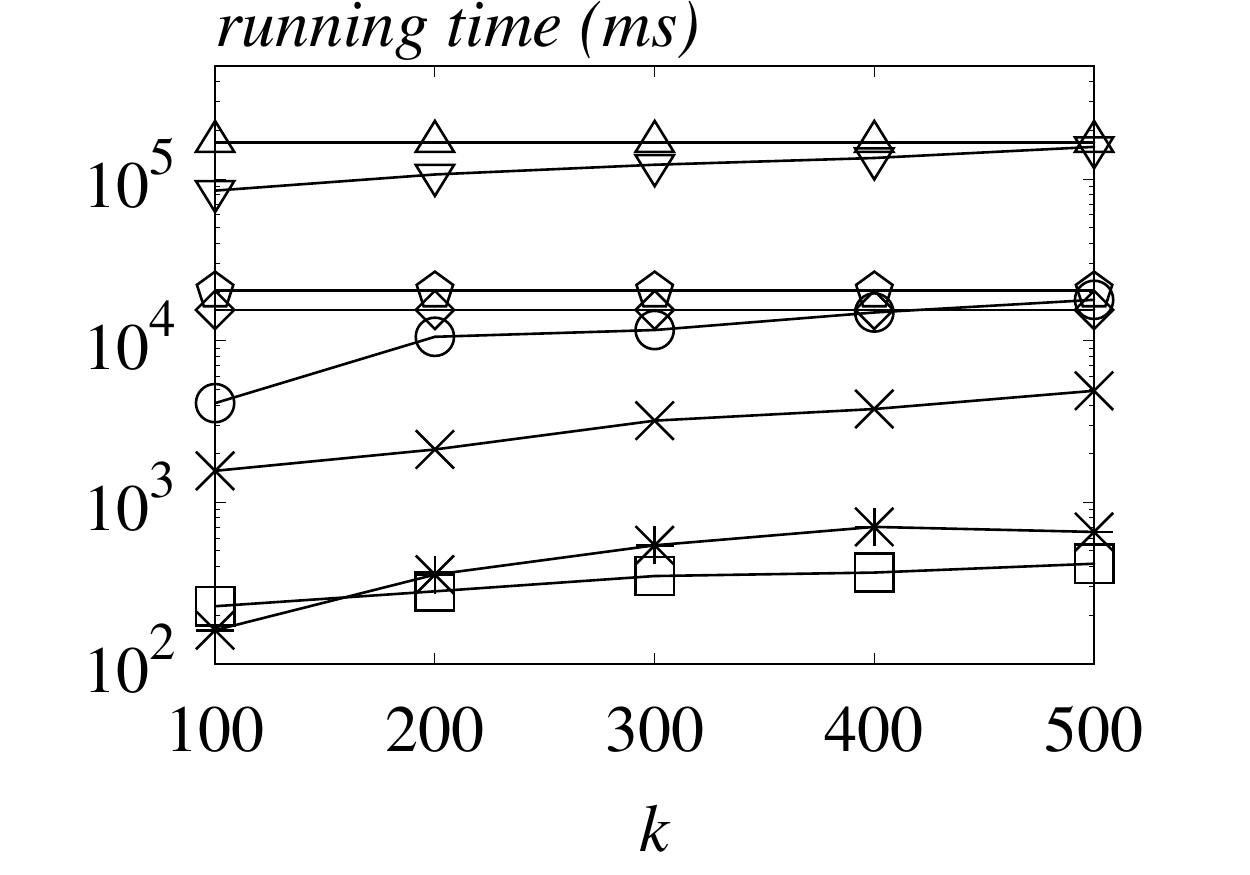} &
        \hspace{-2mm} \includegraphics[height=35mm]{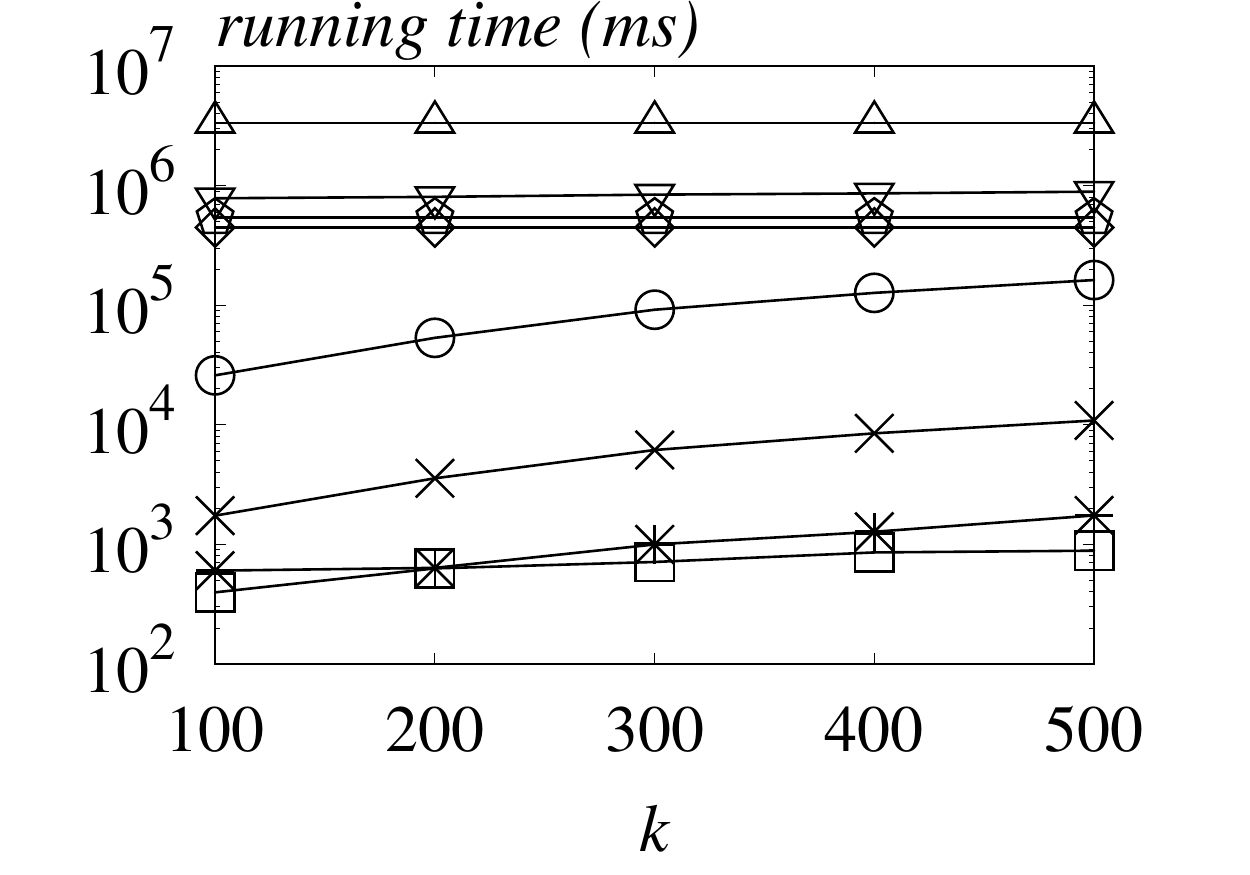} \\
        \hspace{-2mm} (c) Orkut &
        \hspace{-2mm} (d) Twitter \\%[-1mm]

 \end{tabular}
 \caption{Top-$k$ SSPPR query efficiency: varying $k$.} \label{exp:topk-time}
\end{small}
\end{figure*}
\begin{figure*}[!t]
 \centering
   %\vspace{-2mm}
\begin{small}
    \begin{tabular}{cc}
        %\multicolumn{2}{c}{\hspace{-2mm} \includegraphics[height=3mm]{./figure/new/algo-legend.eps}}  \\[0mm]
        \multicolumn{2}{c}{\hspace{-8mm} \includegraphics[height=10mm]{./figure/Revision-new-figures/topk_query_legend-eps-converted-to.pdf}}  \\

%        \hspace{-2mm} \includegraphics[height=35mm]{./figure/new/dblp-topk-prec.eps} &
%        \hspace{-2mm} \includegraphics[height=35mm]{./figure/new/pokec-topk-prec.eps}\\
        \hspace{-2mm} \includegraphics[height=35mm]{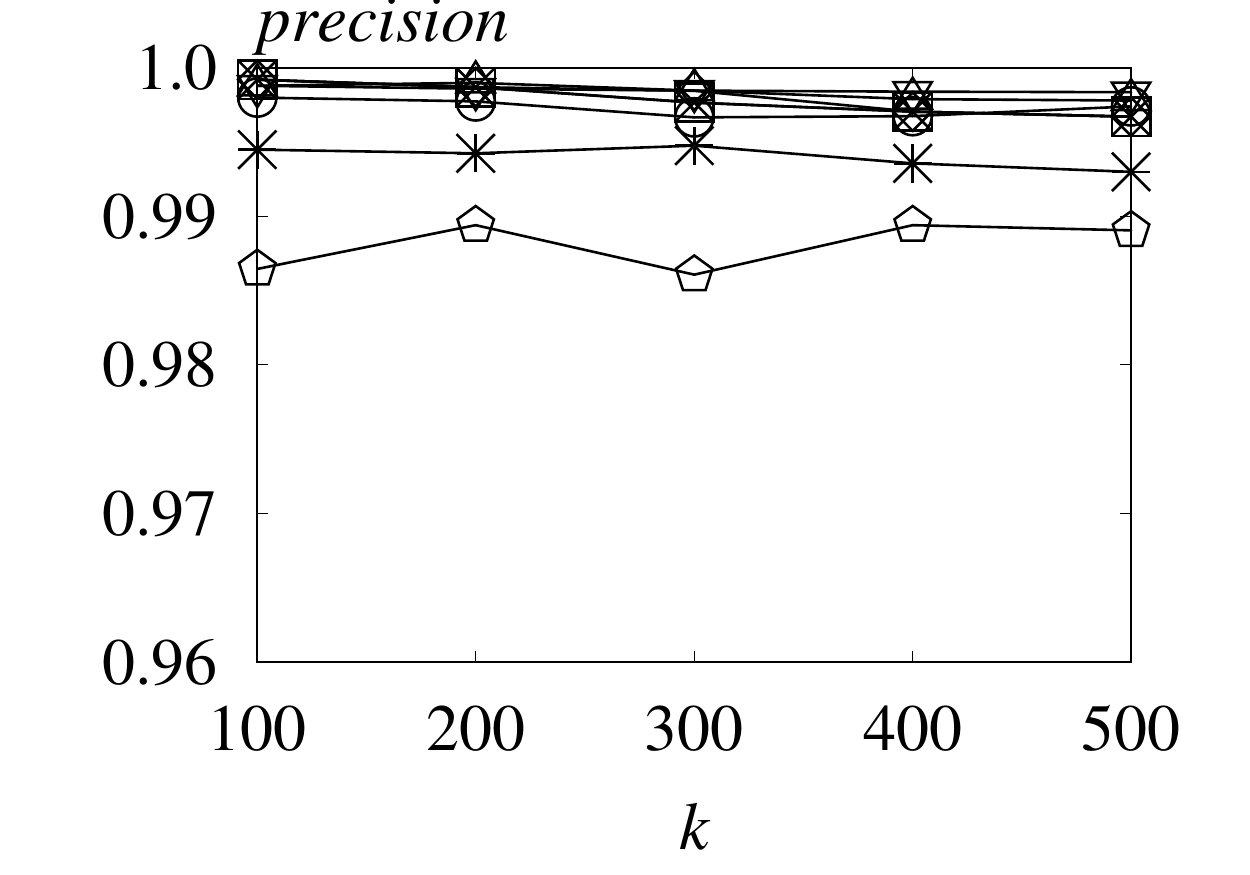} &
        \hspace{-2mm} \includegraphics[height=35mm]{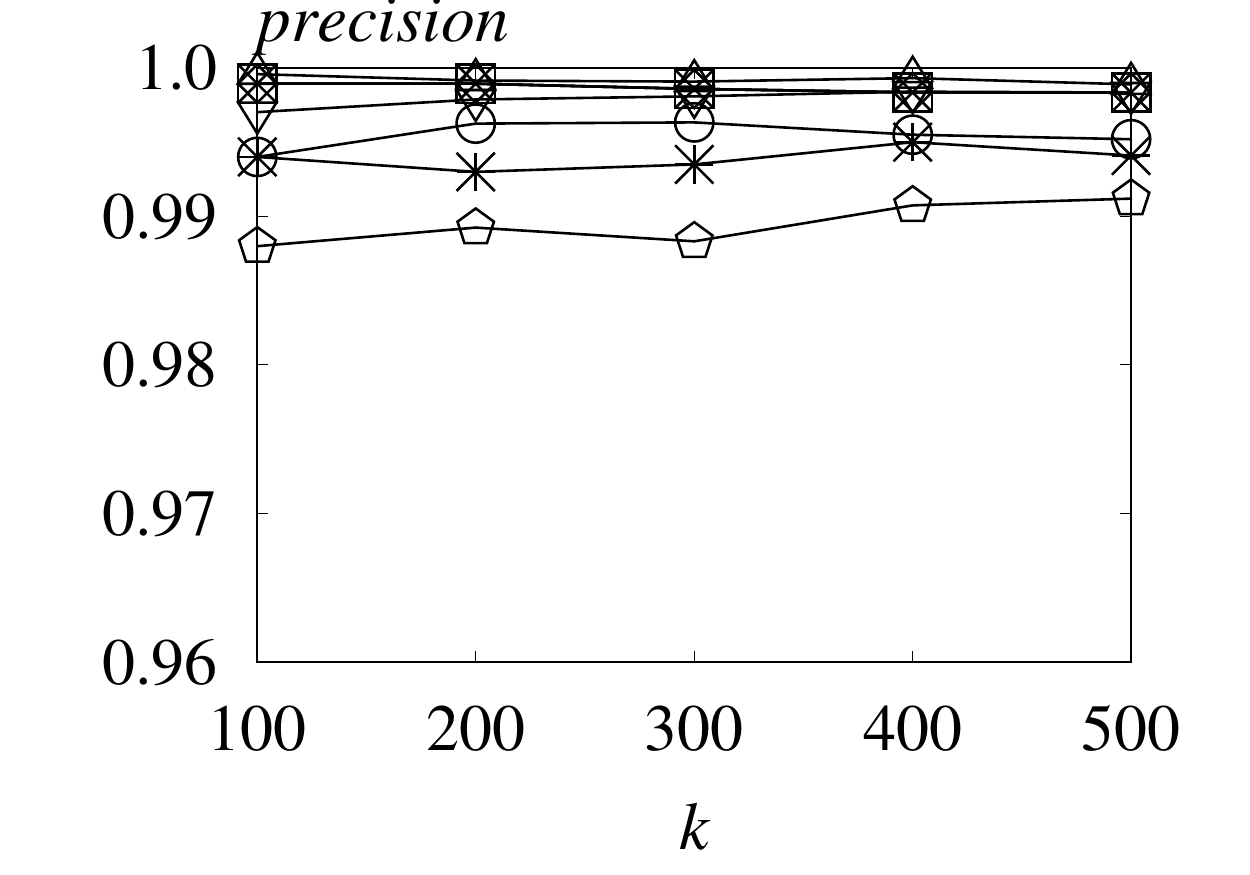}\\
        \hspace{-2mm} (a) DBLP  &
        \hspace{-2mm} (b) Pokec \\
%        \hspace{-2mm} \includegraphics[height=35mm]{./figure/new/orkut-topk-prec.eps} &
%        \hspace{-2mm} \includegraphics[height=35mm]{./figure/new/tw-topk-prec.eps}\\
	 	\hspace{-2mm} \includegraphics[height=35mm]{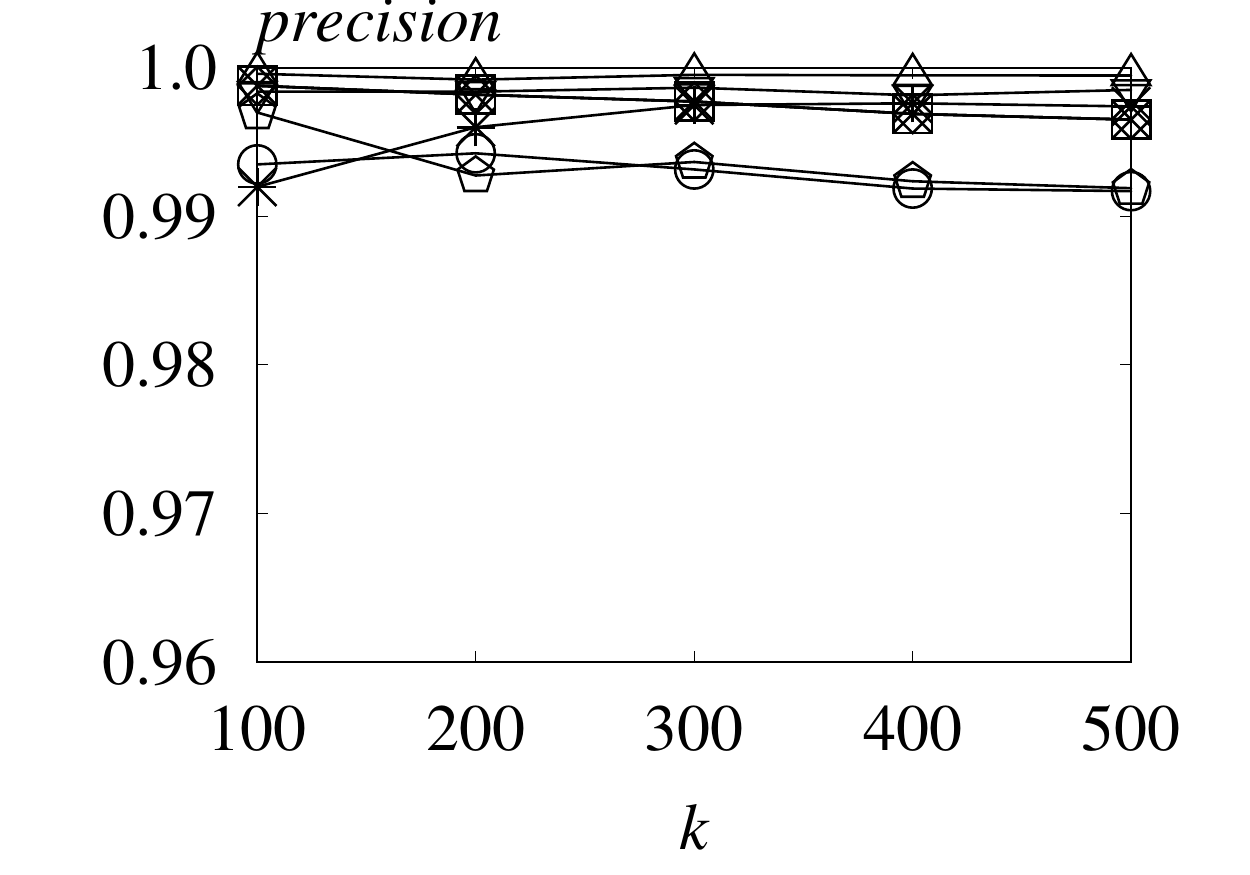} &
		\hspace{-2mm} \includegraphics[height=35mm]{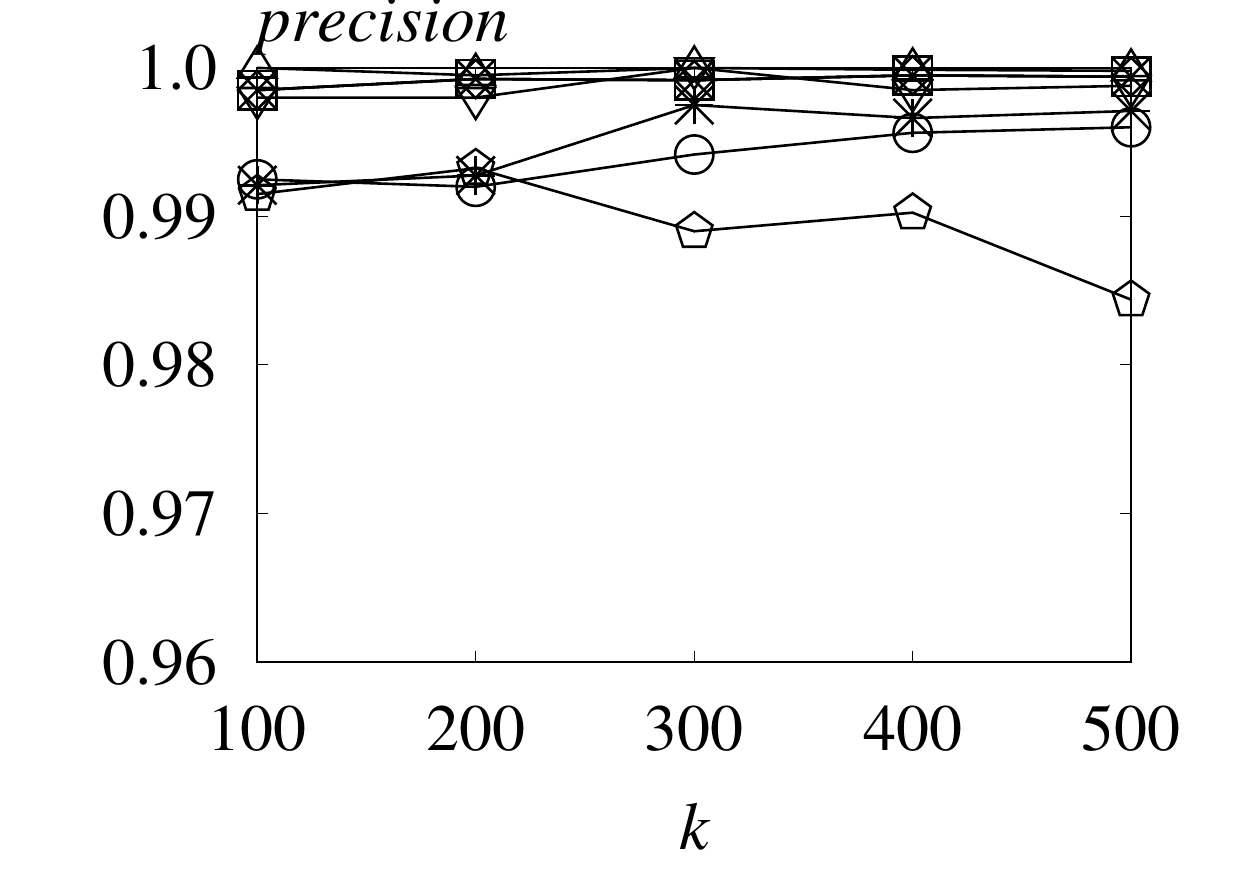}\\
        \hspace{-2mm} (c) Orkut &
        \hspace{-2mm} (d) Twitter \\
 \end{tabular}
 \caption{Top-$k$ SSPPR query accuracy: varying $k$.} \label{exp:topk-acc}
\end{small}
\end{figure*}

\begin{figure*}[!t]
	\centering
	%\vspace{-2mm}
	\begin{small}
		\begin{tabular}{cc}
			%\multicolumn{2}{c}{\hspace{-2mm} \includegraphics[height=3mm]{./figure/new/algo-legend.eps}}  \\[0mm]
			\multicolumn{2}{c}{\hspace{-8mm} \includegraphics[height=10mm]{./figure/Revision-new-figures/topk_query_legend-eps-converted-to.pdf}}  \\

			%        \hspace{-2mm} \includegraphics[height=35mm]{./figure/new/dblp-topk-prec.eps} &
			%        \hspace{-2mm} \includegraphics[height=35mm]{./figure/new/pokec-topk-prec.eps}\\
			\hspace{-2mm} \includegraphics[height=35mm]{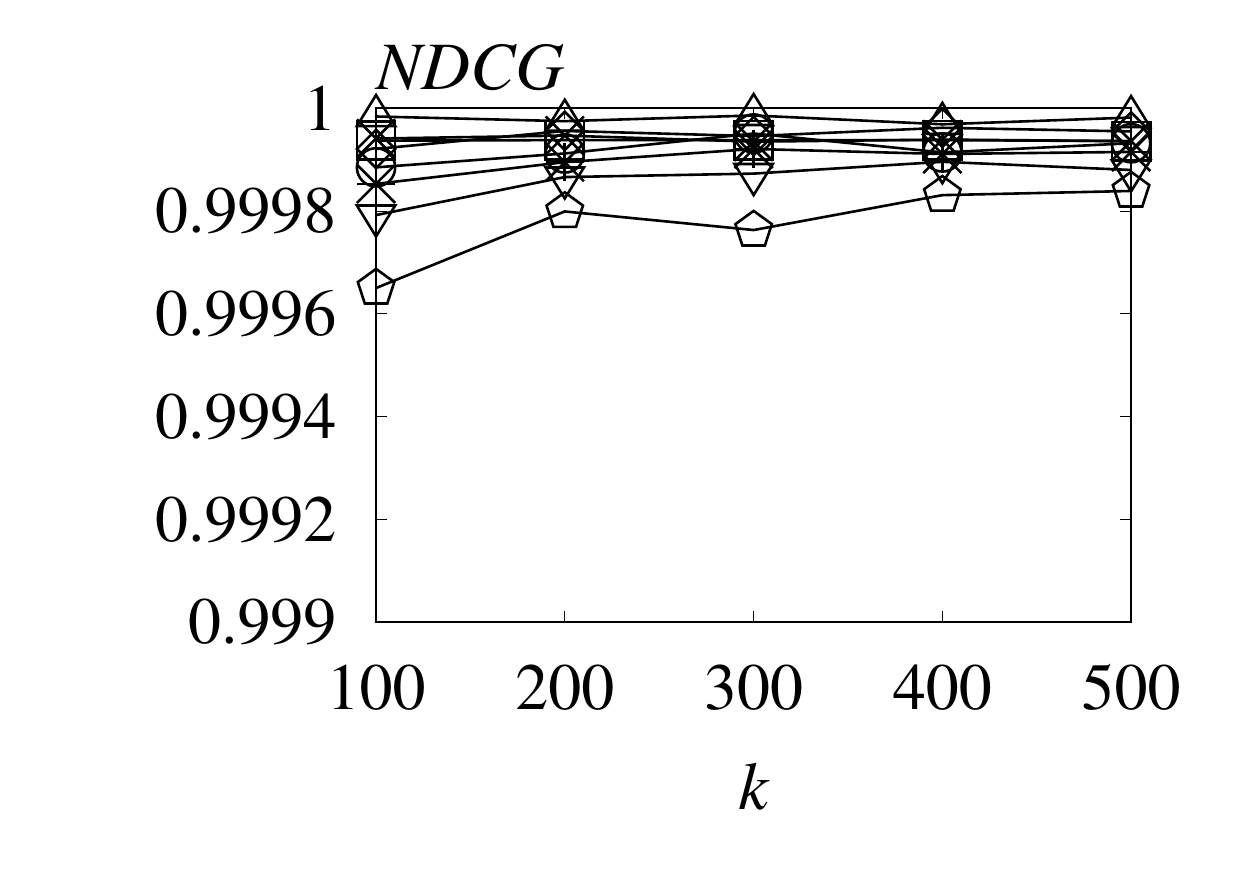} &
			\hspace{-2mm} \includegraphics[height=35mm]{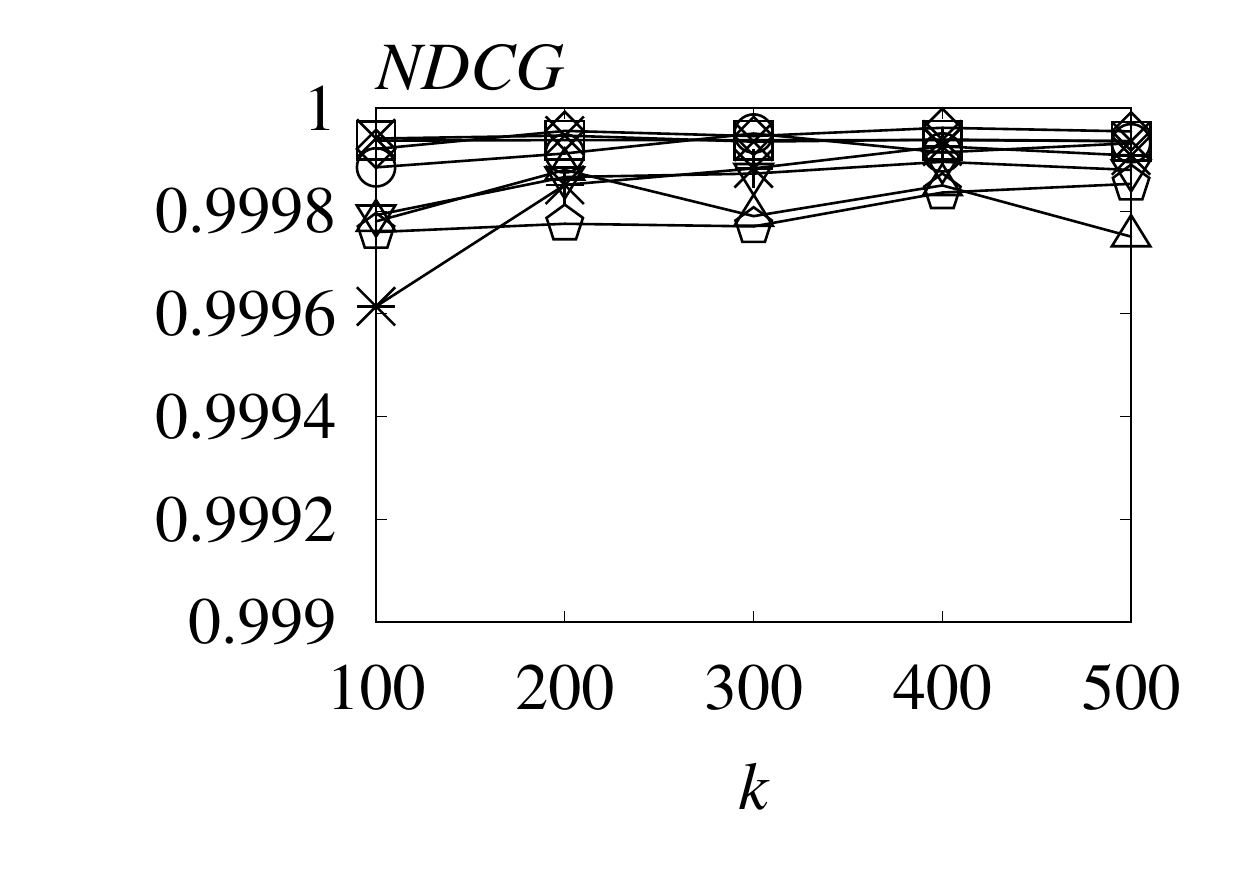}\\
			\hspace{-2mm} (a) DBLP  &
			\hspace{-2mm} (b) Pokec \\
			%        \hspace{-2mm} \includegraphics[height=35mm]{./figure/new/orkut-topk-prec.eps} &
			%        \hspace{-2mm} \includegraphics[height=35mm]{./figure/new/tw-topk-prec.eps}\\
			\hspace{-2mm} \includegraphics[height=35mm]{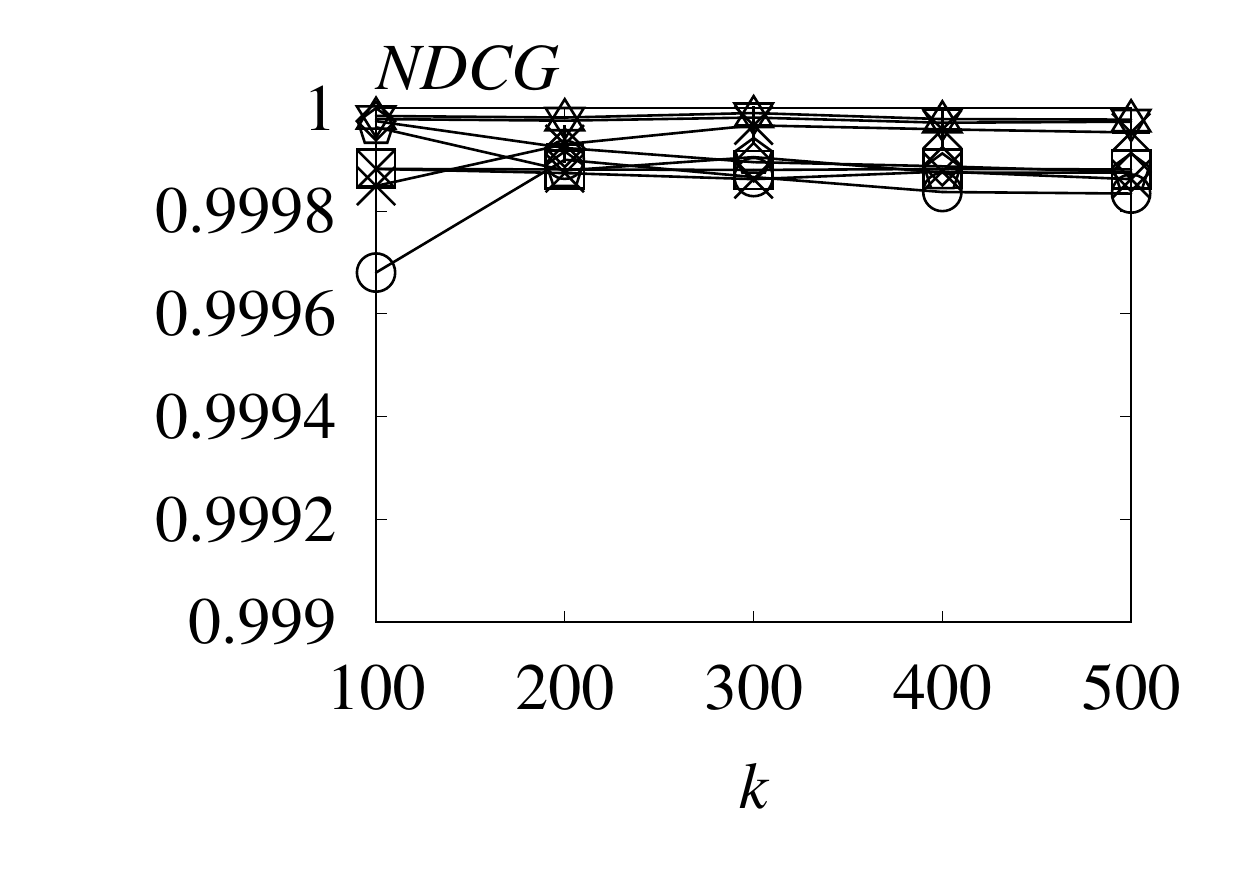} &
			\hspace{-2mm} \includegraphics[height=35mm]{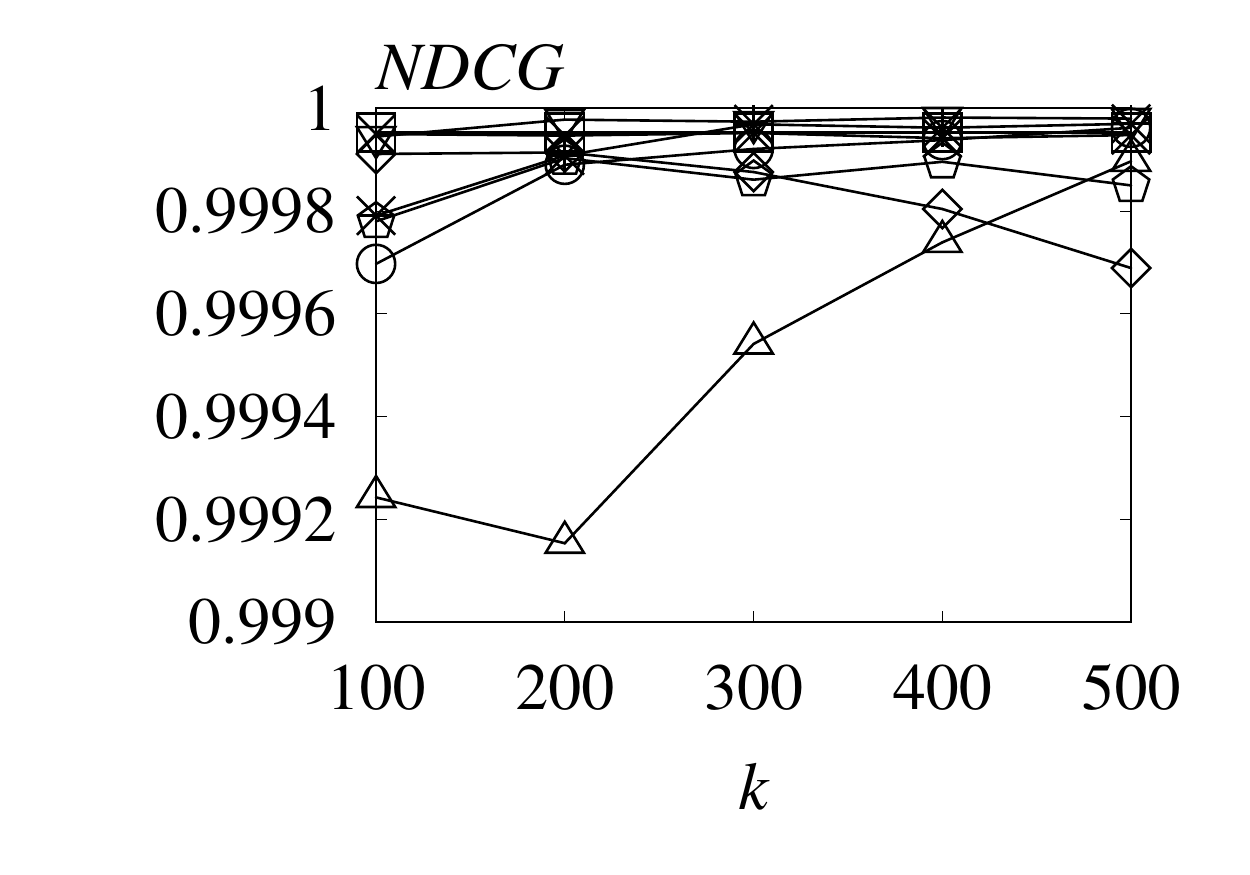}\\
			\hspace{-2mm} (c) Orkut &
			\hspace{-2mm} (d) Twitter \\
		\end{tabular}
		\caption{Top-$k$ SSPPR query NDCG: varying $k$.} \label{exp:topk-ndcg}
	\end{small}
\end{figure*}

%\vspace{-1mm}
\subsection{Top-$\boldsymbol k$ SSPPR Queries} \label{sec:exp-topk}
In our second set of experiments, we evaluate the efficiency and accuracy of each method for top-$k$ SSPPR queries. For our methods, we only include our {\em FORA} and {\em FORA+} that includes the optimizations presented in Section \ref{sec:topk-running-time-guarantee} to avoid the figures being too crowded. We will examine the effectivess of our optimization techniques in the next set of experiments.

{\cblue
\subsubsection{Top-$\boldsymbol k$ query efficiency}

Figures \ref{exp:topk-time} reports the average query time of each method on four representative datasets: {\em DBLP}, {\em Pokec}, {\em Orkut}, and {\em Twitter}. (The results on the other two datasets are qualitatively similar, and are omitted due to the space constraint.) Note that the y-axis is in log-scale.
Recall that {\em Forward Push} provides no approximation guarantee, and we tune the $\rmax$ on each dataset separately, so that it provides the same precision for top-500 SSPPR queries as our {\em \ssppr} algorithm does. Similarly, for TPA \cite{yoon2018tpa}, we tune their parameters so that it provides the best possible precision.

The main observation is that our {\em \ssppr+} achieves the best performance on all the datasets. For instance, on Twitter dataset with $k=500$, our {\em \ssppr+} provides similar or better accuracy and NDCG as competitors, and runs 2x faster than TopPPR, two order of magnitude faster than MC-Topk, and more than 500x faster than Forward Push, TPA, BiPPR, and HubPPR.
This is expected since our FORA+ applies an iterative approach to refine the top-$k$ answers and terminates immediately whenever the answer could provide the desired approximation guarantee; it further uses the indexing scheme to reduce the expensive costs of random walks.

Our FORA is the fastest online algorithm except TopPPR since TopPPR uses a more advanced filter-refinement paradigm to answer the top-$k$ queries. However, their proposed approach does not benefit from indexing scheme, and is outperformed by our FORA+. It is also difficult for TopPPR to benefit from indexing scheme since {\em (i)} their random walk adopts the $\sqrt{\alpha}$-random walk, and all the nodes visited will be used to estimate the PPR scores; {\em (ii)} the sources of random walks in Monte-Carlo phase are generated randomly in TopPPR and it may results in poor cache performances. In contrast, FORA+ only need to scan the index structures in order and only stores the destinations in index structure, making it light-weighted and cache-friendly.

Another observation is that after extending the MC approach with our top-$k$ algorithm, MC-Topk achieves more than 20x speedup on Twitter over its single-source alternative. This further demonstrates the effectiveness of our top-$k$ algorithm proposed in Section \ref{sec:topk-running-time-guarantee}, which improves the time complexity to $O(\frac{\log{n}}{\epsilon^2\cdot \pi(s,v_k^*)})$. However, our FORA and FORA+ are still far more efficient than MC-Topk since FORA and FORA+ explores the forward push to reduce the random walk costs.

Finally, with our index structure, {\em \ssppr+} further improves over {\em \ssppr} by an order of magnitude, which demonstrates the effectiveness of the index structure. Notably, on the Twitter social network with 1.5 billion edges, our {\em FORA+} can answer the top-500 query in 0.8 second.

\subsubsection{Top-$\boldsymbol k$ query accuracy}

To compare the accuracy of the top-$k$ results returned by each method, we first calculate the ground-truth answer of the top-$k$ queries using the {\em Power Iteration} \cite{page1999pagerank} method with 100 iterations. Afterwards, we evaluate the top-$k$ results of each algorithm by their precision and NDCG \cite{JarvelinK00} with respect to the ground truth.
Note that the precision and recall are the same for the top-$k$ SSPPR queries, and is the fraction of nodes returned by the top-$k$ algorithm that are real top-$k$ nodes.  For NDCG, let $s$ be the query node, $v_1,v_2,\cdots, v_k$ be the $k$ nodes returned by the top-$k$ algorithm, and $v_1^*, v_2^*,\cdots, v_k^*$ be the true top-$k$ nodes. Then, the NDCG is defined as $\frac{1}{Z_k}\sum_{i=1}^k{\frac{2^{\pi(s,v_i)}-1}{\log{(i+1)}}}$, where $Z_k=\sum_{i=1}^k{\frac{2^{\pi(s,v_k^*)-1}}{\log{(i+1)}}}$.

Figure~\ref{exp:topk-acc} (resp. Figure~\ref{exp:topk-ndcg}) show the accuracy (resp. NDCG) of the top-$k$ query algorithms on four datasets: {\em DBLP}, {\em Pokec}, {\em Orkut}, and {\em Twitter}. Observe that all methods except TPA consistently provide high precisions. In the meantime, notice that our FORA+ consistently provides similar precision as TopPPR and sometimes even slightly better precision than TopPPR on all the tested datasets. In terms of NDCG, all methods achieve very high NDCG scores, above 0.999 on all datasets.

\subsubsection{Scalability of Top-$\boldsymbol k$ algorithms}
In this section, we examine the scalability of our FORA+ using synthetic datasets. We first generate a synthetic dataset with the same size as Twitter and then generate a graph with 8.6 billion edges. The results are reported in Table \ref{tbl:rmat-scalability}. As we can see, our FORA+ achieves 3x improvement over TopPPR on both datasets while providing the same accuracy and NDCG. As we will see in Section \ref{sec:exp-preprocessing}, the space consumption of FORA+ is only 2x and 1.5x that of TopPPR on RMAT-1 and RMAT-2, respectively. This demonstrates that FORA+ achieves a better trade-off between query efficiency and space consumption when providing the same accuracy for the top-$k$ queries.
}

\subsection{Effectiveness of the Top-$\boldsymbol k$ optimization}
In this set of experiment, we evaluate the effectiveness of the new top-$k$ algorithm proposed in Section \ref{sec:topk-running-time-guarantee}, which provides improved time complexity, and reduces the time to calculate the bounds for each node. We present the results for 6 representative datasets: {\em DBLP}, {\em Pokec}, {\em Orkut}, {\em Twitter}, {\em X2}, and {\em X3}.

Figures \ref{exp:top-opt-time}(a)-(f) demonstrate the running time of {\em FORA} and {\em FORA+}, against the versions without the optimization techniques in Section \ref{sec:topk-running-time-guarantee}, referred to as {\em FORA-Basic} and {\em FORA-Basic+} for the index-free and index-based solution, respectively. As we can observe, with the new algorithm, {\em FORA} (resp. {\em FORA+} improves over {\em FORA-Basic} (resp. {\em FORA-Basic+}) by a large margin. For instance, on Twitter dataset, {\em FORA} (resp. {\em FORA+}) improves over {\em FORA-Basic} (resp. {\em FORA-Basic+}) by around 4x (resp. 6x). The main reason for the significant improvement are two-fold: (i) the time complexity of {\em FORA} and {\em FORA+} depend on $\frac{1}{\pi(s,v^*_k)}$ while {\em FORA-Basic} and {\em FORA-Basic+} depend on $\frac{1}{n}$; (ii) {\em FORA} and {\em FORA+} avoids the expensive bound calculation part that are required by {\em FORA-Basic} and {\em FORA-Basic+}.

Next, we report the accuracy of the four methods on the 6 datasets as shown in Figures \ref{exp:top-opt-acc}(a)-(f). As we can observe, all four methods provide similarly high accuracy for the top-$k$ queries on all datasets. To explain, all four methods share the similar spirit by adaptively refine the top-$k$ answer and return the approximate answer with theoretical guarantees. Therefore, all the four methods provide similarly high accuracy.

\begin{table}[!t]
	\centering
	%\tbl{Whole-graph SSPPR performance (s) (ii). ($K=10^3$) \label{tbl:fora-ss-opt}}{
	\caption{Scalability test on synthetic datasets. ($K=500$)}
	\label{tbl:rmat-scalability}
	\begin{tabular}{ | l | r | r|r|r|r|r|}
		\hline
		\multirow{2}{*}{}&  \multicolumn{2}{c|}{\bf Query Time} & \multicolumn{2}{c|}{\bf Precision} & \multicolumn{2}{c|}{\bf NDCG}\\
		\cline{2-7}
        & {RMAT-1} & {RMAT-2} & {RMAT-1} & {RMAT-2} & {RMAT-1} & {RMAT-2} \\ \hline
		{\em FORA+}   & 17.4 & 149.5	& 0.993 &	0.995 &0.9999 & 0.9999 \\
		\hline
		{\em TopPPR}   &  51.1&  481.0	& 0.993 & 0.995	 & 0.9999 & 0.9999 \\
		\hline
	\end{tabular}%}
\end{table}

\begin{figure*}[!t]
 \centering
   %\vspace{-2mm}
    \begin{small}
    \begin{tabular}{cc}
          \multicolumn{2}{c}{\hspace{-2mm} \includegraphics[height=3mm]{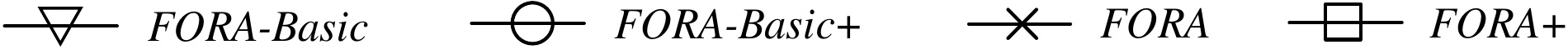}}  \\[0mm]
      %\multicolumn{4}{c}{\hspace{-8mm} \includegraphics[height=2.8mm]{./figs/topk_query_legend.eps}}  \\[-2mm]
        \hspace{-2mm} \includegraphics[height=35mm]{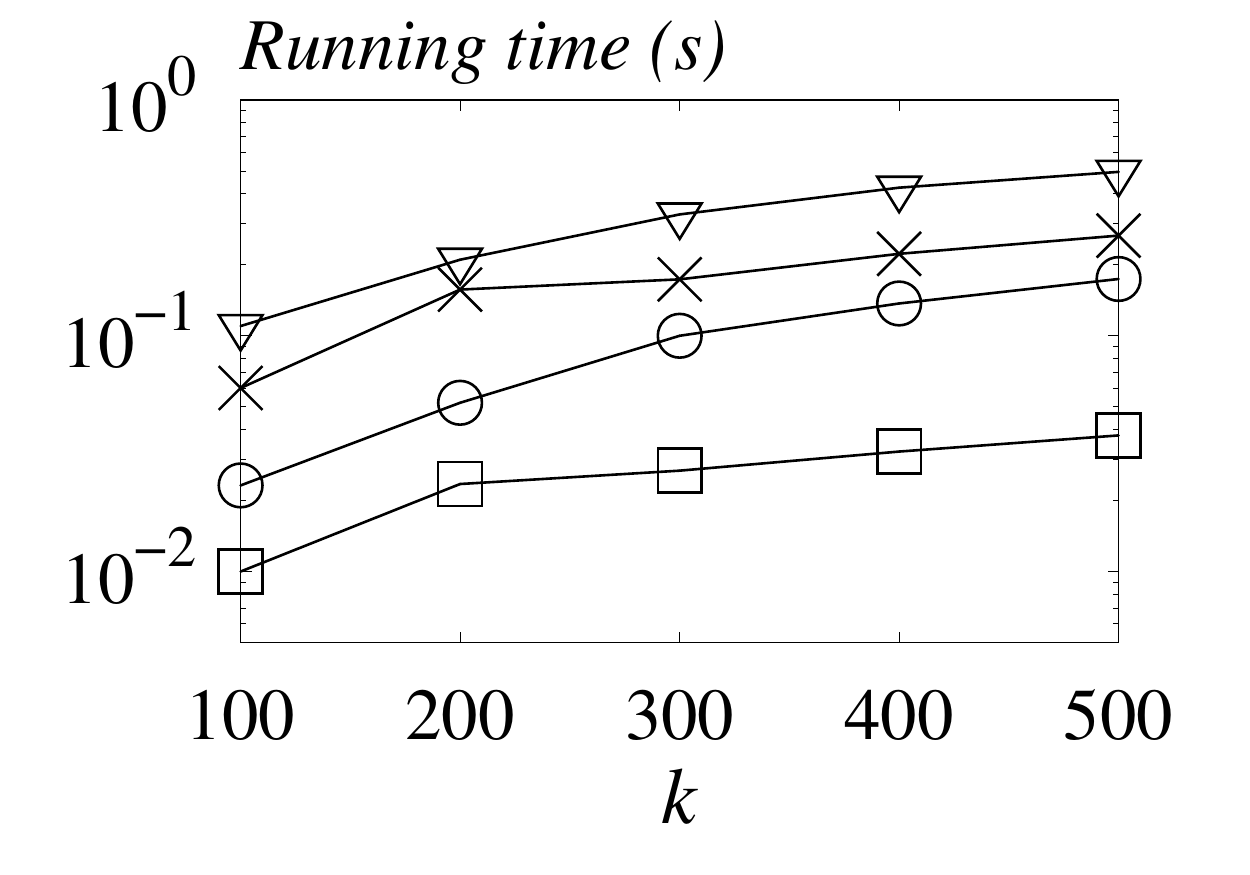} &
        \hspace{-2mm} \includegraphics[height=35mm]{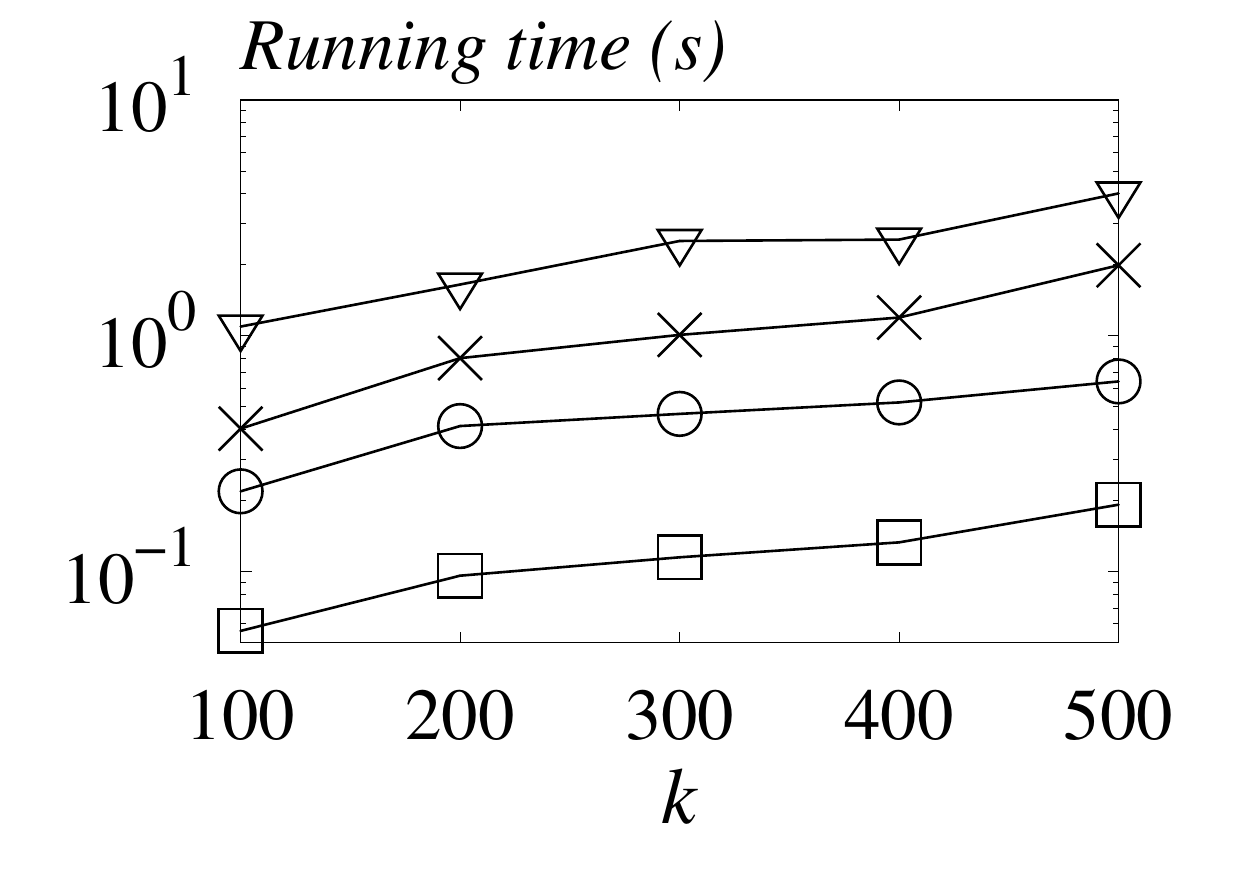} \\
        \hspace{-2mm} (a) DBLP  &
        \hspace{-2mm} (b) Pokec \\%[-1mm]
        \hspace{-2mm} \includegraphics[height=35mm]{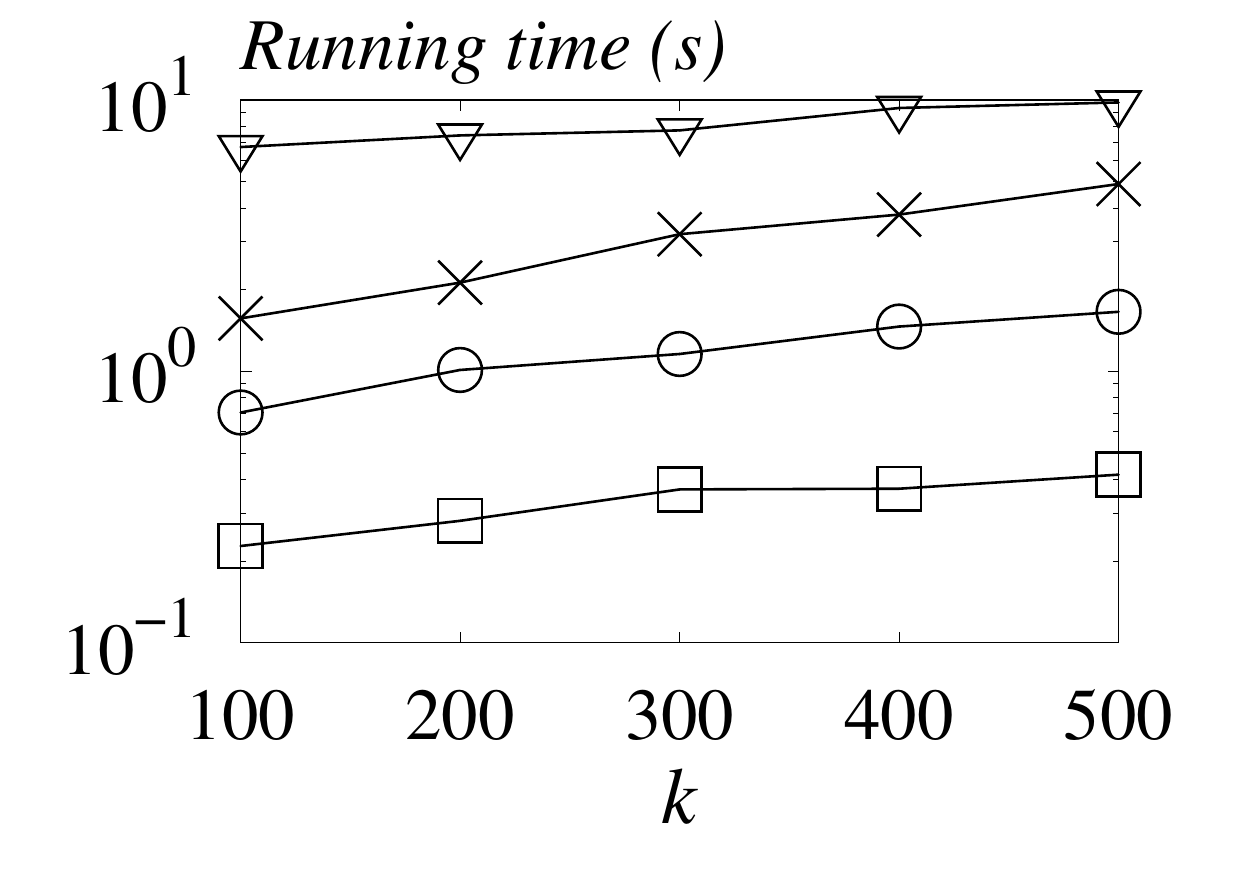} &
        \hspace{-2mm} \includegraphics[height=35mm]{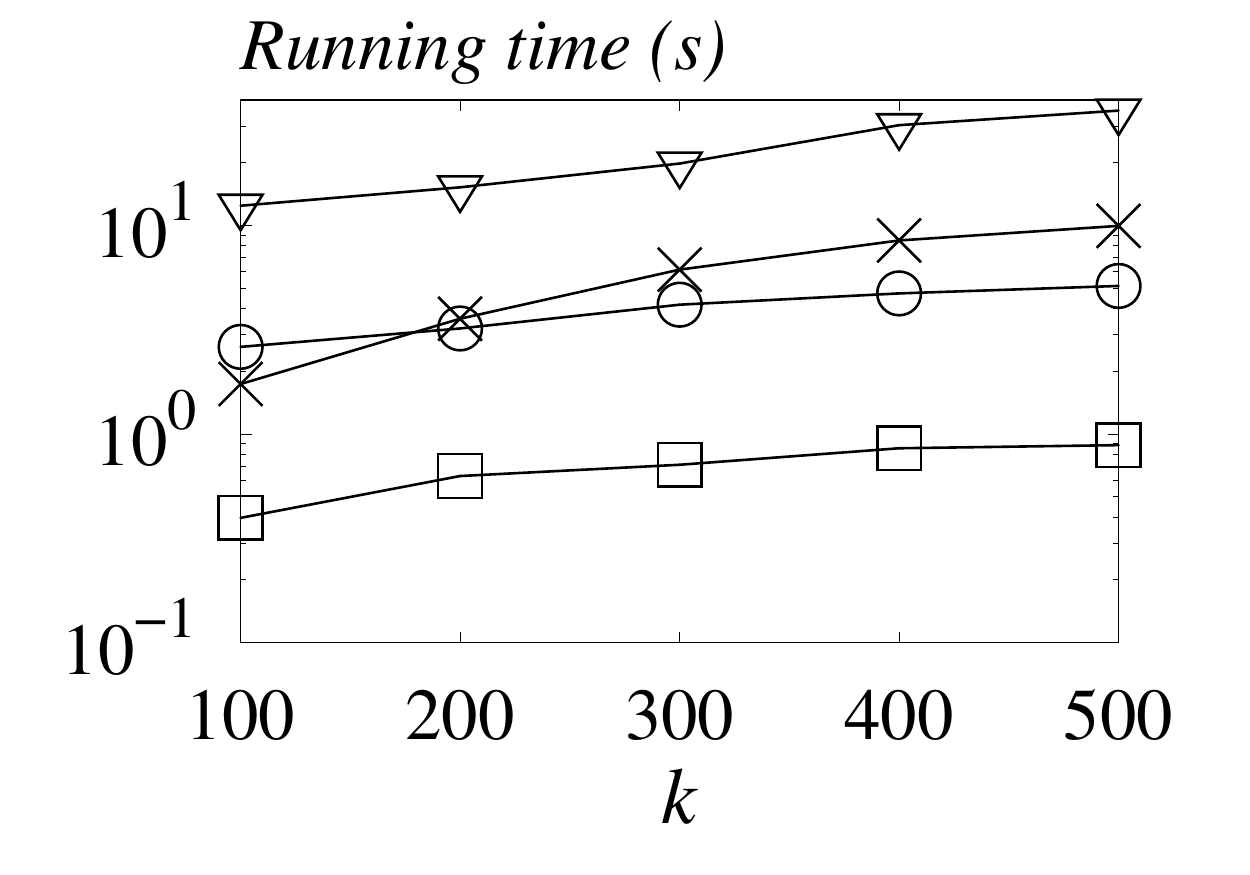} \\
        \hspace{-2mm} (c) Orkut &
        \hspace{-2mm} (d) Twitter \\%[-1mm]
        \hspace{-2mm} \includegraphics[height=35mm]{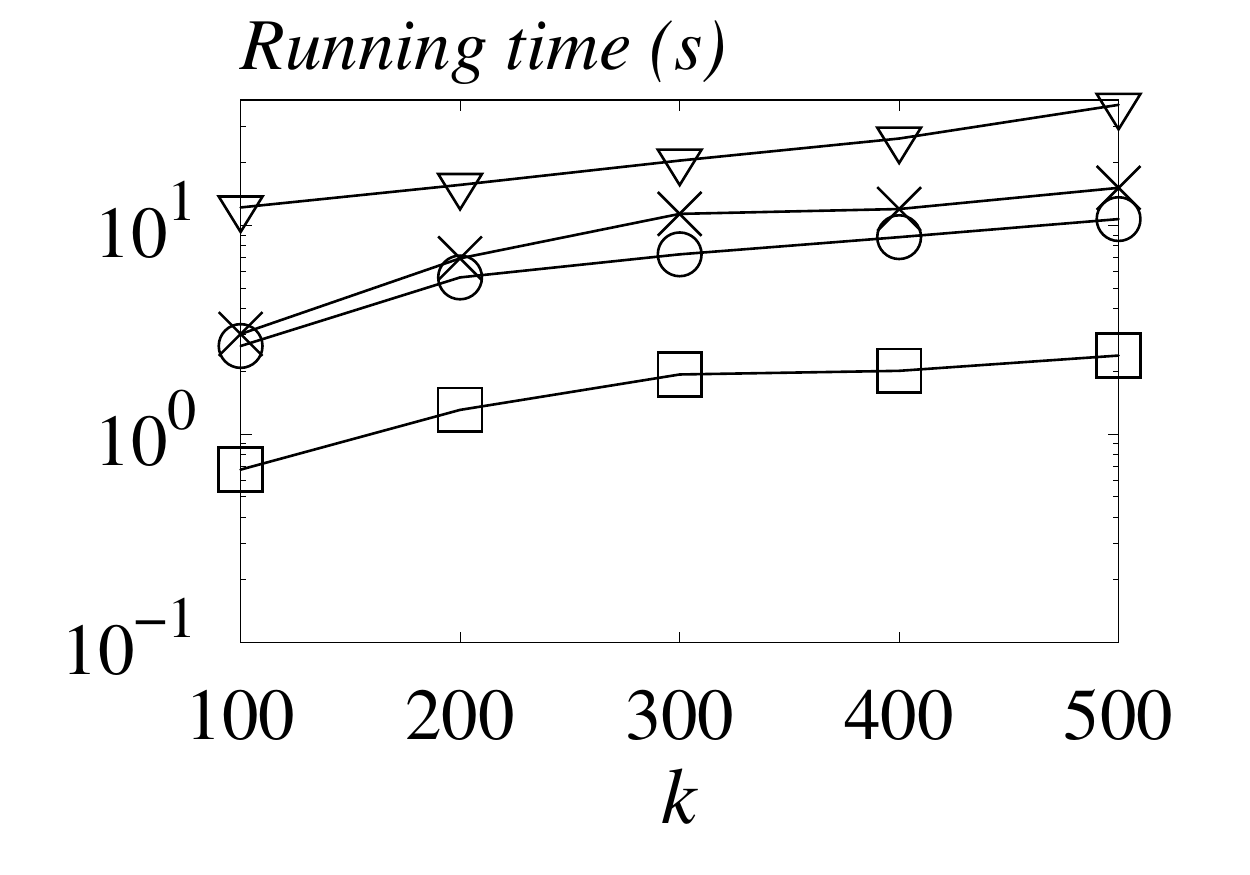} &
        \hspace{-2mm} \includegraphics[height=35mm]{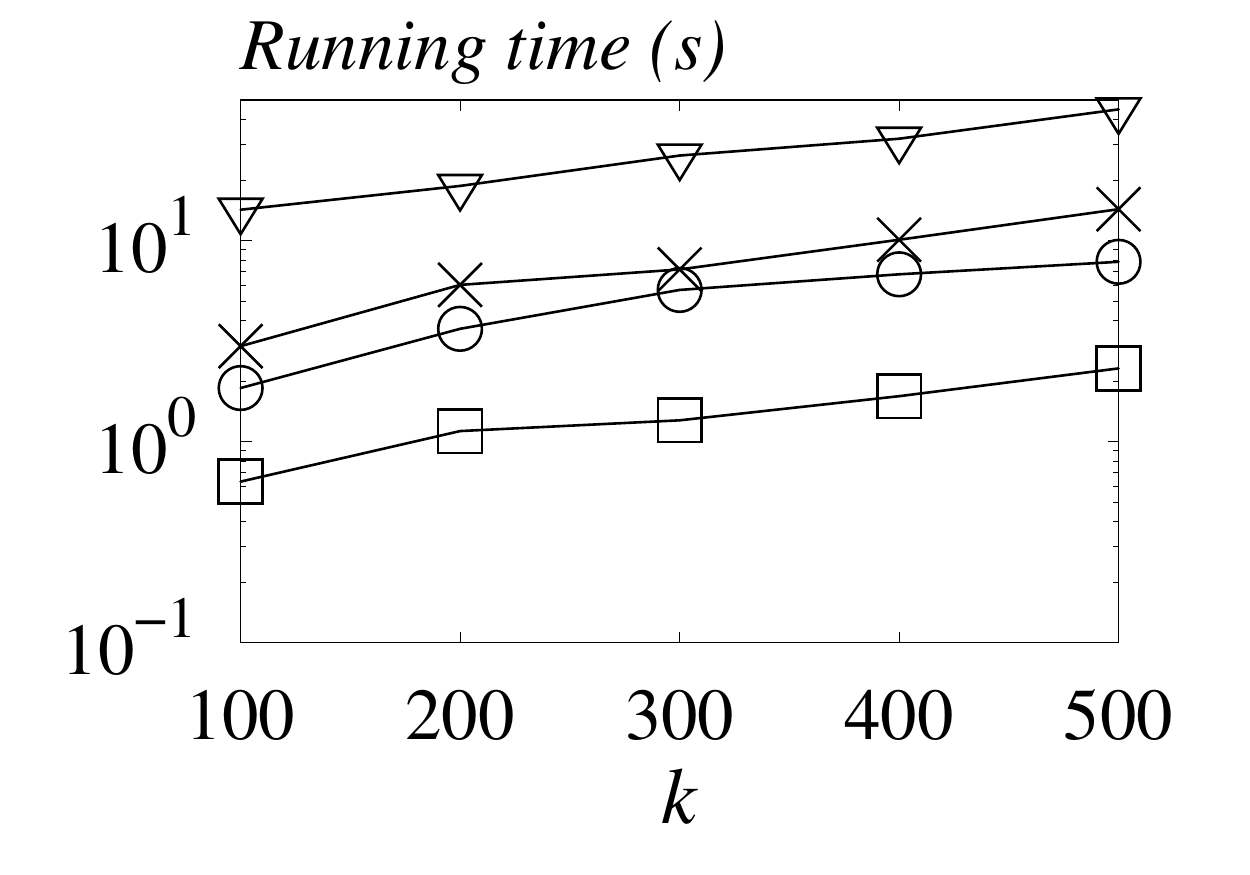} \\
        \hspace{-2mm} (e) X2 &
        \hspace{-2mm} (f) X3 \\%[-1mm]

 \end{tabular}
 \caption{Effectiveness of top-$k$ Optimization: query efficiency.} \label{exp:top-opt-time}
\end{small}
\end{figure*}

\begin{figure*}[!t]
 \centering
   %\vspace{-2mm}
\begin{small}
    \begin{tabular}{cc}
      \multicolumn{2}{c}{\hspace{-2mm} \includegraphics[height=3mm]{./figure/topk-opt/topk-opt-algo-legend-eps-converted-to.pdf}}  \\[0mm]
        \hspace{-2mm} \includegraphics[height=35mm]{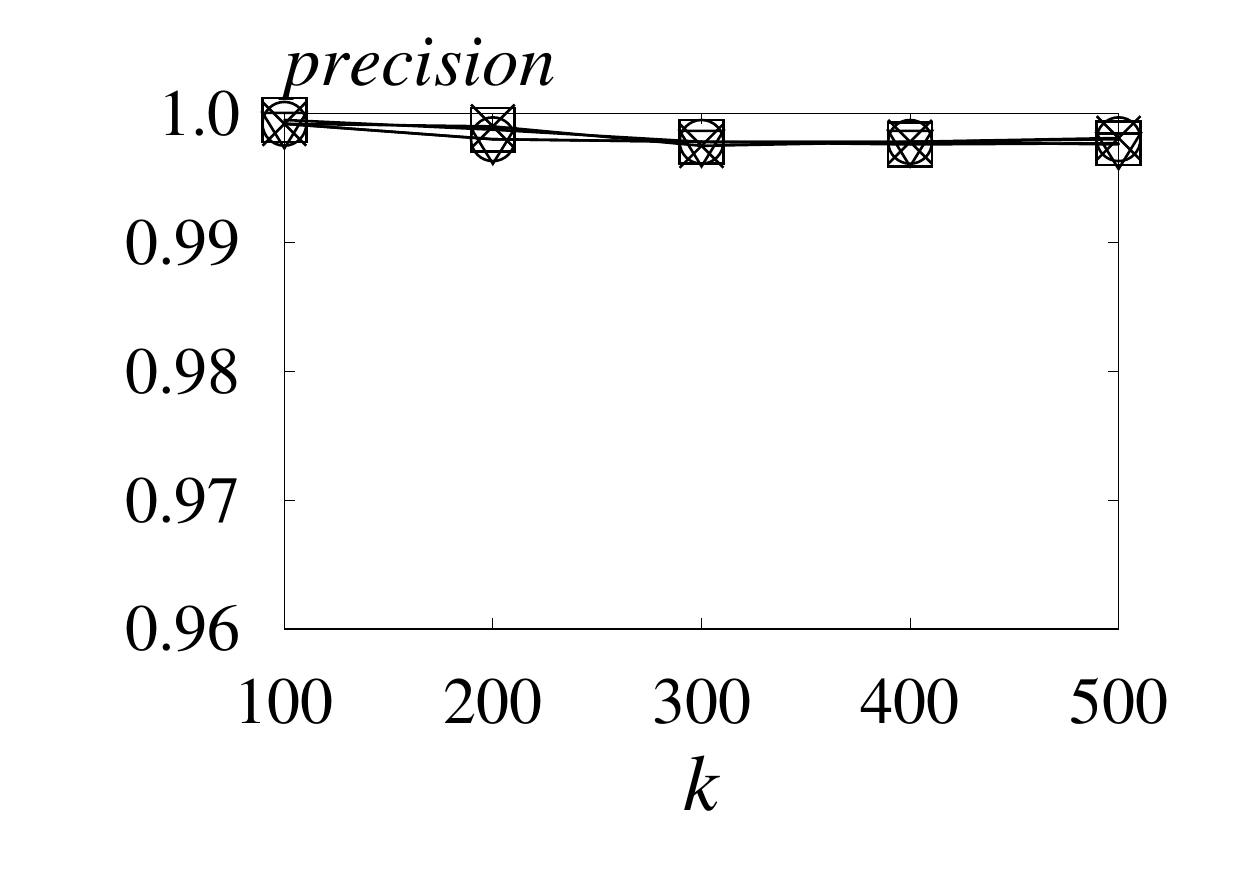} &
        \hspace{-2mm} \includegraphics[height=35mm]{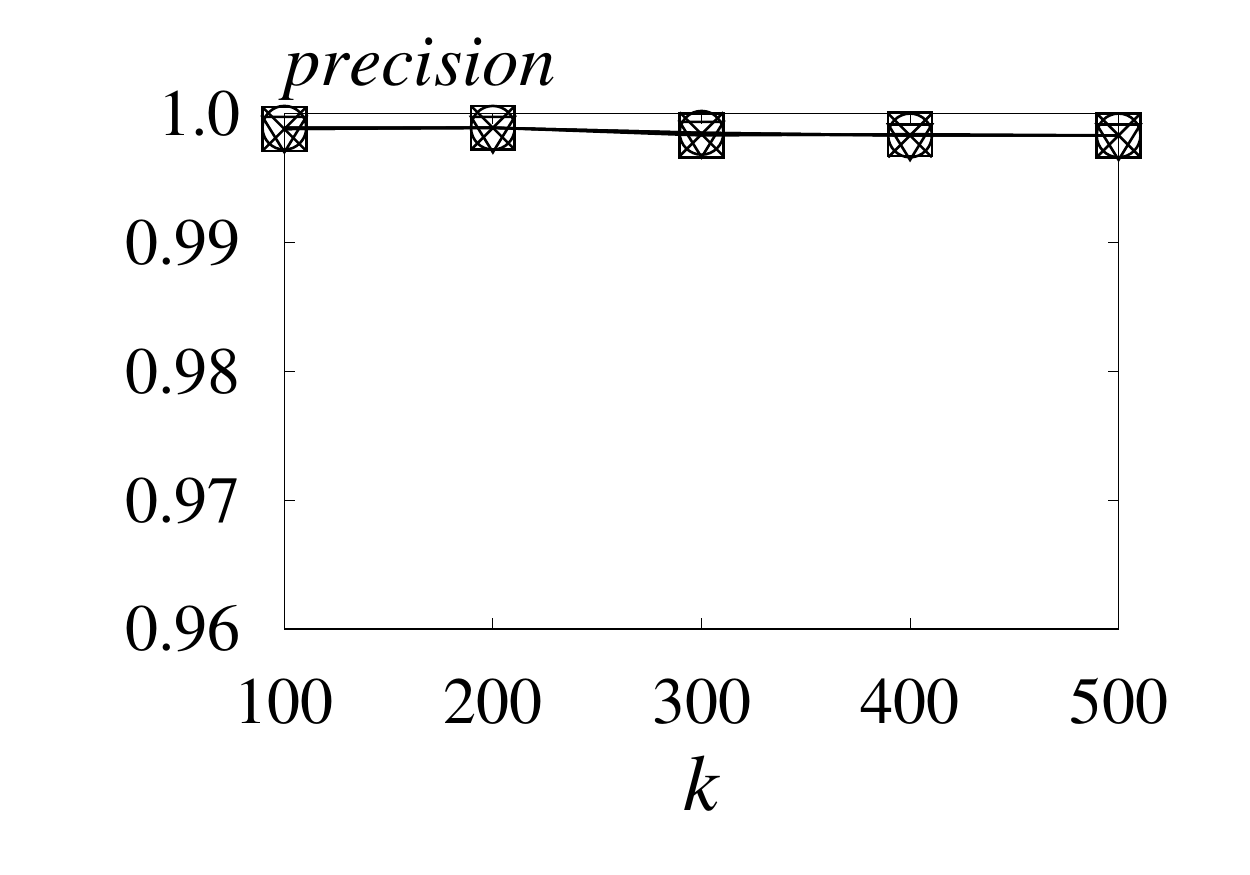}\\
        \hspace{-2mm} (a) DBLP  &
        \hspace{-2mm} (b) Pokec \\
        \hspace{-2mm} \includegraphics[height=35mm]{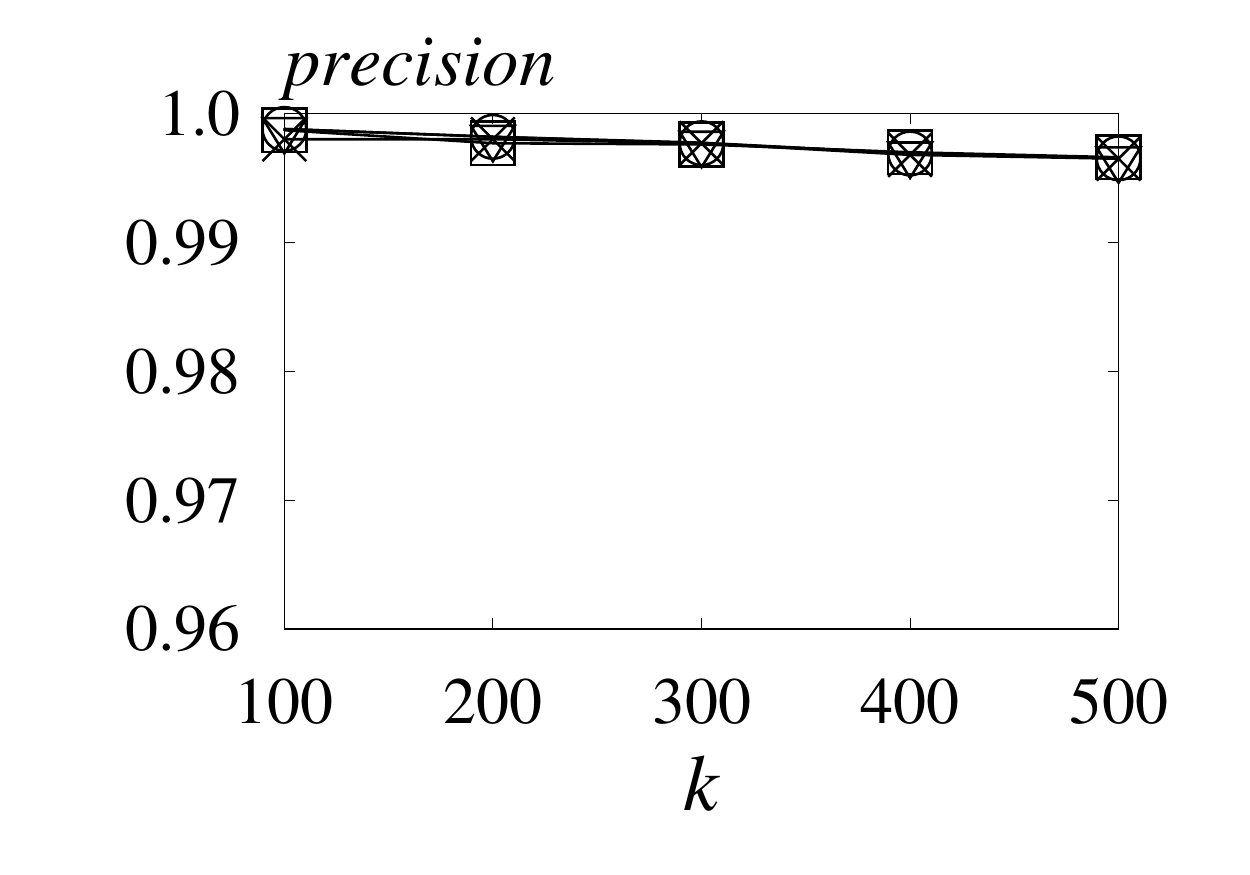} &
        \hspace{-2mm} \includegraphics[height=35mm]{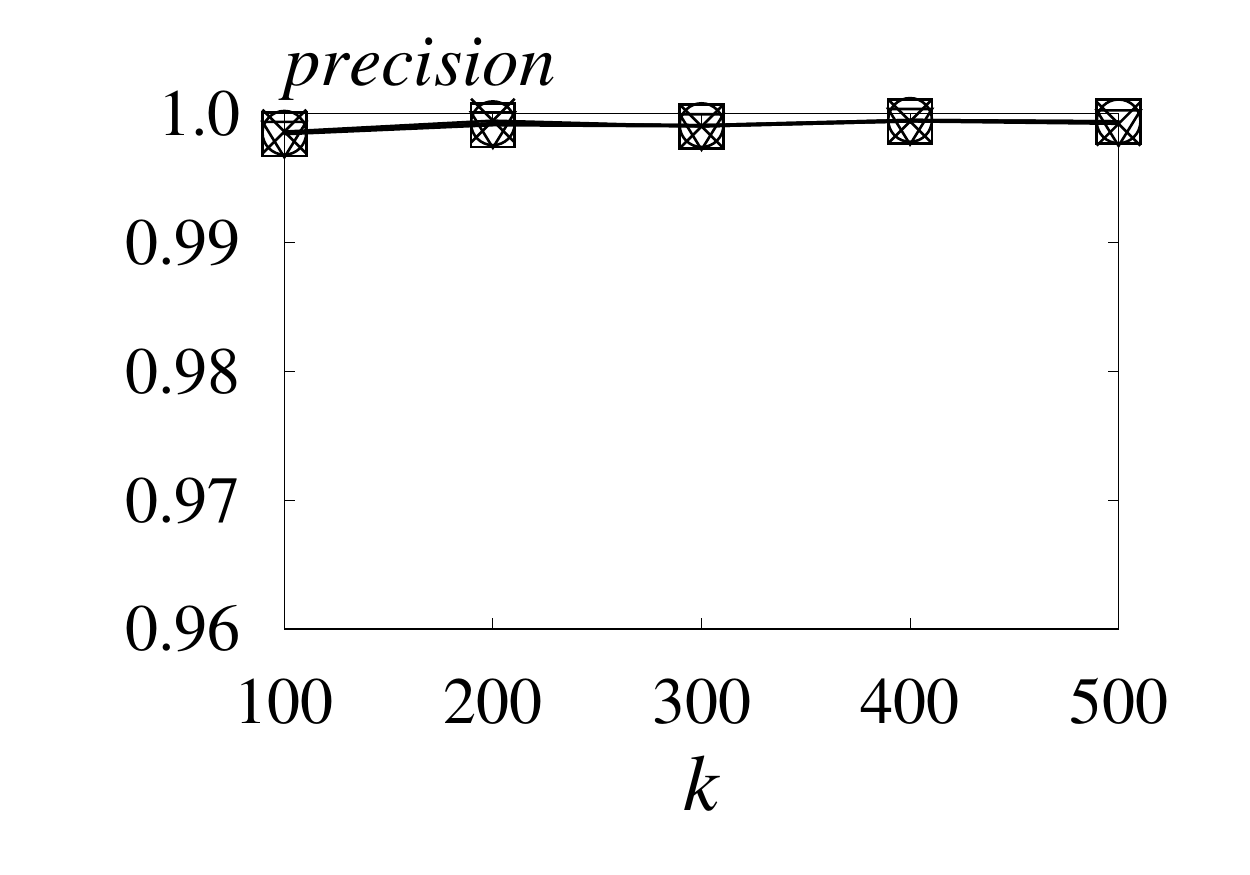}\\
        \hspace{-2mm} (c) Orkut &
        \hspace{-2mm} (d) Twitter \\
        \hspace{-2mm} \includegraphics[height=35mm]{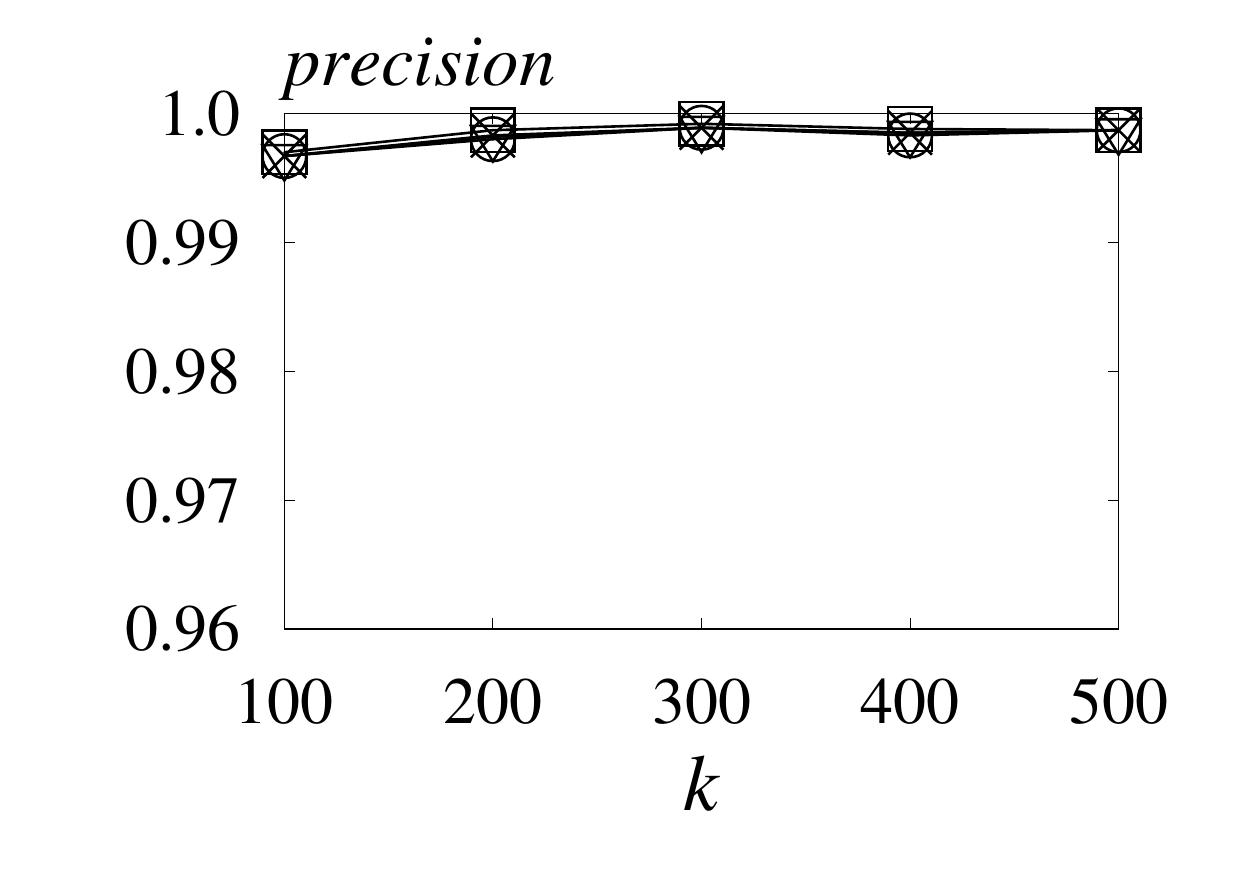} &
        \hspace{-2mm} \includegraphics[height=35mm]{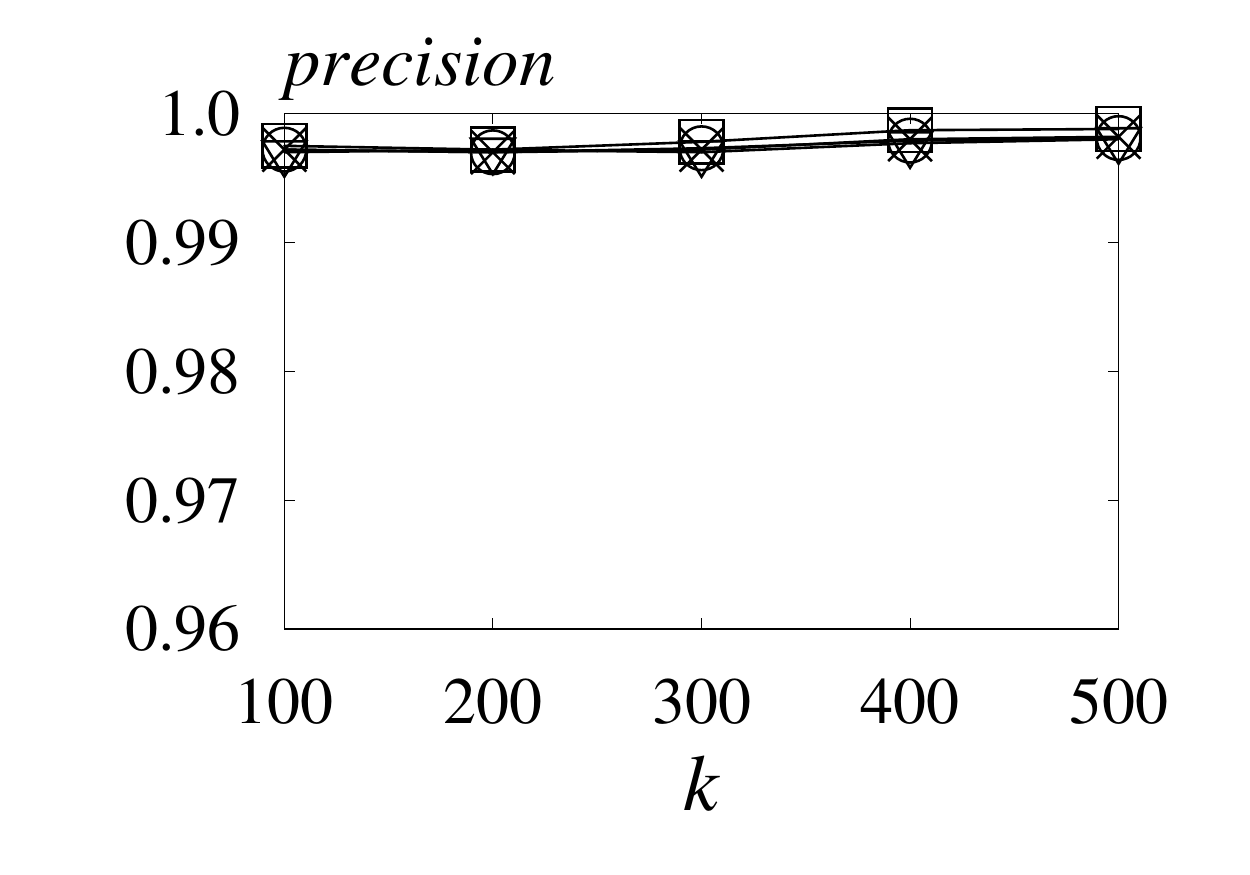}\\
        \hspace{-2mm} (e) X2 &
        \hspace{-2mm} (f) X3 \\
 \end{tabular}
 \caption{Effectiveness of top-$k$ Optimization: query accuracy.} \label{exp:top-opt-acc}
\end{small}
\end{figure*}

\subsection{Preprocessing Costs}\label{sec:exp-preprocessing}

\begin{table}[!t]
\vspace{0mm}
\centering
%\hspace*{-3mm}
%\tbl{Preprocessing time.\label{tbl:exp-preprocessing-time}}{
\caption{Preprocessing time.}
\label{tbl:exp-preprocessing-time}
  \begin{tabular}{ | p{0.6in} | p{0.6in} |p{0.4in}|p{0.8in}|p{1.3in} |}
    \hline
     {\bf Datasets}  & {\em HubPPR}& {\em TPA} & {\em \ssppr+ (for whole graph)} &{\em \ssppr-Basic+} \& {\em \ssppr+ (for top-$k$)} \\
    \hline
 {\em DBLP} & 26.4& 4.4 & 7.6	& 3.6	  \\
    \hline
{\em Web-St}  &  9.3& 1.3 & 3.2 &1.5  \\
    \hline
 {\em Pokec} &90.6& 42.5 & 77.1	&34.6 \\
    \hline
 {\em LJ}& 279.9& 112.5 & 181.3	&81.2   \\
    \hline
{\em Orkut} &530.6& 177.3& 383.5 & 165.9\\
    \hline
{\em Twitter}   &5088.6& 3112.3 & 5130.4	& 2255.4	\\
    \hline\hline

    {\em RMAT-1} &-& -& -& 2695.2 \\
    \hline
{\em RMAT-2}   &-& - & -	& 12076.8	\\
    \hline\hline
{\em X0}   &-&- & 2031.2	&967.14	\\
\hline
{\em X1}   &-&- & 5244.5	& 2428.6	\\
\hline
{\em X2}   &-&- & 5721.7	& 2602.9	\\
\hline
{\em X3}   &-&- & 8836.4	& 4213.5	\\
\hline
  \end{tabular}%}
\end{table}

Finally, we inspect the preprocessing costs of our methods against alternatives. We first examine the preprocessing time of the index-based methods: {\em \ssppr+}, {\em \ssppr-Origin+}, {\em HubPPR}, and {\em TPA}. Note that the $\rmax$ of {\em FORA+} for whole-graph and top-$k$ queries are different, and therefore the preprocessing time are shown as seperately in Table \ref{tbl:exp-preprocessing-time}. The choice of tuned $\rmax$ of {\em FORA+} is the same as that of {\em FORA-Basic}, and therefore their preprocessing time are the same.
As shown in Table \ref{tbl:exp-preprocessing-time}, the preprocessing time of {\em \ssppr+}, {\em FORA-Basic+}, {\em HubPPR}, and {\em TPA} are all moderate. On the largest dataset {\em X3}, our {\em FORA+} for whole graph SSPPR queries (resp. for top-$k$ SSPPR queries) can still finish index construction in less than 3 hours (resp. less than 1 and half an hour), which is more than compensated by its high query performance as shown in Table \ref{tbl:ppr-query}. Besides, this preprocessing time can be further significantly reduced by parallelizing the index construction process.

\begin{table}[!t]
\vspace{0mm}
\centering
%\hspace*{-3mm}
%\tbl{Space consumption.\label{tbl:exp-space-consumption}}{
\caption{Space consumption.}
\label{tbl:exp-space-consumption}
  \begin{tabular}{|  p{0.6in}|p{0.65in}| p{0.85in} |p{0.4in}|p{0.4in}|p{1.0in}|p{0.8in} |}
    \hline
      {\bf Datasets}   & {\em MC}/{\em FORA} & {\em BiPPR/TopPPR}&{\em HubPPR} & {\em TPA} &  {\em FORA-Basic+} \& {\em FORA+} (for top-$k$) & {\em FORA+} (for whole-graph) \\
    \hline
 {\em DBLP} & 18.4MB &	36.8MB& 127.6MB	&23.3MB	& 109.2MB & 200.1MB \\
    \hline
{\em Web-St}  &  10.4MB & 20.8MB& 66.3MB  &12.6MB	& 55.9MB& 101.4MB \\
    \hline
 {\em Pokec} &130.8MB &	261.5MB & 673.5MB &143.6MB	& 542.7MB& 954.6MB\\
    \hline
 {\em LJ}& 295.2MB	&590.5MB & 1.7GB &333.6MB	&1.4GB & 2.5GB \\
    \hline
{\em Orkut} &950.1MB	&1.9GB& 3.5GB &974.9MB	& 2.6GB& 4.2GB \\
    \hline
{\em Twitter}   &6.2GB &12.5GB&  25.1GB	&6.5GB& 18.8GB & 31.4GB	\\
    \hline\hline
{\em RMAT-1}    &6.2GB &12.5GB&  -	&-& 20.0GB & -	\\\hline
{\em RMAT-2}    &33.3GB &66.5GB&  -	&-& 96.3GB & -	\\
\hline\hline
{\em X0} &  2.0GB	&-& - &-	& 7.5GB& 13.1GB \\
    \hline
{\em X1} &3.2GB	&-& - &-	& 13.2GB& 23.2GB \\
    \hline
{\em X2} &4.6GB	&-& - &-	& 17.5GB& 30.4GB \\
    \hline
{\em X3} & 6.0GB	&-& - &-	& 22.6GB& 39.2GB \\
    \hline
  \end{tabular}%}
\end{table}

Next, we examine the space consumption of all methods. Note that each method will need to store at least a copy of the graph. For {\em MC}, {\em FORA}, and {\em TPA}, they can store only a single copy of the input graph (each node stores the out-neighbor list). However, for {\em HubPPR}, {\em BiPPR} and {\em TopPPR}, they need to store two copies of the graph (each node stores an out-neighbor and an in-neighbor list) for efficient random walks and backward propagation.

Table \ref{tbl:exp-space-consumption} reports the space consumption of all methods. As we can observe, the index size of our {\em \ssppr+} for whole-graph queries and top-$k$ queries are no more than 7.5x and 4x of the original graph, respectively. The space consumption is more than compensated by the superb performance of {\em FORA+}, where we can answer a whole-graph query within 30 seconds and  top-$k$ query within 0.9 second on the 1.5 billion edge Twitter graph, which improves over the index-free {\em FORA} by more than an order of magnitude on almost all datasets. This demonstrates the effectiveness and efficiency of our indexing scheme.
{\cblue

Compared to TopPPR, FORA+ achieves 2x, 3x, and 3.4x improvement on Twitter, RMAT-1, and RMAT-2, with only 1.5x, 1.6x, and 1.4x space consumption, respectively.
}
%Besides, even without index, our {\em FORA} can still outperform all alternative index-free and index-based methods on whole-graph SSPPR queries and  top-$k$ SSPPR queries by a large margin.
The experimental results demonstrate that our proposed methods achieve a good trade-off between the space consumption and query efficiency when providing query answers with identical accuracy.

\begin{figure*}[!t]
	\centering
	%\vspace{-2mm}
	\begin{small}
		\begin{tabular}{cc}
			%\multicolumn{2}{c}{\hspace{-2mm} \includegraphics[height=3mm]{./figure/new/algo-legend.eps}}  \\[0mm]
    \hspace{-4mm} \includegraphics[height=2.4mm]{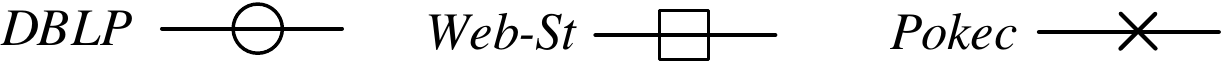} & \hspace{8mm}\includegraphics[height=2.4mm]{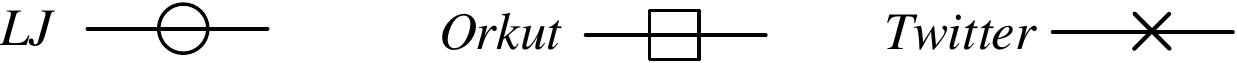} \\[0mm]

			%        \hspace{-2mm} \includegraphics[height=35mm]{./figure/new/dblp-topk-prec.eps} &
			%        \hspace{-2mm} \includegraphics[height=35mm]{./figure/new/pokec-topk-prec.eps}\\
			\hspace{-2mm} \includegraphics[height=35mm]{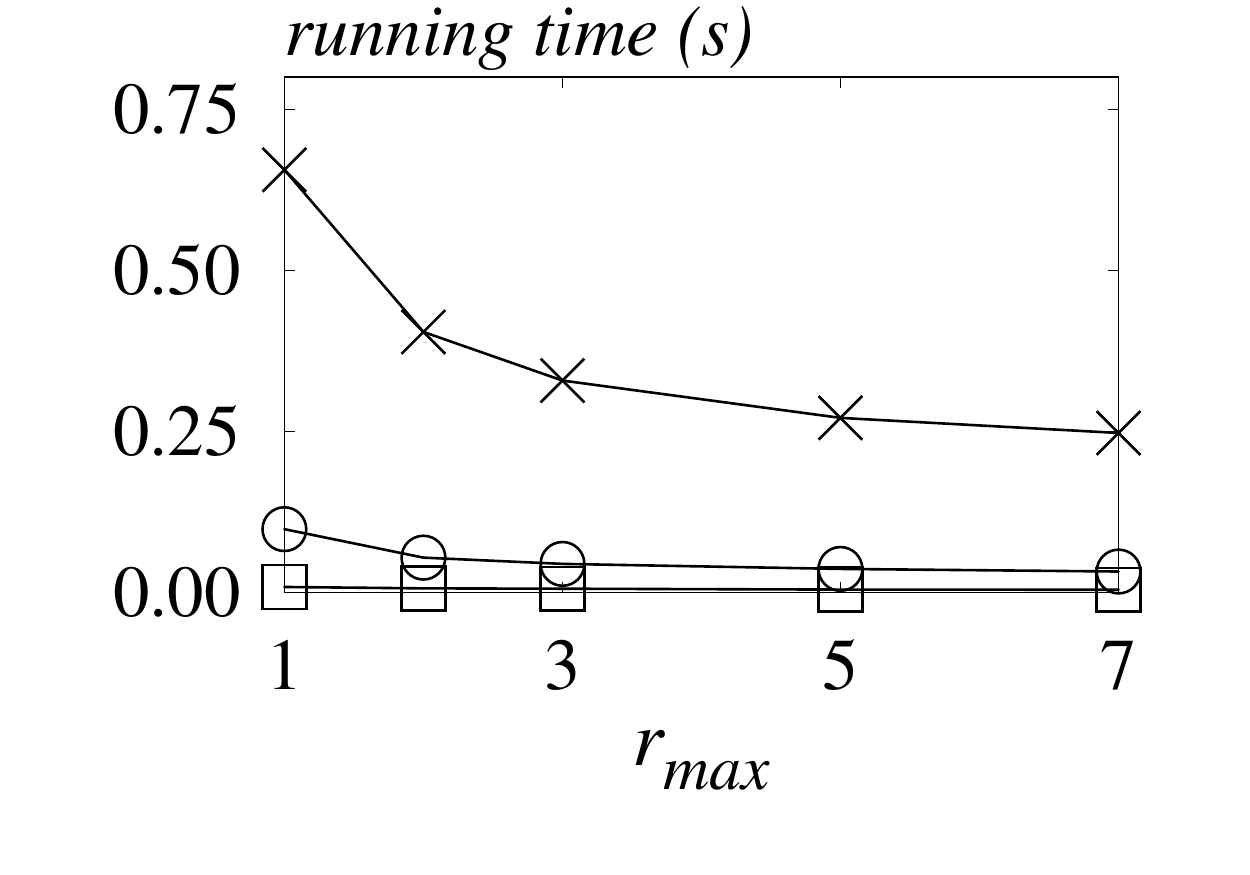} &
			\hspace{-2mm} \includegraphics[height=35mm]{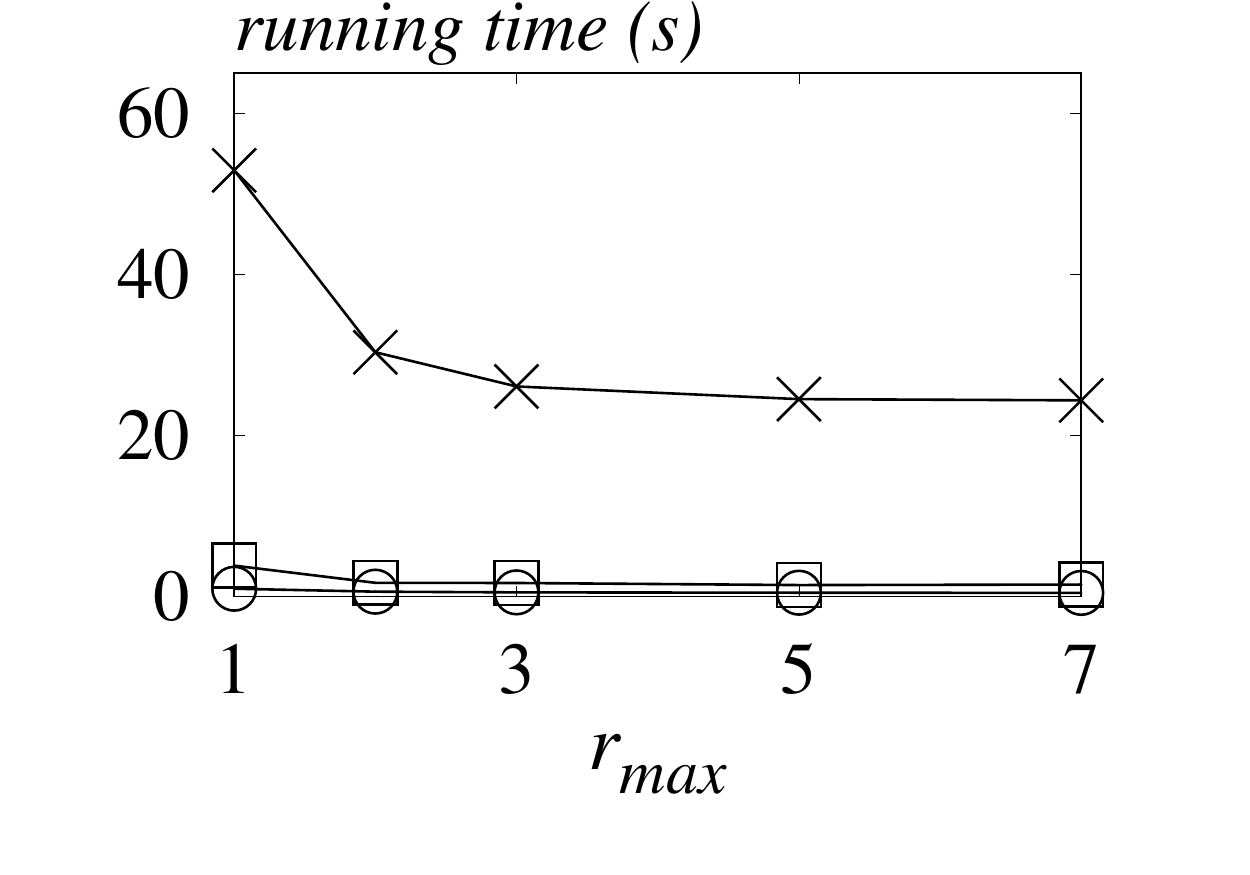}\\
			\hspace{-2mm} (a) &
			\hspace{-2mm} (b) \\
		\end{tabular}
		\caption{Tunning $\rmax$ for whole-graph SSPPR queries ($\times\rmax^*$)} \label{exp:ssppr-tune-rmax}
	\end{small}
\end{figure*}

\begin{figure*}[!t]
 \centering
   %\vspace{-2mm}
    \begin{tabular}{cc}
    \hspace{-4mm} \includegraphics[height=2.4mm]{./figure/topk-tune-rmax/tune-rmax-legend1-eps-converted-to.pdf} & \hspace{8mm}\includegraphics[height=2.4mm]{./figure/topk-tune-rmax/tune-rmax-legend2-eps-converted-to.pdf} \\[0mm]
        \hspace{-4mm}  \includegraphics[height=35mm]{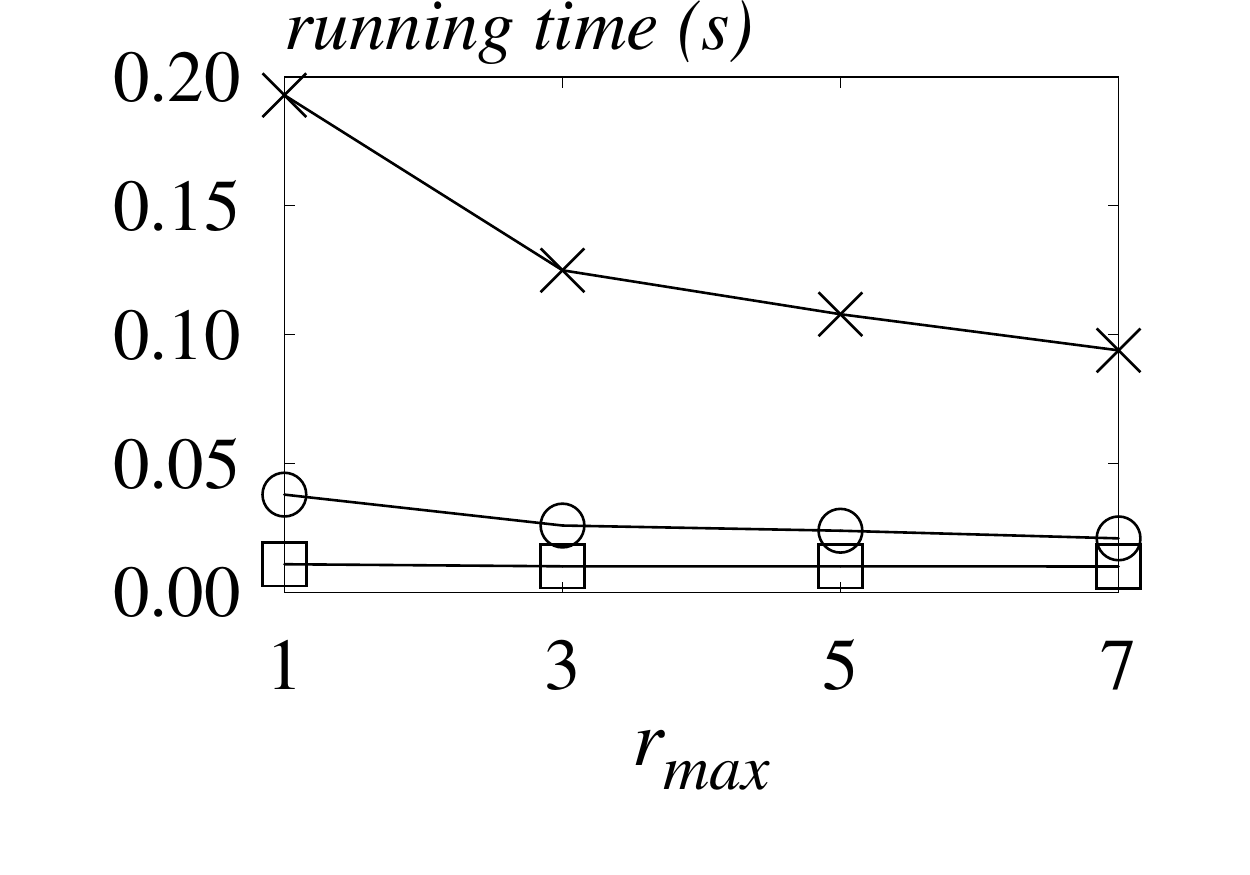} &
         \hspace{8mm} \includegraphics[height=35mm]{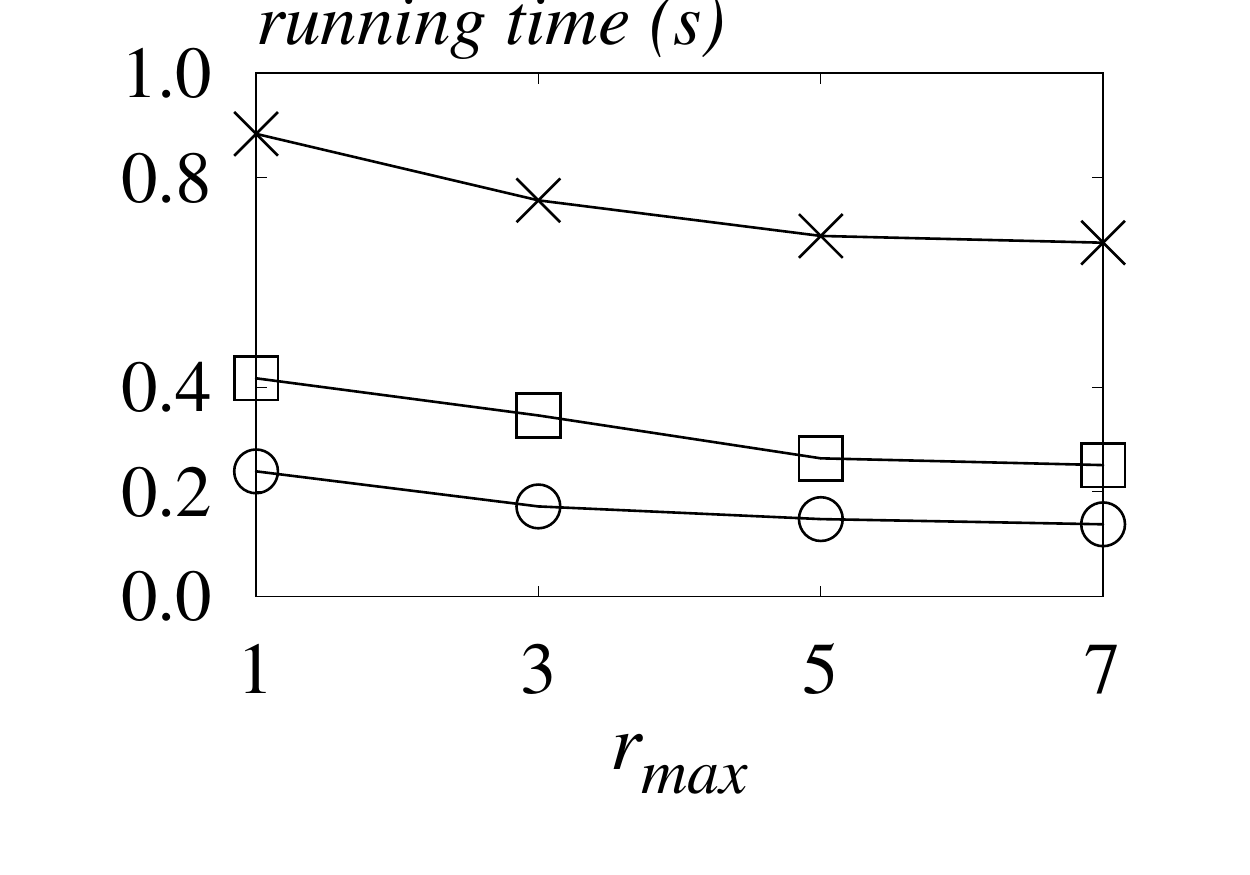}\\
        \hspace{-4mm} (a)   &
        \hspace{8mm}(b)  \\

 \end{tabular}
 \caption{Tunning $\rmax$ for top-$k$ SSPPR queries ($\times\rmax^*$)}. \label{exp:top-tune-rmax}
\end{figure*}

\subsection{Tuning $\rmax$}

Finally, we examine the trade-off between the space consumption and query performance of our {\em FORA+} algorithm by tuning $\rmax$ for whole graph queries and top-$k$ queries.  Let $\rmax^*$ be the $\rmax$ value set according to Section \ref{sec:ana}. We then vary $\rmax$ from $\rmax^*$ to $7\rmax^*$. Notice that, the index size is proportional to $\rmax$ and let $y$ be the index size when $\rmax = \rmax^*$. Then, if $\rmax=3\rmax^*$, then the index size will be $3y$. As we can observe, when we increase the index size for {\em FORA+}, the query performance also improves on almost all datasets for both whole-graph queries and top-$k$ queries. To explain, by increasing the index size, {\em FORA+} avoids the expensive random walks and thus saving the query time. However, from the experiment, we can observe that when $\rmax = 2\rmax^*$ (resp. $rmax= \rmax^*$), FORA+ actually achieves the best trade-off between the query performance and space consumption on whole-graph queries (resp. top-$k$ queries), and hence in our experiment, we set $\rmax=2\rmax^*$ (resp. $\rmax=\rmax^*$) for all the whole-graph queries (resp. top-$k$ queries).

%!TEX root=../ssppr_kdd17_newformat.tex

\section{Conclusion} \label{sec:conc}

We present {\ssppr}, a novel algorithm for approximate single-source personalized PageRank computation. The main ideas include (i) combining Monte-Carlo random walks with Forward Push in a non-trivial and optimized way (ii) pre-computing and indexing random walk results and (iii) additional pruning based on top-$k$ selection. Compared to existing solutions, {\ssppr} involves a reduced number of random walks, avoids expensive backward searches, and provides rigorous guarantees on result quality. Extensive experiments demonstrate that {\ssppr} outperforms existing solutions by a large margin, and enables fast responses for top-$k$ SSPPR searches on very large graphs with little computational resource. 

%\vspace{-1mm}

%\section{Acknowledgments}
%This research is supported by the DSAIR center at NTU, Singapore, as well as grants MOE2015-T2-2-069 from MOE, Singapore, NSFC.61502503 from NSCF, China, and NPRP9-466-1-103 from QNRF, Qatar.

\bibliographystyle{ACM-Reference-Format}
\bibliography{sigproc}


\begin{thebibliography}{33}

%%% the Web address.

\ifx \showCODEN    \undefined \def \showCODEN     #1{\unskip}     \fi
\ifx \showDOI      \undefined \def \showDOI       #1{#1}\fi
\ifx \showISBNx    \undefined \def \showISBNx     #1{\unskip}     \fi
\ifx \showISBNxiii \undefined \def \showISBNxiii  #1{\unskip}     \fi
\ifx \showISSN     \undefined \def \showISSN      #1{\unskip}     \fi
\ifx \showLCCN     \undefined \def \showLCCN      #1{\unskip}     \fi
\ifx \shownote     \undefined \def \shownote      #1{#1}          \fi
\ifx \showarticletitle \undefined \def \showarticletitle #1{#1}   \fi
\ifx \showURL      \undefined \def \showURL       {\relax}        \fi
% invisible to TeX
\providecommand\bibfield[2]{#2}
\providecommand\bibinfo[2]{#2}
\providecommand\natexlab[1]{#1}
\providecommand\showeprint[2][]{arXiv:#2}

\bibitem[\protect\citeauthoryear{Andersen, Borgs, Chayes, Hopcroft, Mirrokni,
  and Teng}{Andersen et~al\mbox{.}}{2007}]%
        {AndersenBCHMT07}
\bibfield{author}{\bibinfo{person}{Reid Andersen}, \bibinfo{person}{Christian
  Borgs}, \bibinfo{person}{Jennifer~T. Chayes}, \bibinfo{person}{John~E.
  Hopcroft}, \bibinfo{person}{Vahab~S. Mirrokni}, {and}
  \bibinfo{person}{Shang{-}Hua Teng}.} \bibinfo{year}{2007}\natexlab{}.
\newblock \showarticletitle{Local Computation of PageRank Contributions}. In
  \bibinfo{booktitle}{\emph{{WAW}}}. \bibinfo{pages}{150--165}.
\newblock


\bibitem[\protect\citeauthoryear{Andersen, Chung, and Lang}{Andersen
  et~al\mbox{.}}{2006}]%
        {AndersenCL06}
\bibfield{author}{\bibinfo{person}{Reid Andersen}, \bibinfo{person}{Fan R.~K.
  Chung}, {and} \bibinfo{person}{Kevin~J. Lang}.}
  \bibinfo{year}{2006}\natexlab{}.
\newblock \showarticletitle{Local Graph Partitioning using PageRank Vectors}.
  In \bibinfo{booktitle}{\emph{{FOCS}}}. \bibinfo{pages}{475--486}.
\newblock


\bibitem[\protect\citeauthoryear{Avrachenkov, Litvak, Nemirovsky, Smirnova, and
  Sokol}{Avrachenkov et~al\mbox{.}}{2011}]%
        {AvrachenkovLNSS11}
\bibfield{author}{\bibinfo{person}{Konstantin Avrachenkov},
  \bibinfo{person}{Nelly Litvak}, \bibinfo{person}{Danil Nemirovsky},
  \bibinfo{person}{Elena Smirnova}, {and} \bibinfo{person}{Marina Sokol}.}
  \bibinfo{year}{2011}\natexlab{}.
\newblock \showarticletitle{Quick Detection of Top-k Personalized PageRank
  Lists}. In \bibinfo{booktitle}{\emph{Algorithms and Models for the Web Graph
  - 8th International Workshop, {WAW} 2011, Atlanta, GA, USA, May 27-29, 2011.
  Proceedings}}. \bibinfo{pages}{50--61}.
\newblock


\bibitem[\protect\citeauthoryear{Backstrom and Leskovec}{Backstrom and
  Leskovec}{2011}]%
        {BackstromL11}
\bibfield{author}{\bibinfo{person}{Lars Backstrom} {and} \bibinfo{person}{Jure
  Leskovec}.} \bibinfo{year}{2011}\natexlab{}.
\newblock \showarticletitle{Supervised random walks: predicting and
  recommending links in social networks}. In
  \bibinfo{booktitle}{\emph{{WSDM}}}. \bibinfo{pages}{635--644}.
\newblock


\bibitem[\protect\citeauthoryear{Bahmani, Chakrabarti, and Xin}{Bahmani
  et~al\mbox{.}}{2011}]%
        {BahmaniCX11}
\bibfield{author}{\bibinfo{person}{Bahman Bahmani}, \bibinfo{person}{Kaushik
  Chakrabarti}, {and} \bibinfo{person}{Dong Xin}.}
  \bibinfo{year}{2011}\natexlab{}.
\newblock \showarticletitle{Fast personalized PageRank on MapReduce}. In
  \bibinfo{booktitle}{\emph{{SIGMOD}}}. \bibinfo{pages}{973--984}.
\newblock


\bibitem[\protect\citeauthoryear{Berkhin}{Berkhin}{2005}]%
        {Berkhin05}
\bibfield{author}{\bibinfo{person}{Pavel Berkhin}.}
  \bibinfo{year}{2005}\natexlab{}.
\newblock \showarticletitle{Survey: {A} Survey on PageRank Computing}.
\newblock \bibinfo{journal}{\emph{Internet Mathematics}} \bibinfo{volume}{2},
  \bibinfo{number}{1} (\bibinfo{year}{2005}), \bibinfo{pages}{73--120}.
\newblock


\bibitem[\protect\citeauthoryear{Chakrabarti, Zhan, and Faloutsos}{Chakrabarti
  et~al\mbox{.}}{2004}]%
        {ChakrabartiZF04}
\bibfield{author}{\bibinfo{person}{Deepayan Chakrabarti},
  \bibinfo{person}{Yiping Zhan}, {and} \bibinfo{person}{Christos Faloutsos}.}
  \bibinfo{year}{2004}\natexlab{}.
\newblock \showarticletitle{{R-MAT:} {A} Recursive Model for Graph Mining}. In
  \bibinfo{booktitle}{\emph{{SDM}}}. \bibinfo{pages}{442--446}.
\newblock


\bibitem[\protect\citeauthoryear{Chung and Lu}{Chung and Lu}{2006}]%
        {ChungL06}
\bibfield{author}{\bibinfo{person}{Fan R.~K. Chung} {and}
  \bibinfo{person}{Lincoln Lu}.} \bibinfo{year}{2006}\natexlab{}.
\newblock \showarticletitle{Survey: Concentration Inequalities and Martingale
  Inequalities: {A} Survey}.
\newblock \bibinfo{journal}{\emph{Internet Mathematics}} \bibinfo{volume}{3},
  \bibinfo{number}{1} (\bibinfo{year}{2006}), \bibinfo{pages}{79--127}.
\newblock


\bibitem[\protect\citeauthoryear{Fogaras, R{\'{a}}cz, Csalog{\'{a}}ny, and
  Sarl{\'{o}}s}{Fogaras et~al\mbox{.}}{2005}]%
        {fogaras2005towards}
\bibfield{author}{\bibinfo{person}{D{\'{a}}niel Fogaras},
  \bibinfo{person}{Bal{\'{a}}zs R{\'{a}}cz}, \bibinfo{person}{K{\'{a}}roly
  Csalog{\'{a}}ny}, {and} \bibinfo{person}{Tam{\'{a}}s Sarl{\'{o}}s}.}
  \bibinfo{year}{2005}\natexlab{}.
\newblock \showarticletitle{Towards Scaling Fully Personalized PageRank:
  Algorithms, Lower Bounds, and Experiments}.
\newblock \bibinfo{journal}{\emph{Internet Mathematics}} \bibinfo{volume}{2},
  \bibinfo{number}{3} (\bibinfo{year}{2005}), \bibinfo{pages}{333--358}.
\newblock


\bibitem[\protect\citeauthoryear{Fujiwara, Nakatsuji, Shiokawa, Mishima, and
  Onizuka}{Fujiwara et~al\mbox{.}}{2013}]%
        {FujiwaraNSMO13sigmod}
\bibfield{author}{\bibinfo{person}{Yasuhiro Fujiwara}, \bibinfo{person}{Makoto
  Nakatsuji}, \bibinfo{person}{Hiroaki Shiokawa}, \bibinfo{person}{Takeshi
  Mishima}, {and} \bibinfo{person}{Makoto Onizuka}.}
  \bibinfo{year}{2013}\natexlab{}.
\newblock \showarticletitle{Efficient ad-hoc search for personalized PageRank}.
  In \bibinfo{booktitle}{\emph{{SIGMOD}}}. \bibinfo{pages}{445--456}.
\newblock


\bibitem[\protect\citeauthoryear{Fujiwara, Nakatsuji, Yamamuro, Shiokawa, and
  Onizuka}{Fujiwara et~al\mbox{.}}{2012}]%
        {FujiwaraNYSO12}
\bibfield{author}{\bibinfo{person}{Yasuhiro Fujiwara}, \bibinfo{person}{Makoto
  Nakatsuji}, \bibinfo{person}{Takeshi Yamamuro}, \bibinfo{person}{Hiroaki
  Shiokawa}, {and} \bibinfo{person}{Makoto Onizuka}.}
  \bibinfo{year}{2012}\natexlab{}.
\newblock \showarticletitle{Efficient personalized pagerank with accuracy
  assurance}. In \bibinfo{booktitle}{\emph{{KDD}}}. \bibinfo{pages}{15--23}.
\newblock


\bibitem[\protect\citeauthoryear{Gupta, Pathak, and Chakrabarti}{Gupta
  et~al\mbox{.}}{2008}]%
        {GuptaPC08}
\bibfield{author}{\bibinfo{person}{Manish~S. Gupta}, \bibinfo{person}{Amit
  Pathak}, {and} \bibinfo{person}{Soumen Chakrabarti}.}
  \bibinfo{year}{2008}\natexlab{}.
\newblock \showarticletitle{Fast algorithms for topk personalized pagerank
  queries}. In \bibinfo{booktitle}{\emph{{WWW}}}. \bibinfo{pages}{1225--1226}.
\newblock


\bibitem[\protect\citeauthoryear{Gupta, Goel, Lin, Sharma, Wang, and
  Zadeh}{Gupta et~al\mbox{.}}{2013}]%
        {gupta2013wtf}
\bibfield{author}{\bibinfo{person}{Pankaj Gupta}, \bibinfo{person}{Ashish
  Goel}, \bibinfo{person}{Jimmy Lin}, \bibinfo{person}{Aneesh Sharma},
  \bibinfo{person}{Dong Wang}, {and} \bibinfo{person}{Reza Zadeh}.}
  \bibinfo{year}{2013}\natexlab{}.
\newblock \showarticletitle{Wtf: The who to follow service at twitter}. In
  \bibinfo{booktitle}{\emph{{WWW}}}. \bibinfo{pages}{505--514}.
\newblock


\bibitem[\protect\citeauthoryear{Haveliwala}{Haveliwala}{2002}]%
        {Haveliwala02}
\bibfield{author}{\bibinfo{person}{Taher~H. Haveliwala}.}
  \bibinfo{year}{2002}\natexlab{}.
\newblock \showarticletitle{Topic-sensitive PageRank}. In
  \bibinfo{booktitle}{\emph{{WWW}}}. \bibinfo{pages}{517--526}.
\newblock


\bibitem[\protect\citeauthoryear{J{\"{a}}rvelin and
  Kek{\"{a}}l{\"{a}}inen}{J{\"{a}}rvelin and Kek{\"{a}}l{\"{a}}inen}{2000}]%
        {JarvelinK00}
\bibfield{author}{\bibinfo{person}{Kalervo J{\"{a}}rvelin} {and}
  \bibinfo{person}{Jaana Kek{\"{a}}l{\"{a}}inen}.}
  \bibinfo{year}{2000}\natexlab{}.
\newblock \showarticletitle{{IR} evaluation methods for retrieving highly
  relevant documents}. In \bibinfo{booktitle}{\emph{{SIGIR}}}.
  \bibinfo{pages}{41--48}.
\newblock


\bibitem[\protect\citeauthoryear{Jeh and Widom}{Jeh and Widom}{2003}]%
        {JehW03}
\bibfield{author}{\bibinfo{person}{Glen Jeh} {and} \bibinfo{person}{Jennifer
  Widom}.} \bibinfo{year}{2003}\natexlab{}.
\newblock \showarticletitle{Scaling personalized web search}. In
  \bibinfo{booktitle}{\emph{{WWW}}}. \bibinfo{pages}{271--279}.
\newblock


\bibitem[\protect\citeauthoryear{Lin}{Lin}{2019}]%
        {Lin19}
\bibfield{author}{\bibinfo{person}{Wenqing Lin}.}
  \bibinfo{year}{2019}\natexlab{}.
\newblock \showarticletitle{Distributed Algorithms for Fully Personalized
  PageRank on Large Graphs}. In \bibinfo{booktitle}{\emph{The World Wide Web
  Conference, {WWW} 2019, San Francisco, CA, USA, May 13-17, 2019}}.
  \bibinfo{pages}{1084--1094}.
\newblock


\bibitem[\protect\citeauthoryear{Lofgren}{Lofgren}{2015}]%
        {Lofgrenthesis15}
\bibfield{author}{\bibinfo{person}{Peter Lofgren}.}
  \bibinfo{year}{2015}\natexlab{}.
\newblock \showarticletitle{Efficient Algorithms for Personalized PageRank}.
\newblock \bibinfo{journal}{\emph{CoRR}}  \bibinfo{volume}{abs/1512.04633}
  (\bibinfo{year}{2015}).
\newblock


\bibitem[\protect\citeauthoryear{Lofgren, Banerjee, and Goel}{Lofgren
  et~al\mbox{.}}{2016}]%
        {lofgren2015personalized}
\bibfield{author}{\bibinfo{person}{Peter Lofgren}, \bibinfo{person}{Siddhartha
  Banerjee}, {and} \bibinfo{person}{Ashish Goel}.}
  \bibinfo{year}{2016}\natexlab{}.
\newblock \showarticletitle{Personalized pagerank estimation and search: A
  bidirectional approach}. In \bibinfo{booktitle}{\emph{{WSDM}}}.
  \bibinfo{pages}{163--172}.
\newblock


\bibitem[\protect\citeauthoryear{Lofgren, Banerjee, Goel, and
  Seshadhri}{Lofgren et~al\mbox{.}}{2014}]%
        {lofgren2014fast}
\bibfield{author}{\bibinfo{person}{Peter~A Lofgren},
  \bibinfo{person}{Siddhartha Banerjee}, \bibinfo{person}{Ashish Goel}, {and}
  \bibinfo{person}{C Seshadhri}.} \bibinfo{year}{2014}\natexlab{}.
\newblock \showarticletitle{Fast-ppr: Scaling personalized pagerank estimation
  for large graphs}. In \bibinfo{booktitle}{\emph{{KDD}}}.
  \bibinfo{pages}{1436--1445}.
\newblock


\bibitem[\protect\citeauthoryear{Maehara, Akiba, Iwata, and
  Kawarabayashi}{Maehara et~al\mbox{.}}{2014}]%
        {maehara2014computing}
\bibfield{author}{\bibinfo{person}{Takanori Maehara}, \bibinfo{person}{Takuya
  Akiba}, \bibinfo{person}{Yoichi Iwata}, {and} \bibinfo{person}{Ken-ichi
  Kawarabayashi}.} \bibinfo{year}{2014}\natexlab{}.
\newblock \showarticletitle{Computing personalized PageRank quickly by
  exploiting graph structures}.
\newblock \bibinfo{journal}{\emph{PVLDB}} \bibinfo{volume}{7},
  \bibinfo{number}{12} (\bibinfo{year}{2014}), \bibinfo{pages}{1023--1034}.
\newblock


\bibitem[\protect\citeauthoryear{Ohsaka, Maehara, and Kawarabayashi}{Ohsaka
  et~al\mbox{.}}{2015}]%
        {OhsakaMK15}
\bibfield{author}{\bibinfo{person}{Naoto Ohsaka}, \bibinfo{person}{Takanori
  Maehara}, {and} \bibinfo{person}{Ken{-}ichi Kawarabayashi}.}
  \bibinfo{year}{2015}\natexlab{}.
\newblock \showarticletitle{Efficient PageRank Tracking in Evolving Networks}.
  In \bibinfo{booktitle}{\emph{{SIGKDD} 2015}}. \bibinfo{pages}{875--884}.
\newblock


\bibitem[\protect\citeauthoryear{Page, Brin, Motwani, and Winograd}{Page
  et~al\mbox{.}}{1999}]%
        {page1999pagerank}
\bibfield{author}{\bibinfo{person}{Lawrence Page}, \bibinfo{person}{Sergey
  Brin}, \bibinfo{person}{Rajeev Motwani}, {and} \bibinfo{person}{Terry
  Winograd}.} \bibinfo{year}{1999}\natexlab{}.
\newblock \showarticletitle{The PageRank citation ranking: bringing order to
  the web.}
\newblock  (\bibinfo{year}{1999}).
\newblock


\bibitem[\protect\citeauthoryear{Sarma, Molla, Pandurangan, and Upfal}{Sarma
  et~al\mbox{.}}{2013}]%
        {SarmaMPU13}
\bibfield{author}{\bibinfo{person}{Atish~Das Sarma},
  \bibinfo{person}{Anisur~Rahaman Molla}, \bibinfo{person}{Gopal Pandurangan},
  {and} \bibinfo{person}{Eli Upfal}.} \bibinfo{year}{2013}\natexlab{}.
\newblock \showarticletitle{Fast Distributed PageRank Computation}. In
  \bibinfo{booktitle}{\emph{ICDCN}}. \bibinfo{pages}{11--26}.
\newblock


\bibitem[\protect\citeauthoryear{Shin, Jung, Sael, and Kang}{Shin
  et~al\mbox{.}}{2015}]%
        {ShinJSK15}
\bibfield{author}{\bibinfo{person}{Kijung Shin}, \bibinfo{person}{Jinhong
  Jung}, \bibinfo{person}{Lee Sael}, {and} \bibinfo{person}{U. Kang}.}
  \bibinfo{year}{2015}\natexlab{}.
\newblock \showarticletitle{{BEAR:} Block Elimination Approach for Random Walk
  with Restart on Large Graphs}. In \bibinfo{booktitle}{\emph{{SIGMOD}}}.
  \bibinfo{pages}{1571--1585}.
\newblock


\bibitem[\protect\citeauthoryear{Wang, Tang, Xiao, Yang, and Li}{Wang
  et~al\mbox{.}}{2016}]%
        {WangTXYL16}
\bibfield{author}{\bibinfo{person}{Sibo Wang}, \bibinfo{person}{Youze Tang},
  \bibinfo{person}{Xiaokui Xiao}, \bibinfo{person}{Yin Yang}, {and}
  \bibinfo{person}{Zengxiang Li}.} \bibinfo{year}{2016}\natexlab{}.
\newblock \showarticletitle{HubPPR: Effective Indexing for Approximate
  Personalized PageRank}.
\newblock \bibinfo{journal}{\emph{{PVLDB}}} \bibinfo{volume}{10},
  \bibinfo{number}{3} (\bibinfo{year}{2016}), \bibinfo{pages}{205--216}.
\newblock
\urldef\tempurl%
\url{http://www.vldb.org/pvldb/vol10/p205-wang.pdf}
\showURL{%
\tempurl}


\bibitem[\protect\citeauthoryear{Wang and Tao}{Wang and Tao}{2018}]%
        {WangT18}
\bibfield{author}{\bibinfo{person}{Sibo Wang} {and} \bibinfo{person}{Yufei
  Tao}.} \bibinfo{year}{2018}\natexlab{}.
\newblock \showarticletitle{Efficient Algorithms for Finding Approximate Heavy
  Hitters in Personalized PageRanks}. In \bibinfo{booktitle}{\emph{{SIGMOD}}}.
  \bibinfo{pages}{1113--1127}.
\newblock


\bibitem[\protect\citeauthoryear{Wang, Yang, Xiao, Wei, and Yang}{Wang
  et~al\mbox{.}}{2017}]%
        {Wang17}
\bibfield{author}{\bibinfo{person}{Sibo Wang}, \bibinfo{person}{Renchi Yang},
  \bibinfo{person}{Xiaokui Xiao}, \bibinfo{person}{Zhewei Wei}, {and}
  \bibinfo{person}{Yin Yang}.} \bibinfo{year}{2017}\natexlab{}.
\newblock \showarticletitle{{FORA:} Simple and Effective Approximate
  Single-Source Personalized PageRank}. In
  \bibinfo{booktitle}{\emph{{SIGKDD}}}. \bibinfo{pages}{505--514}.
\newblock


\bibitem[\protect\citeauthoryear{Wei, He, Xiao, Wang, Liu, Du, and Wen}{Wei
  et~al\mbox{.}}{2019}]%
        {WeiHX0LDW19}
\bibfield{author}{\bibinfo{person}{Zhewei Wei}, \bibinfo{person}{Xiaodong He},
  \bibinfo{person}{Xiaokui Xiao}, \bibinfo{person}{Sibo Wang},
  \bibinfo{person}{Yu Liu}, \bibinfo{person}{Xiaoyong Du}, {and}
  \bibinfo{person}{Ji{-}Rong Wen}.} \bibinfo{year}{2019}\natexlab{}.
\newblock \showarticletitle{PRSim: Sublinear Time SimRank Computation on Large
  Power-Law Graphs}. In \bibinfo{booktitle}{\emph{{SIGMOD}}}.
  \bibinfo{pages}{1042--1059}.
\newblock


\bibitem[\protect\citeauthoryear{Wei, He, Xiao, Wang, Shang, and Wen}{Wei
  et~al\mbox{.}}{2018}]%
        {WeiHX0SW18}
\bibfield{author}{\bibinfo{person}{Zhewei Wei}, \bibinfo{person}{Xiaodong He},
  \bibinfo{person}{Xiaokui Xiao}, \bibinfo{person}{Sibo Wang},
  \bibinfo{person}{Shuo Shang}, {and} \bibinfo{person}{Ji{-}Rong Wen}.}
  \bibinfo{year}{2018}\natexlab{}.
\newblock \showarticletitle{TopPPR: Top-k Personalized PageRank Queries with
  Precision Guarantees on Large Graphs}. In
  \bibinfo{booktitle}{\emph{{SIGMOD}}}. \bibinfo{pages}{441--456}.
\newblock


\bibitem[\protect\citeauthoryear{Yoon, Jung, and Kang}{Yoon
  et~al\mbox{.}}{2018}]%
        {yoon2018tpa}
\bibfield{author}{\bibinfo{person}{Minji Yoon}, \bibinfo{person}{Jinhong Jung},
  {and} \bibinfo{person}{U Kang}.} \bibinfo{year}{2018}\natexlab{}.
\newblock \showarticletitle{TPA: Fast, Scalable, and Accurate method for
  Approximate Random Walk with Restart on Billion Scale Graphs}. In
  \bibinfo{booktitle}{\emph{{ICDE}}}.
\newblock


\bibitem[\protect\citeauthoryear{Zhang, Lofgren, and Goel}{Zhang
  et~al\mbox{.}}{2016}]%
        {ZhangLG16}
\bibfield{author}{\bibinfo{person}{Hongyang Zhang}, \bibinfo{person}{Peter
  Lofgren}, {and} \bibinfo{person}{Ashish Goel}.}
  \bibinfo{year}{2016}\natexlab{}.
\newblock \showarticletitle{Approximate Personalized PageRank on Dynamic
  Graphs}. In \bibinfo{booktitle}{\emph{{KDD}}}. \bibinfo{pages}{1315--1324}.
\newblock


\bibitem[\protect\citeauthoryear{Zhu, Fang, Chang, and Ying}{Zhu
  et~al\mbox{.}}{2013}]%
        {ZhuFCY13}
\bibfield{author}{\bibinfo{person}{Fanwei Zhu}, \bibinfo{person}{Yuan Fang},
  \bibinfo{person}{Kevin~Chen{-}Chuan Chang}, {and} \bibinfo{person}{Jing
  Ying}.} \bibinfo{year}{2013}\natexlab{}.
\newblock \showarticletitle{Incremental and Accuracy-Aware Personalized
  PageRank through Scheduled Approximation}.
\newblock \bibinfo{journal}{\emph{{PVLDB}}} \bibinfo{volume}{6},
  \bibinfo{number}{6} (\bibinfo{year}{2013}), \bibinfo{pages}{481--492}.
\newblock


\end{thebibliography}
\balance
\appendix
%!TEX root=../ssppr_kdd17_newformat.tex

\section*{Appendix} \label{sec:app}

\section{Proof of Theorem}\label{app:proofs}

{\bf Proof of Theorem \ref{thm:bounds}.} Given $\omega_j$, we can derive that:
$$\Pr[|\pi(s,t) - \epi(s,t)|\ge \eps\cdot \pi(s,t)]\le 2\cdot \exp\left(-\frac{\eps^2\cdot \omega_j\cdot \pi(s,t)}{2 + 2a\cdot\eps/3}\right).$$
Since $a\le1$ and $\pi(s,t)\ge \rese(s,t)$, and $\pi(s,t) \ge LB_{j-1}(t)$, we can derive that:
\begin{align*}
\Pr[&|\pi(s,t) - \epi(s,t)|\ge \eps\cdot \pi(s,t)]\le \\
&2 \exp\left(-\frac{\eps^2\cdot \omega_j\cdot \max\{ \rese_j(s,v), LB_{j-1}(s,v) \}}{2 + 2\eps/3}\right).
\end{align*}
Let $\pf'$ equal the RHS of the above inequality.
%$$2\exp\left(-\frac{\eps^2\cdot \omega_j \cdot \max\{ \rese_j(s,v), LB_{j-1}(s,v) \}}{2 + 2\eps/3}\right) = \pf'.$$
We have $\eps \ge \sqrt{\frac{3\log{(2/\pf')}}{\omega_j \cdot \max\{ \rese_j(s,v), LB_{j-1}(s,v) \}}}.$

By setting $\eps_j = \sqrt{\frac{3\log{(2/\pf')}}{\omega_j \cdot \max\{ \rese_j(s,v), LB_{j-1}(s,v) \}}}$, we can derive that
$ \epi_j(s,v)/(1+\eps_j)\le \pi(s,v) \le \epi_j(s,v)/(1-\eps_j)$  holds with $1-\pf'$ probability.
Similarly, we have
$$
\Pr[|\pi(s,t) - \epi(s,t)|\ge \eps\cdot \pi(s,t)]\le 2 \exp\left(-\frac{\lambda^2\cdot \omega}{\rsum \cdot (2\pi(s,t) + 2 \lambda/3)}\right)
$$
Let $\pf'$ equal the RHS of the above inequality, and note that $\pi(s,t)\le UB_{j-1}(t)$. This helps us derive the designed bound for $\lambda_j$. \done

\header
{\bf Proof of Theorem \ref{thm:layer-topk}.}
We apply a similar technique in \cite{WangTXYL16}. If $LB_j(v'_i)\cdot(1+\eps)> UB_j(v'_i)$. Then, it can be derived that
\begin{align*}
\epi(v'_i)& \le UB_j(v'_i) \le LB_j(v'_i) \cdot(1+\eps) \le (1+\eps)\cdot \pi(s,v'_i).\\
\epi(v'_i)& \ge LB_j(v'_i)\ge UB_j(v'_i)/(1+\eps) \ge (1-\eps)\cdot \pi(s,v'_i).
\end{align*}

Hence, $v'_1,\cdots, v'_k$ satisfy Equation \ref{eqn:topk-self}. Let $v_1,\cdots, v_k$ be the $k$ nodes that have the top-$k$ exact PPR values. Assume that all bounds are correct, then as $LB(v'_l)\ge \delta$, it indicates that the top-$k$ PPR values are no smaller than $\delta$. In this case, all the top-$k$ nodes should satisfy $\epsilon$-approximation guarantee, i.e., they satisfy that $UB(v_i)<(1+\epsilon)\cdot LB(v_i)/(1-\epsilon)$.

Let $UB'_j(1), UB'_j(2),\cdots UB'_j(k)$ be the top-$k$ largest PPR upper bounds in the $j$-th iteration. Note that, the $i$-th largest PPR satisfy that $UB'_j(i) \ge \pi(s, v_i) \ge LB_j(s,v_i')$.

Now assume that one of the upper bounds say $UB'_j(i)$ is not from $UB( v_1),\cdots, UB( v_k)$. If it satisfies that $LB_j(v'_k)\cdot(1+\eps) \ge UB'_j(i)$, it indicates that $LB_j(v'_i) \cdot(1+\eps) \ge UB'_j(i)$ for all nodes. Hence, we update the node whose upper bound is minimum among $UB_j( v'_1),\cdots, UB_j( v'_1)$, we can still guarantee that $LB_j( v'_i)\cdot(1+\eps)> UB_j( v'_i)$. We repeat this process until $UB_j( v'_1),\cdots, UB_j(v'_k)$ are the top-$k$ upper bounds. On the other hand, let $U$ be the set of nodes such that, the node $u\in U$ satisfies that $UB_j(u) < LB_j(v'_k)$. Then it is still possible that these nodes are from the exact top$-$k answers. However, recall that if a node $u\in U$ is from the top-$k$, it should satisfy that $UB(u)<(1+\epsilon)\cdot LB(u)/(1-\epsilon)$. As a result, if there exists no node $u\in U$ such that $UB(u)<(1+\epsilon)\cdot LB(u)/(1-\epsilon)$, then no node $u\in U$ is from the top-$k$ answers, in which case it will not affect the approximation guarantee.

Afterwards, we proceed a bubble sort on the top-$k$ upper bounds in decreasing order. If we replace two upper bounds $UB_j(v'_x)$ and $UB_j(v'_y)$ with $x < y$, then $UB_j(v'_x)<UB_j(v'_y)$.
Also $LB_j(v'_x) > LB_j(v'_y)$ from the definition.
As
\begin{align*}
UB_j(v'_x)/LB_j(v'_y)<UB_j(v'_y)/LB_j(v'_y) \le (1+\eps) \\
UB_j(v'_y)/LB_j(v'_x)< UB_j(v'_y)/LB_j(v'_y) \le (1+\eps)
\end{align*}
When the sort finishes, the inequations still hold. We then have
$$UB'_j(1)\le (1+\eps)\cdot LB_j( v'_i)\cdots, UB'_j(k)\le (1+\eps)\cdot LB_j( v'_k).$$
Also note that
$$
\epi(v'_i) \ge LB_j(v'_i)\ge \frac{1}{1+\eps}UB'_j(i) \ge (1-\eps)\cdot \pi(s,v_i),
$$
for all $i\in [1,k]$. So, the answer provides approximation guarantee if the bounds from the first to the $j$-th iteration are all correct. By applying the union bound, we can obtain that the approximation is guaranteed with probability at least  $1- n\cdot j \cdot\pf'$.

\end{sloppy}

\end{document}